\documentclass[11pt]{article}
\usepackage{longtable}
\usepackage{caption}
\usepackage{multirow,multicol}
\usepackage{epsf, graphicx}
\usepackage{latexsym,amsfonts,amsbsy,amssymb}
\usepackage{amsmath,amsthm}
\usepackage{cleveref}
\usepackage{indentfirst}
\usepackage{bm}
\usepackage [latin1]{inputenc}
\allowdisplaybreaks[4]
\makeatletter

\@addtoreset{equation}{section} \makeatother
\newtheorem{Theorem}{Theorem}[section]
\newtheorem{Lemma}{Lemma}[section]
\newtheorem{Corollary}{Corollary}[section]
\newtheorem{Remark}{Remark}[section]

\newtheorem{Example}{Example}[section]
\makeatletter \@addtoreset{equation}{section} \makeatother
\begin{document}
	
	\title{Improved upper bounds on the number of non-zero weights of cyclic codes}
	\author{Bocong Chen$^1$, Yuqing Fu$^2$ and Hongwei Liu$^2$\footnote{E-mail addresses: bocongchen@foxmail.com (B. Chen),
			yuqingfu@mails.ccnu.edu.cn (Y. Fu),
			hwliu@ccnu.edu.cn (H. Liu)}}
	\date{\small
		$1.$ School of Mathematics, South China University of Technology, Guangzhou, 510641, China\\
		$2.$ School of Mathematics and Statistics,
		Central China Normal University,
		Wuhan,  430079, China\\
	}
	\maketitle
	{\noindent\small{\bf Abstract:} Let $\mathcal{C}$ be an arbitrary simple-root cyclic code  and let $\mathcal{G}$ be the subgroup of ${\rm Aut}(\mathcal{C})$ (the automorphism group of $\mathcal{C}$)  generated by the  multiplier, the cyclic shift   and the scalar multiplications.
		To the best of our knowledge,
		the subgroup $\mathcal{G}$ is the largest subgroup of ${\rm Aut}(\mathcal{C})$.  In this paper, an explicit formula, in some cases an upper bound, for the number of orbits of $\mathcal{G}$ on $\mathcal{C}\backslash \{\bf 0\}$ is established. An explicit upper bound on the number of non-zero weights of $\mathcal{C}$ is consequently derived and a necessary and sufficient condition for the code $\mathcal{C}$ meeting the bound is exhibited. Many examples are presented to show that our new upper bounds are tight and are strictly less than the upper bounds in [Chen and Zhang, IEEE-TIT, 2023]. In addition, for two special classes of cyclic codes, smaller upper bounds on the number of non-zero weights of such codes are obtained by replacing $\mathcal{G}$ with larger subgroups of the automorphism groups of these codes.
		As a byproduct, our main results  suggest a new way to find few-weight cyclic codes.}
	
	\vspace{1ex}
	
	
	\section{Introduction}
	The class of cyclic codes
	enjoys a prominent place in the theory of error-correcting
	codes. In fact, the class of cyclic codes is one
	of the most significant and well studied of all codes.
	Let $\mathbb{F}_{q}$ be the finite field with $q$ elements, where $q$ is a prime power. An $[n,k,d]$ {\em linear code}
	$\mathcal{C}$ over $\mathbb{F}_{q}$ is defined as a $k$-dimensional subspace of $\mathbb{F}_{q}^{n}$ with minimum (Hamming) distance $d$.
	A linear code $\mathcal{C}$ is said to be {\em cyclic} if the codewords of $\mathcal{C}$ are invariant
	under the {\em cyclic shift}: every codeword $\mathbf{c}=(c_0, c_1,\cdots,c_{n-1})\in \mathcal{C}$
	implies that $(c_{n-1},c_0,c_1,\cdots, c_{n-2})\in \mathcal{C}$.
	One of the principal reasons why cyclic codes play a significant  role in coding theory is that they have rich algebraic structures.
	Indeed,
	identifying each vector $(a_{0},a_{1},\cdots,a_{n-1})\in \mathbb{F}_{q}^{n}$ with a polynomial $$a_{0}+a_{1}x+\cdots+a_{n-1}x^{n-1}\in \mathcal{R}_{n}:=\mathbb{F}_{q}[x]/\langle x^{n}-1\rangle,$$ a linear code $\mathcal{C}$ of length $n$ over $\mathbb{F}_{q}$ corresponds to an $\mathbb{F}_{q}$-subspace of the algebra $\mathcal{R}_{n}$. Moreover, $\mathcal{C}$ is cyclic if and only if the corresponding subspace is an ideal of $\mathcal{R}_{n}$. A cyclic code $\mathcal{C}$ is called an {\em irreducible cyclic code} if $\mathcal{C}$ is a minimal ideal of $\mathcal{R}_{n}$.
	
	It is well-known that the distance of a code is a crucial parameter in the sense that it determines  the error-detecting and error-correcting capabilities of the code.
	Nevertheless, Delsarte \cite{8} showed that the number of non-zero weights of a code is a fundamental parameter in the sense that
	it is intimately related to other parameters of a code. Delsarte \cite{8} obtained an upper bound (resp. a lower bound) on the code size of a linear code $\mathcal{C}$ by using the number of non-zero weights of $\mathcal{C}$ (resp. $\mathcal{C}^{\perp}$),
	where $\mathcal{C}^{\perp}$ denotes the dual code of $\mathcal{C}$. It was also shown in \cite{8} that the covering radius of a code is at most equal to the number of non-zero weights of its dual code.
	For a linear code $\mathcal{C}$, the minimum distance of $\mathcal{C}$, the number of non-zero weights of $\mathcal{C}$, the minimum distance of $\mathcal{C}^{\perp}$ and the number of non-zero weights of $\mathcal{C}^{\perp}$ were titled  \textquotedblleft Four Fundamental Parameters of a Code" in \cite{8}.  Assmus and Mattson \cite{3} derived a relationship between codes and designs in terms of the number of non-zero weights of a linear code. The number of non-zero weights of a code has close connection with orthogonal arrays and combinatorial designs (see \cite{8}, \cite{9}).
	
	In general, it is very hard to determine the number of non-zero weights of a linear code. A more modest goal is to find acceptable bounds on the number of non-zero weights of a linear code.
	Indeed, there are some recent works in the literature about lower and upper bounds on the number of non-zero weights of linear codes. Alderson \cite{1} gave necessary and sufficient conditions for the existence of full weight spectrum codes. Shi et al. \cite{19} found that for a linear code of dimension $k$ over the finite field $\mathbb{F}_{q}$, the number of non-zero weights of this code is bounded from above by $(q^{k}-1)/(q-1)$, which is sharp for binary codes and for all codes of dimension two, and this bound was proved in \cite{2} to be sharp for all $q$ and $k$. Shi et al. \cite{17} derived lower and upper bounds on the number of non-zero weights of cyclic codes, and gave sharper upper bounds for irreducible cyclic codes and strongly cyclic codes. Shi et al. \cite{18} obtained several lower and upper bounds on the number of non-zero weights of quasi-cyclic codes.
	
	Chen and Zhang \cite{7} recently established a tight upper bound on the number of non-zero weights of a simple-root cyclic code $\mathcal{C}$ of length $n$ over $\mathbb{F}_{q}$. They used the observation that the number of non-zero weights of $\mathcal{C}$ is bounded from above by the number of orbits of the automorphism group of $\mathcal{C}$ (or a subgroup of the automorphism group of $\mathcal{C}$) acting on $\mathcal{C}\backslash \{\bf 0\}$, with equality if and only if any two codewords of $\mathcal{C}$ with the same weight belong to the same orbit. Let  $\mathcal{G}'$ be the subgroup of ${\rm Aut}(\mathcal{C})$ (the automorphism group of $\mathcal{C}$) generated by the cyclic shift and the scalar multiplications. Chen and Zhang \cite{7} obtained an explicit upper bound on the number of non-zero weights of $\mathcal{C}$ by calculating the number of orbits of $\mathcal{G}'$ on $\mathcal{C}\backslash \{\bf 0\}$.
	
	Motivated by the work \cite{7}, in this paper we choose larger subgroups of ${\rm Aut}(\mathcal{C})$ to replace $\mathcal{G}'$ and then we obtain improved upper bounds on the number of non-zero weights of $\mathcal{C}$.
	Apart from the cyclic shift and the scalar multiplications, we first note that the multiplier $\mu_{q}$ defined on $\mathcal{R}_{n}$ by
	$$\mu_{q}\big(\sum_{i=0}^{n-1}a_{i}x^{i}\big)=\sum_{i=0}^{n-1}a_{i}x^{qi}~({\rm mod}~x^n-1)$$
	is also an automorphism of $\mathcal{C}$. Let $\mathcal{G}$ be the subgroup of ${\rm Aut}(\mathcal{C})$ generated by these three types of automorphisms. To the best of our knowledge,
	$\mathcal{G}$ is the largest subgroup of ${\rm Aut}(\mathcal{C})$ for any simple-root cyclic code $\mathcal{C}$ over $\mathbb{F}_{q}$. Clearly, $\mathcal{G}'$ is a subgroup of $\mathcal{G}$, and the number of orbits of
	$\mathcal{G}$ on $\mathcal{C}\backslash\{\mathbf{0}\}$ is less than the number of orbits of  $\mathcal{G}'$ on $\mathcal{C}\backslash\{\mathbf{0}\}$.
	Therefore, we need to find the number of orbits of $\mathcal{G}$ on $\mathcal{C}\backslash\{\mathbf{0}\}$, which in turn
	gives improved upper bounds on the number of non-zero weights of $\mathcal{C}$.
	However, intuitively, the structure of $\mathcal{G}'$ is more complicate than that of $\mathcal{G}$, which implies that the problem of finding 
	the number of orbits of $\mathcal{G}$ on $\mathcal{C}\backslash\{\mathbf{0}\}$ may be more difficult than that of 
	finding  the number of orbits of  $\mathcal{G}'$ on $\mathcal{C}\backslash\{\mathbf{0}\}$.
	
	In this paper,
	an explicit formula, in some cases an upper bound, for the number of orbits of $\mathcal{G}$ on $\mathcal{C}\backslash \{\bf 0\}$ is established. An explicit upper bound on the number of non-zero weights of $\mathcal{C}$ is consequently derived and a necessary and sufficient condition for the code $\mathcal{C}$ meeting the bound is exhibited. Many examples are presented to show that our upper bounds are tight and are strictly less than the upper bounds in \cite{7}. In addition, for two special classes of cyclic codes, we replace $\mathcal{G}$ with lager subgroups of the automorphism groups of these codes, and then we obtain smaller upper bounds on the number of non-zero weights of these codes. The upper bounds presented in this paper improve the upper bounds in \cite{7}. As a byproduct, our main results  suggest a new way to find few-weight cyclic codes,
	which have found wide applications in cryptography, association schemes and network coding (for example, see \cite{Ding2009}).
	
	This paper is organized as follows. In Section 2, we review some definitions and basic results about group action, cyclic codes and subgroups of the automorphism group of $\mathcal{R}_{n}$. In Section 3, we derive the main results of this paper. This is divided into four subsections: In subsections 3.1 and 3.2, we present improved upper bounds on the number of non-zero weights of irreducible cyclic codes and general cyclic codes, respectively, by calculating the number of $\mathcal{G}$-orbits. In subsections 3.3 and 3.4, for two special classes of cyclic codes, we give smaller upper bounds on the number of non-zero weights of such codes by replacing $\mathcal{G}$ with lager subgroups of the automorphism groups of these codes. In Section 4, we conclude this paper with remarks and some possible future works.
	
	\section{Preliminaries}
	
	Throughout this paper, $\mathbb{F}_{q}$ denotes the finite field with $q=p^e$ elements, where $p$ is a prime, $e$ is a positive integer and $n$ is a positive integer coprime to $q$. By $\mathbb{F}_{q}^{*}$ we denote the multiplicative group of $\mathbb{F}_{q}$, and for $\alpha\in \mathbb{F}_{q}^{*}$, ${\rm ord}(\alpha)$ denotes the order of $\alpha$ in $\mathbb{F}_{q}^{*}$. Let $\mathbb{Z}_{n}$ be the residue ring of the integer ring $\mathbb{Z}$ modulo $n$, and let $\mathbb{Z}_{n}^{*}$ be the group of units in $\mathbb{Z}_{n}$. For $a\in \mathbb{Z}_{n}^{*}$, $a^{-1}$ denotes the multiplicative inverse of $a$ in $\mathbb{Z}_{n}^{*}$. As usual, $|X|$ stands for the cardinality of a finite set $X$. For integers $b_{1},b_{2},\cdots,b_{r}$, where $r\geq 2$ is a positive integer, ${\rm gcd}(b_{1},b_{2},\cdots,b_{r})$ denotes the greatest common divisor of $b_{1},b_{2},\cdots,b_{r}$. Given two integers $b_{1}$ and $b_{2}$, if $b_{1}$ divides $b_{2}$, then we write $b_{1}|b_{2}$. For a positive integer $b$, $\varphi(b)$ is the Euler's function of $b$, which is the number of positive integers not exceeding $b$ and prime to $b$. 
	
	In this section, we present some necessary background concerning group action on a linear code, some notions and results about simple-root cyclic codes, and subgroups of the automorphism group of $\mathcal{R}_{n}$.
	
	\subsection{Group action on a linear code}
	
	Suppose that a finite group $G$ acts on a finite set $X$. For each $x\in X$, $Gx=\{gx~|~g\in G\}$ is called an orbit of this group action containing $x$ (or simply a $G$-orbit). All the $G$-orbits partition $X$, that is, $X$ is the disjoint union of the $G$-orbits. For convenience, the set of all the orbits of $G$ on $X$ is denoted as
	$$G\backslash X=\{Gx~|~x\in X\}.$$
	
	Let $\mathcal{C}$ be a linear code and let $\mathcal{G}$ be a subgroup ${\rm Aut}(\mathcal{C})$. The next lemma reveals that the number of non-zero weights of $\mathcal{C}$ is bounded from above by the number of orbits of $\mathcal{G}$ on $\mathcal{C}^{*}=\mathcal{C}\backslash \{\bf 0\}$, with equality if and only if any two codewords of $\mathcal{C}^{*}$ with the same weight are in the same $\mathcal{G}$-orbit.
	
	\begin{Lemma}\cite{7}\label{l2.1}
		Let $\mathcal{C}$ be a linear code of length $n$ over $\mathbb{F}_{q}$ with $\ell$ non-zero weights and let ${\rm Aut}(\mathcal{C})$ be the automorphism group of $\mathcal{C}$. Suppose that $\mathcal{G}$ is a subgroup of ${\rm Aut}(\mathcal{C})$. If the number of orbits of $\mathcal{G}$ on $\mathcal{C}^{*}=\mathcal{C}\backslash \{\bf 0\}$ is equal to $N$, then $\ell \leq N$. Moreover, the equality holds if and only if for any two non-zero codewords ${\bf c}_{1},{\bf c}_{2}\in \mathcal{C}^{*}$ with the same weight, there exists an automorphism $A\in \mathcal{G}$ such that $A{\bf c}_{1}={\bf c}_{2}$.
	\end{Lemma}
	
	By Burnside's lemma (see \cite[Theorem 2.113]{16}), the number of orbits of $\mathcal{G}$ on $\mathcal{C}^{*}$ is equal to
	\begin{equation}\label{e2.1}
		|\mathcal{G} \backslash \mathcal{C}^{*}|=\frac{1}{|\mathcal{G}|}\sum_{g\in \mathcal{G}}|{\rm Fix}(g)|,
	\end{equation}
	where ${\rm Fix}(g)=\{{\bf c}\in \mathcal{C}~|~g{\bf c}={\bf c}\}$.
	
	Suppose $\mathcal{G}'$ is a subgroup of $\mathcal{G}$. For ${\bf c}\in \mathcal{C}^{*}$, it is easy to see that $|\mathcal{G}'{\bf c}|\leq |\mathcal{G}{\bf c}|$. Since $\mathcal{C}^{*}=\bigcup\limits_{{\bf c}\in \mathcal{C}^{*}}\mathcal{G}{\bf c}=\bigcup\limits_{{\bf c}\in \mathcal{C}^{*}}\mathcal{G}'{\bf c}$, we have $\big|\{\mathcal{G}{\bf c}~|~{\bf c}\in \mathcal{C}^{*}\}\big|\leq \big|\{\mathcal{G}'{\bf c}~|~{\bf c}\in \mathcal{C}^{*}\}\big|$, that is, the number of orbits of $\mathcal{G}$ on $\mathcal{C}^{*}$ is less than or equal to the number of orbits of $\mathcal{G}'$ on $\mathcal{C}^{*}$.
	
	\subsection{Cyclic codes and primitive idempotents}
	
	It is well known that every irreducible cyclic code of length $n$ over $\mathbb{F}_{q}$ is generated uniquely by a primitive idempotent of $\mathcal{R}_{n}$, and every cyclic code of length $n$ over $\mathbb{F}_{q}$ is a direct sum of some irreducible cyclic codes (see, for example, \cite[Theorem 4.3.8]{11}). There is a one-to-one correspondence between the primitive idempotents of $\mathcal{R}_{n}$ and the $q$-cyclotomic cosets modulo $n$. 
	
	Let $m$ be the order of $q$ in $\mathbb{Z}_{n}^{*}$, that is, $m$ is the least positive integer such that $n$ divides $q^m-1$. Assume that all the distinct $q$-cyclotomic cosets modulo $n$ are given by
	\begin{align*}
		\Gamma_{0}&=\{i_{0}=0\},\\
		\Gamma_{1}&=\{i_{1}=1,q,\cdots,q^{k_{1}-1}=q^{m-1}\},\\
		\Gamma_{2}&=\{i_{2},i_{2}q,i_{2}q^{2},\cdots,i_{2}q^{k_{2}-1}\},\\
		\vdots\\
		\Gamma_{s}&=\{i_{s},i_{s}q,i_{s}q^{2},\cdots,i_{s}q^{k_{s}-1}\},
	\end{align*}
	where the elements in the braces above are calculated modulo $n$ and $k_{i}$ is the cardinality of the $q$-cyclotomic coset $\Gamma_{i}$ for $0\leq i\leq s$ with $k_{0}=1$ and $k_{1}=m$.  Suppose $\zeta$ is a primitive $n$-th root of unity in $\mathbb{F}_{q^m}$. The quotient ring $\mathbb{F}_{q^m}[x]/\langle x^{n}-1\rangle$ has exactly $n$ primitive idempotents given by (see, for example \cite{6})
	$$e_{i}=\frac{1}{n}\sum_{j=0}^{n-1}\zeta^{-ij}x^{j}~~~{\rm for}~0\leq i\leq n-1.$$
	Moreover, $\mathcal{R}_{n}=\mathbb{F}_{q}[x]/\langle x^{n}-1\rangle$ has exactly $s+1$ primitive idempotents given by (see, for example \cite{6})
	$$\varepsilon_{t}=\sum_{j\in \Gamma_{t}}e_{j}~~~{\rm for}~0\leq t\leq s,$$
	and \cite[Theorem 4.3.8]{11} shows that $\mathcal{R}_{n}$ is the vector space direct sum of the minimal ideals $\mathcal{R}_{n}\varepsilon_{t}$ for $0\leq t\leq s$, in symbols,
	$$\mathcal{R}_{n}=\mathcal{R}_{n}\varepsilon_{0}\bigoplus \mathcal{R}_{n}\varepsilon_{1}\bigoplus \cdots \bigoplus \mathcal{R}_{n}\varepsilon_{s}.$$
	Using the Discrete Fourier Transform (or, using the results in \cite[Section V]{14} directly), we have, for each $0\leq t\leq s$,
	$$\mathcal{R}_{n}\varepsilon_{t}=\Big\{\sum_{j=0}^{k_{t}-1}\big(\sum_{\ell=0}^{k_{t}-1}c_{\ell}\zeta^{\ell i_{t}q^{j}}\big)e_{i_{t}q^{j}}~\Big|~c_{\ell}\in \mathbb{F}_{q}, 0\leq \ell \leq k_{t}-1\Big\}.$$

	\subsection{Subgroups of the automorphism group of $\mathcal{R}_{n}$}
	Suppose that the cyclic group $\mathbb{F}_{q}^{*}$ is generated by $\xi$. We have the following two $\mathbb{F}_{q}$-vector space automorphisms of $\mathcal{R}_{n}=\mathbb{F}_{q}[x]/\langle x^{n}-1\rangle$, denoted by $\rho$ and $\sigma_{\xi}$, respectively:
	$$\rho:~\mathcal{R}_{n}\rightarrow \mathcal{R}_{n},~~\rho\big(\sum_{i=0}^{n-1}f_{i}x^{i}\big)=\sum_{i=0}^{n-1}f_{i}x^{i+1}~({\rm mod}~x^n-1),$$
	and
	$$\sigma_{\xi}:~\mathcal{R}_{n}\rightarrow \mathcal{R}_{n},~~\sigma_{\xi}\big(\sum_{i=0}^{n-1}f_{i}x^{i}\big)=\sum_{i=0}^{n-1}\xi f_{i}x^{i}.$$
	Clearly, the subgroup $\langle \rho \rangle$ of ${\rm Aut}(\mathcal{R}_{n})$ (the $\mathbb{F}_{q}$-vector space automorphism group of $\mathcal{R}_{n}$) generated by $\rho$ is of order $n$, and the subgroup $\langle\sigma_{\xi} \rangle$ of ${\rm Aut}(\mathcal{R}_{n})$ generated by $\sigma_{\xi}$ is of order $q-1$. For any cyclic code $\mathcal{C}$ of length $n$ over $\mathbb{F}_{q}$, it is readily seen that $\langle \rho,\sigma_{\xi}\rangle$ is a subgroup of ${\rm Aut}(\mathcal{C})$.
	
	Let $a\in \mathbb{Z}_{n}^{*}$. The multiplier $\mu_{a}$ defined on $\mathcal{R}_{n}$ by
	$$\mu_{a}:~\mathcal{R}_{n}\rightarrow \mathcal{R}_{n},~~\mu_{a}\big(\sum_{i=0}^{n-1}f_{i}x^{i}\big)=\sum_{i=0}^{n-1}f_{i}x^{ai}~({\rm mod}~x^n-1)$$
	is a ring automorphism of $\mathcal{R}_{n}$ (see \cite[Theorem 4.3.12]{11}). Let ${\rm ord}_{n}(a)$ denote the order of $a$ in $\mathbb{Z}_{n}^{*}$. It is easy to check that the subgroup $\langle \mu_{a}\rangle$ of ${\rm Aut}(\mathcal{R}_{n})$ generated by $\mu_{a}$ is of order ${\rm ord}_{n}(a)$.
	
	\begin{Lemma}\label{l2.2}
		Suppose $a\in \mathbb{Z}_{n}^{*}$. The subgroup $\langle \mu_{a},\rho,\sigma_{\xi}\rangle$ of ${\rm Aut}(\mathcal{R}_{n})$ is of order ${\rm ord}_{n}(a)n(q-1)$, and each element of $\langle \mu_{a},\rho,\sigma_{\xi} \rangle$ can be written uniquely as a product $\mu_{a}^{r_{1}}\rho^{r_{2}}\sigma_{\xi}^{r_{3}}$ for some $0\leq r_{1}\leq {\rm ord}_{n}(a)-1$, $0\leq r_{2}\leq n-1$ and $0\leq r_{3}\leq q-2$.
	\end{Lemma}
	
	\begin{proof}
		We first note that $\mu_{a}\rho=\rho^{a}\mu_{a}$ and $\rho\mu_{a}=\mu_{a}\rho^{a^{-1}}$, then $\langle \mu_{a} \rangle\langle\rho \rangle=\langle\rho \rangle\langle \mu_{a} \rangle$, which implies that $\langle \mu_{a} \rangle\langle\rho \rangle$ is a subgroup of ${\rm Aut}(\mathcal{R}_{n})$. It is easy to verify that $\langle \mu_{a} \rangle\langle\rho \rangle \langle\sigma_{\xi} \rangle=\langle\sigma_{\xi} \rangle\langle \mu_{a} \rangle\langle\rho \rangle $, and so $\langle \mu_{a} \rangle\langle\rho \rangle \langle\sigma_{\xi} \rangle$ is a subgroup of ${\rm Aut}(\mathcal{R}_{n})$, thus $\langle \mu_{a},\rho,\sigma_{\xi} \rangle=\langle \mu_{a} \rangle\langle\rho \rangle \langle\sigma_{\xi} \rangle$. 
		
		Suppose $\alpha\in \langle \mu_{a} \rangle \cap \langle\rho \rangle$, then $\alpha=\mu_{a}^{r_{1}}=\rho^{r_{2}}$ for some $0\leq r_{1}\leq {\rm ord}_{n}(a)-1$ and $0\leq r_{2}\leq n-1$. Let $f(x)=1\in \mathcal{R}_{n}$. As $1=\mu_{a}^{r_{1}}(f(x))=\rho^{r_{2}}(f(x))=x^{r_{2}}$, $r_{2}=0$ implying that $\langle \mu_{a} \rangle \cap \langle\rho \rangle=id$, where $id$ is the identity element of ${\rm Aut}(\mathcal{R}_{n})$. Suppose $\alpha\in \langle \mu_{a} \rangle \langle\rho \rangle\cap \langle\sigma_{\xi} \rangle$, then $\alpha=\mu_{a}^{r_{1}}\rho^{r_{2}}=\sigma_{\xi}^{r_{3}}$ for some $0\leq r_{1}\leq {\rm ord}_{n}(a)-1$, $0\leq r_{2}\leq n-1$ and $0\leq r_{3}\leq q-2$. Let $f(x)=1\in \mathcal{R}_{n}$, then $x^{a^{r_{1}}r_{2}}\equiv \mu_{a}^{r_{1}}\rho^{r_{2}}(f(x))\equiv \sigma_{\xi}^{r_{3}}(f(x))\equiv \xi^{r_{3}}~({\rm mod}~x^n-1)$, and hence $r_{3}=0$. So $\langle \mu_{a} \rangle \langle\rho \rangle\cap \langle\sigma_{\xi} \rangle=id$. Then the order of $\langle \mu_{a},\rho,\sigma_{\xi} \rangle$ is equal to
		$$\frac{|\langle \mu_{a}\rangle \langle\rho \rangle||\langle\sigma_{\xi} \rangle|}{|\langle \mu_{a}\rangle \langle\rho \rangle\cap \langle\sigma_{\xi} \rangle|}=|\langle \mu_{a},\rho \rangle||\langle\sigma_{\xi} \rangle|=\frac{|\langle \mu_{a} \rangle||\langle\rho \rangle||\langle\sigma_{\xi} \rangle|}{|\langle \mu_{a} \rangle\cap \langle\rho \rangle|}=|\langle \mu_{a} \rangle||\langle\rho \rangle||\langle\sigma_{\xi} \rangle|={\rm ord}_{n}(a)n(q-1).$$
		
		For any $\alpha \in \langle \mu_{a},\rho,\sigma_{\xi} \rangle$, $\alpha=\mu_{a}^{r_{1}}\rho^{r_{2}}\sigma_{\xi}^{r_{3}}$ for some $0\leq r_{1}\leq {\rm ord}_{n}(a)-1$, $0\leq r_{2}\leq n-1$ and $0\leq r_{3}\leq q-2$. If $\alpha$ can also be written as $\alpha=\mu_{a}^{r_{1}'}\rho^{r_{2}'}\sigma_{\xi}^{r_{3}'}$, where $0\leq r_{1}'\leq {\rm ord}_{n}(a)-1$, $0\leq r_{2}'\leq n-1$ and $0\leq r_{3}'\leq q-2$, then $\mu_{a}^{r_{1}-r_{1}'}\rho^{r_{2}-r_{2}'}=\sigma_{\xi}^{r_{3}'-r_{3}}$. As $\langle \mu_{a} \rangle \langle\rho \rangle\cap \langle\sigma_{\xi} \rangle=id$, we have $r_{3}'=r_{3}$ and $\mu_{a}^{r_{1}-r_{1}'}=\rho^{r_{2}'-r_{2}}$. But $\langle \mu_{a} \rangle \cap \langle\rho \rangle=id$, then $r_{1}'=r_{1}$ and $r_{2}'=r_{2}$. Thus the representation of $\alpha \in \langle \mu_{a},\rho,\sigma_{\xi} \rangle$ is unique.
	\end{proof}
	\section{Improved upper bounds}
	
	Let $\mathcal{C}$ be a cyclic code of length $n$ over $\mathbb{F}_{q}$ and let $\mathcal{G}$ be a subgroup of ${\rm Aut}(\mathcal{C})$. We know from Lemma \ref{l2.1} that the number of non-zero weights of $\mathcal{C}$ is bounded from above by the number of orbits of $\mathcal{G}$ on $\mathcal{C}^{*}=\mathcal{C}\backslash \{\bf 0\}$. \cite{7} choose $\mathcal{G}=\langle \rho,\sigma_{\xi} \rangle$. In this paper, we choose $\mathcal{G}$ to be a larger subgroup of ${\rm Aut}(\mathcal{C})$ which contains $\langle \rho,\sigma_{\xi} \rangle$ as a subgroup, and then find the number of orbits of $\mathcal{G}$ on $\mathcal{C}^{*}$.
	\subsection{An improved upper bound on the number of non-zero weights of an irreducible cyclic code}
	
	For a cyclic code $\mathcal{C}$ of length $n$ over $\mathbb{F}_{q}$, $\mu_{q}\in {\rm Aut}(\mathcal{C})$ (see \cite[Theorem 4.3.13]{11}), and so we can consider the action of $\langle \mu_{q},\rho,\sigma_{\xi} \rangle$ on $\mathcal{C}^{*}$. We first assume that $\mathcal{C}$ is an irreducible cyclic code and we have the following result.
	
	\begin{Theorem}\label{t3.1}
		Let $\mathcal{C}$ be an $[n,k]$ irreducible cyclic code over $\mathbb{F}_{q}$. Suppose that $\mathcal{C}$ is generated by $\varepsilon_{t}$, where the primitive idempotent $\varepsilon_{t}$ corresponds to the $q$-cyclotomic coset $\{i_{t},i_{t}q,\cdots, i_{t}q^{k-1}\}$. Then the number of orbits of $\langle \mu_{q},\rho,\sigma_{\xi} \rangle$ on $\mathcal{C}^{*}=\mathcal{C}\backslash \{\bf 0\}$ is equal to
		$$\frac{1}{k}\sum_{r\mid k}\varphi(\frac{k}{r}){\rm gcd}\big(q^{r}-1,\frac{q^{k}-1}{q-1},\frac{i_{t}(q^{k}-1)}{n}\big).$$
		In particular, the number of non-zero weights of $\mathcal{C}$ is less than or equal to the number of orbits of $\langle \mu_{q},\rho,\sigma_{\xi} \rangle$ on $\mathcal{C}^{*}$, with equality if and only if for any two codewords ${\bf c}_{1},{\bf c}_{2}\in \mathcal{C}^{*}$ with the same weight, there exist integers $j_{1}$, $j_{2}$ and $j_{3}$ such that $\mu_{q}^{j_{1}}\rho^{j_{2}}(\xi^{j_{3}}{\bf c}_{1})={\bf c}_{2}$.
	\end{Theorem}
	
	\begin{proof}
		It is enough to count the number of orbits of $\langle \mu_{q},\rho,\sigma_{\xi} \rangle$ on $\mathcal{C}^{*}$, since the rest of the statements are clear from Lemma \ref{l2.1}. It follows from Equation (\ref{e2.1}) and Lemma \ref{l2.2} that
		$$\big|\langle \mu_{q},\rho,\sigma_{\xi} \rangle\backslash \mathcal{C}^{*}\big|=\frac{1}{mn(q-1)}\sum_{r_{1}=0}^{m-1}\sum_{r_{2}=0}^{n-1}\sum_{r_{3}=0}^{q-2}\Big|\big\{{\bf c}\in \mathcal{C}^{*}~\big|~\mu_{q}^{r_{1}}\rho^{r_{2}}\sigma_{\xi}^{r_{3}}({\bf c})={\bf c}\big\}\Big|.$$
		
		Take a typical non-zero element
		$${\bf c}=\sum_{j=0}^{k-1}\big(c_{0}+c_{1}\zeta^{i_{t}q^{j}}+\cdots+c_{k-1}\zeta^{(k-1)i_{t}q^{j}}\big)e_{i_{t}q^{j}}\in \mathcal{C}^{*}.$$
		Note that $e_{i_{t}q^{j}}=\frac{1}{n}\sum\limits_{l=0}^{n-1}\zeta^{-i_{t}q^{j}l}x^{l}$ and $\rho^{r_{2}}\sigma_{\xi}^{r_{3}}(e_{i_{t}q^{j}})=\xi^{r_{3}}\zeta^{i_{t}q^{j}r_{2}}e_{i_{t}q^{j}}$ (see \cite[Lemma III.8]{7}), and thus
		\begin{align*}
			\mu_{q}^{r_{1}}\rho^{r_{2}}\sigma_{\xi}^{r_{3}}(e_{i_{t}q^{j}})&=\xi^{r_{3}}\zeta^{i_{t}q^{j}r_{2}}\mu_{q}^{r_{1}}(e_{i_{t}q^{j}})\\
			&=\xi^{r_{3}}\zeta^{i_{t}q^{j}r_{2}}\cdot \frac{1}{n}\sum_{l=0}^{n-1}\zeta^{-i_{t}q^{j}l}x^{q^{r_{1}}l}\\
			&=\xi^{r_{3}}\zeta^{i_{t}q^{j}r_{2}}\cdot \frac{1}{n}\sum_{l=0}^{n-1}\zeta^{-i_{t}q^{-r_{1}+j}q^{r_{1}}l}x^{q^{r_{1}}l}\\
			&=\xi^{r_{3}}\zeta^{i_{t}q^{j}r_{2}}\cdot \frac{1}{n}\sum_{l=0}^{n-1}\zeta^{-i_{t}q^{-r_{1}+j}l}x^{l}\\
			&=\xi^{r_{3}}\zeta^{i_{t}q^{j}r_{2}}e_{i_{t}q^{-r_{1}+j}},
		\end{align*}
		where the subscript $i_{t}q^{-r_{1}+j}$ is calculated modulo $n$. We then have
		\begin{align*}
			\mu_{q}^{r_{1}}\rho^{r_{2}}\sigma_{\xi}^{r_{3}}({\bf c})&=\mu_{q}^{r_{1}}\rho^{r_{2}}\sigma_{\xi}^{r_{3}}\Big(\sum_{j=0}^{k-1}\big(c_{0}+c_{1}\zeta^{i_{t}q^{j}}+\cdots+c_{k-1}\zeta^{(k-1)i_{t}q^{j}}\big)e_{i_{t}q^{j}}\Big)\\
			&=\sum_{j=0}^{k-1}\big(c_{0}+c_{1}\zeta^{i_{t}q^{j}}+\cdots+c_{k-1}\zeta^{(k-1)i_{t}q^{j}}\big)\mu_{q}^{r_{1}}\rho^{r_{2}}\sigma_{\xi}^{r_{3}}(e_{i_{t}q^{j}})\\
			&=\sum_{j=0}^{k-1}\xi^{r_{3}}\zeta^{i_{t}q^{j}r_{2}}\big(c_{0}+c_{1}\zeta^{i_{t}q^{j}}+\cdots+c_{k-1}\zeta^{(k-1)i_{t}q^{j}}\big)e_{i_{t}q^{-r_{1}+j}}\\
			&=\sum_{j=0}^{k-1}\xi^{r_{3}}\zeta^{i_{t}q^{-r_{1}+j} q^{r_{1}}r_{2}}\big(c_{0}+c_{1}\zeta^{i_{t}q^{-r_{1}+j}}+\cdots+c_{k-1}\zeta^{(k-1)i_{t}q^{-r_{1}+j}}\big)^{q^{r_{1}}}e_{i_{t}q^{-r_{1}+j}}\\
			&=\sum_{j=0}^{k-1}\xi^{r_{3}}\zeta^{i_{t}q^{r_{1}+j}r_{2}}\big(c_{0}+c_{1}\zeta^{i_{t}q^{j}}+\cdots+c_{k-1}\zeta^{(k-1)i_{t}q^{j}}\big)^{q^{r_{1}}}e_{i_{t}q^{j}}.
		\end{align*}
		Hence $\mu_{q}^{r_{1}}\rho^{r_{2}}\sigma_{\xi}^{r_{3}}({\bf c})={\bf c}$ if and only if
		$$\xi^{r_{3}}\big(c_{0}+c_{1}\zeta^{i_{t}q^{j}}+\cdots+c_{k-1}\zeta^{(k-1)i_{t}q^{j}}\big)^{q^{r_{1}}-1}=\zeta^{-i_{t}q^{r_{1}+j}r_{2}}~~{\rm for}~0\leq j\leq k-1,$$
		which is equivalent to
		$$\xi^{r_{3}}\big(c_{0}+c_{1}\zeta^{i_{t}}+\cdots+c_{k-1}\zeta^{(k-1)i_{t}}\big)^{q^{r_{1}}-1}=\zeta^{-i_{t}q^{r_{1}}r_{2}}.$$
		Since the minimal polynomial of $\zeta^{i_{t}}$ over $\mathbb{F}_{q}$ is of degree $k$, the set
		$$\big\{c_{0}+c_{1}\zeta^{i_{t}}+\cdots+c_{k-1}\zeta^{(k-1)i_{t}}~\big|~c_{\ell}\in \mathbb{F}_{q}, 0\leq \ell \leq k-1\big\}$$
		forms a subfield of $\mathbb{F}_{q^{m}}$ of size $q^{k}$. Therefore, the number of ${\bf c}\in \mathcal{C}^{*}$ satisfying $\mu_{q}^{r_{1}}\rho^{r_{2}}\sigma_{\xi}^{r_{3}}({\bf c})={\bf c}$ is equal to the number of $\alpha\in \mathbb{F}_{q^{k}}^{*}$ such that $\xi^{r_{3}}\alpha^{q^{r_{1}}-1}=\zeta^{-i_{t}q^{r_{1}}r_{2}}$.
		
		Write $\mathbb{F}_{q^{k}}^{*}=\langle \theta \rangle$, that is, $\mathbb{F}_{q^{k}}^{*}$ is generated by $\theta$. Suppose that $\xi^{r_{3}}\alpha^{q^{r_{1}}-1}=\zeta^{-i_{t}q^{r_{1}}r_{2}}$ for some $\alpha\in \mathbb{F}_{q^{k}}^{*}$. On the one hand, for any $\beta \in \langle \theta^{\frac{q^{k}-1}{q^{{\rm gcd}(k,r_{1})}-1}} \rangle$, we have $\xi^{r_{3}}(\alpha\beta)^{q^{r_{1}}-1}=\xi^{r_{3}}\alpha^{q^{r_{1}}-1}=\zeta^{-i_{t}q^{r_{1}}r_{2}}$.
		On the other hand, suppose $\gamma\in \mathbb{F}_{q^{k}}^{*}$ such that $\xi^{r_{3}}\gamma^{q^{r_{1}}-1}=\zeta^{-i_{t}q^{r_{1}}r_{2}}$, then $(\gamma\alpha^{-1})^{q^{r_{1}}-1}=1$, and so ${\rm ord}(\gamma\alpha^{-1})\big| {\rm gcd}(q^{k}-1,q^{r_{1}}-1)=q^{{\rm gcd}(k,r_{1})}-1$, which implies that $\gamma\alpha^{-1}\in \langle \theta^{\frac{q^{k}-1}{q^{{\rm gcd}(k,r_{1})}-1}} \rangle$, and hence $\gamma=\alpha\beta$ for some $\beta \in \langle \theta^{\frac{q^{k}-1}{q^{{\rm gcd}(k,r_{1})}-1}} \rangle$. Thus the number of $\alpha\in \mathbb{F}_{q^{k}}^{*}$ such that $\xi^{r_{3}}\alpha^{q^{r_{1}}-1}=\zeta^{-i_{t}q^{r_{1}}r_{2}}$ is equal to $0$ or $q^{{\rm gcd}(k,r_{1})}-1$.  Next we aim to find the number of pairs $(r_{2},r_{3})$ with $0\leq r_{2}\leq n-1$ and $0\leq r_{3}\leq q-2$ such that $\xi^{r_{3}}\alpha^{q^{r_{1}}-1}=\zeta^{-i_{t}q^{r_{1}}r_{2}}$ for some $\alpha\in \mathbb{F}_{q^{k}}^{*}$.

		It is easy to check that $\{\alpha^{q^{r_{1}}-1}~|~\alpha\in \mathbb{F}_{q^{k}}^{*}\}=\langle \theta^{q^{r_{1}}-1}\rangle$, which is a cyclic subgroup of $\mathbb{F}_{q^{k}}^{*}$ of order $\frac{q^{k}-1}{q^{{\rm gcd}(k,r_{1})}-1}$. Since $\langle \xi \rangle$ is a cyclic subgroup of $\mathbb{F}_{q^{k}}^{*}$ of order $q-1$, we see that $\langle \xi \rangle \cap \langle \theta^{q^{r_{1}}-1} \rangle$ is a cyclic subgroup of $\mathbb{F}_{q^{k}}^{*}$ of order ${\rm gcd}\big(q-1,\frac{q^{k}-1}{q^{{\rm gcd}(k,r_{1})}-1}\big)$, and that $\langle \xi \rangle \langle \theta^{q^{r_{1}}-1} \rangle$ is a cyclic subgroup of $\mathbb{F}_{q^{k}}^{*}$ of order $\frac{|\langle \xi \rangle||\langle \theta^{q^{r_{1}}-1} \rangle|}{|\langle \xi \rangle \cap \langle \theta^{q^{r_{1}}-1} \rangle|}$. As ${\rm ord}(\zeta^{-1})=n$, ${\rm ord}(\zeta^{-i_{t}q^{r_{1}}r_{2}})=\frac{n}{{\rm gcd}(n,i_{t}q^{r_{1}}r_{2})}=\frac{n}{{\rm gcd}(n,i_{t}r_{2})}$.
		Then we have
		\begin{align*}
			\zeta^{-i_{t}q^{r_{1}}r_{2}}\in \langle \xi \rangle\langle \theta^{q^{r_{1}}-1}\rangle~&\Leftrightarrow~\frac{n}{{\rm gcd}(n,i_{t}r_{2})}\Big| |\langle \xi \rangle\langle \theta^{q^{r_{1}}-1} \rangle|\\
			&\Leftrightarrow~\frac{n}{{\rm gcd}(n,i_{t}r_{2}){\rm gcd}\big(\frac{n}{{\rm gcd}(n,i_{t}r_{2})},|\langle \xi \rangle\langle \theta^{q^{r_{1}}-1} \rangle|\big)}=1\\
			&\Leftrightarrow~\frac{n}{{\rm gcd}\big(n,i_{t}r_{2}|\langle \xi \rangle\langle \theta^{q^{r_{1}}-1} \rangle|\big)}=1\\
			&\Leftrightarrow~n\big|\big(i_{t}r_{2}|\langle \xi \rangle\langle \theta^{q^{r_{1}}-1} \rangle|\big)\\
			&\Leftrightarrow~\frac{n}{{\rm gcd}\big(n,i_{t}|\langle \xi \rangle\langle \theta^{q^{r_{1}}-1} \rangle|\big)}\Big|r_{2}.
		\end{align*}
		For $0\leq r_{1}\leq m-1$, denote
		$$S(r_{1})=\big\{0\leq r_{2}\leq n-1~\big|~ \zeta^{-i_{t}q^{r_{1}}r_{2}}\in \langle \xi \rangle\langle \theta^{q^{r_{1}}-1}\rangle\big\},$$
		and then
		$$|S(r_{1})|={\rm gcd}\big(n,i_{t}|\langle \xi \rangle \langle \theta^{q^{r_{1}}-1} \rangle|\big).$$
		
		Assume that $r_{2}\in S(r_{1})$. Then $\zeta^{-i_{t}q^{r_{1}}r_{2}}\in \xi^{r_{3}}\langle \theta^{q^{r_{1}}-1}\rangle$ for some $0\leq r_{3}\leq q-2$. Denote
		$$R(r_{1},r_{2})=\big\{0\leq r\leq q-2~\big|~\zeta^{-i_{t}q^{r_{1}}r_{2}}\in \xi^{r}\langle \theta^{q^{r_{1}}-1}\rangle\big\}.$$
		On the one hand, for any $\xi^{r'}\in \langle \xi \rangle \cap \langle \theta^{q^{r_{1}}-1} \rangle$, we have $\zeta^{-i_{t}q^{r_{1}}r_{2}}\in \xi^{r_{3}}\langle \theta^{q^{r_{1}}-1}\rangle=\xi^{r_{3}+r'}\langle \theta^{q^{r_{1}}-1}\rangle$, and so $r_{3}+r'~({\rm mod}~q-1)\in R(r_{1},r_{2})$. On the other hand, suppose $r\in R(r_{1},r_{2})$, that is, $\zeta^{-i_{t}q^{r_{1}}r_{2}}\in \xi^{r}\langle \theta^{q^{r_{1}}-1}\rangle$, then $\xi^{r-r_{3}}\in \langle \xi \rangle \cap \langle \theta^{q^{r_{1}}-1}\rangle$, and so $r=r_{3}+r'~({\rm mod}~q-1)$, where $r'$ is an integer such that $\xi^{r-r_{3}}=\xi^{r'}\in \langle \xi \rangle \cap \langle \theta^{q^{r_{1}}-1}\rangle$. Hence
		$$|R(r_{1},r_{2})|=|\langle \xi \rangle \cap \langle \theta^{q^{r_{1}}-1}\rangle|.$$

		We then conclude that
		\begin{align*}
			&\big|\langle \mu_{q},\rho,\sigma_{\xi} \rangle \backslash \mathcal{C}^{*}\big|\\
			=&\frac{1}{mn(q-1)}\sum_{r_{1}=0}^{m-1}\sum_{r_{2}\in S(r_{1})}\sum_{r_{3}\in R(r_{1},r_{2})}(q^{{\rm gcd}(k,r_{1})}-1)\\
			=&\frac{1}{mn(q-1)}\sum_{r_{1}=0}^{m-1}|S(r_{1})||R(r_{1},r_{2})|(q^{{\rm gcd}(k,r_{1})}-1)\\
			=&\frac{1}{mn(q-1)}\sum_{r_{1}=0}^{m-1}{\rm gcd}\big(n|\langle \xi \rangle \cap \langle \theta^{q^{r_{1}}-1}\rangle|,i_{t}|\langle \xi \rangle||\langle \theta^{q^{r_{1}}-1} \rangle|\big)(q^{{\rm gcd}(k,r_{1})}-1)\\
			=&\frac{1}{mn(q-1)}\sum_{r_{1}=0}^{m-1}{\rm gcd}\big(n(q-1),\frac{n(q^{k}-1)}{q^{{\rm gcd}(k,r_{1})}-1},\frac{i_{t}(q-1)(q^{k}-1)}{q^{{\rm gcd}(k,r_{1})}-1}\big)(q^{{\rm gcd}(k,r_{1})}-1)\\
			=&\frac{1}{mn(q-1)}\sum_{r_{1}=0}^{m-1}{\rm gcd}\big(n(q-1)(q^{{\rm gcd}(k,r_{1})}-1),n(q^{k}-1),i_{t}(q-1)(q^{k}-1)\big)\\
			=&\frac{1}{m}\sum_{r_{1}=0}^{m-1}{\rm gcd}\big(q^{{\rm gcd}(k,r_{1})}-1,\frac{q^{k}-1}{q-1},\frac{i_{t}(q^{k}-1)}{n}\big)\\
			=&\frac{1}{m}\cdot \frac{m}{k}\sum_{r_{1}=0}^{k-1}{\rm gcd}\big(q^{{\rm gcd}(k,r_{1})}-1,\frac{q^{k}-1}{q-1},\frac{i_{t}(q^{k}-1)}{n}\big)\\
			=&\frac{1}{k}\sum_{r\mid k}\varphi(\frac{k}{r}){\rm gcd}\big(q^{r}-1,\frac{q^{k}-1}{q-1},\frac{i_{t}(q^{k}-1)}{n}\big).
		\end{align*}
		We are done.
	\end{proof}
	
	\begin{Remark}{\rm
			Let $\mathcal{C}$ be the irreducible cyclic code in Theorem \ref{t3.1}. By \cite[Lemma III.8]{7}, the number of orbits of $\langle \rho,\sigma_{\xi} \rangle$ on $\mathcal{C}^{*}=\mathcal{C}\backslash \{\bf 0\}$ is equal to $$\frac{{\rm gcd}\big(n,(q-1)i_{t}\big)(q^{k}-1)}{n(q-1)}={\rm gcd}\big(\frac{q^{k}-1}{q-1},\frac{i_{t}(q^{k}-1)}{n}\big).$$
			Then
			\begin{align*}
				&\big|\langle \rho,\sigma_{\xi} \rangle \backslash \mathcal{C}^{*}\big|-\big|\langle \mu_{q},\rho,\sigma_{\xi} \rangle \backslash \mathcal{C}^{*}\big|\\
				=&{\rm gcd}\big(\frac{q^{k}-1}{q-1},\frac{i_{t}(q^{k}-1)}{n}\big)-\frac{1}{k}\sum_{r\mid k}\varphi(\frac{k}{r}){\rm gcd}\big(q^{r}-1,\frac{q^{k}-1}{q-1},\frac{i_{t}(q^{k}-1)}{n}\big)\\
				=&\frac{1}{k}\sum_{r\mid k}\varphi(\frac{k}{r})\Big({\rm gcd}\big(\frac{q^{k}-1}{q-1},\frac{i_{t}(q^{k}-1)}{n}\big)-{\rm gcd}\big(q^{r}-1,\frac{q^{k}-1}{q-1},\frac{i_{t}(q^{k}-1)}{n}\big)\Big)\\
				\geq &0,
			\end{align*}
			and it is easy to check that $\big|\langle \rho,\sigma_{\xi} \rangle \backslash \mathcal{C}^{*}\big|>\big|\langle \mu_{q},\rho,\sigma_{\xi} \rangle \backslash \mathcal{C}^{*}\big|$ if and only if $k>1$ and ${\rm gcd}\big(q-1,\frac{q^{k}-1}{q-1},\frac{i_{t}(q^{k}-1)}{n}\big)<{\rm gcd}\big(\frac{q^{k}-1}{q-1},\frac{i_{t}(q^{k}-1)}{n}\big)$. }
	\end{Remark}
	
	We include a small example to show that the upper bound in Theorem \ref{t3.1} is tight and in some cases is strictly less than the upper bound in the reference \cite[Lemma III.8]{7}.
	\begin{Example}{\rm
			Take $q=2$ and $n=9$. All the distinct $2$-cyclotomic cosets modulo $9$ are given by
			$$\Gamma_{0}=\{0\},~\Gamma_{1}=\{1,2,4,5,7,8\},~\Gamma_{2}=\{3,6\}.$$
			Consider the irreducible cyclic code $\mathcal{C}=\mathcal{R}_{n}\varepsilon_{1}$, where the primitive idempotent $\varepsilon_{1}$ corresponds to $\Gamma_{1}$. Let $\ell$ be the number of non-zero weights of $\mathcal{C}$. By \cite[Lemma III.8]{7}, we have
			$$\ell\leq \frac{{\rm gcd}(9,1)(2^{6}-1)}{9}=7.$$
			By Theorem \ref{t3.1}, we have
			\begin{align*}
				\ell&\leq \frac{1}{6}\sum_{r\mid 6}\varphi(\frac{6}{r}){\rm gcd}\big(2^{r}-1,\frac{2^6-1}{2-1},\frac{2^{6}-1}{9}\big)\\
				&=\frac{1}{6}[\varphi(6)+\varphi(3)+7\varphi(2)+7\varphi(1)]\\
				&=\frac{1}{6}(2+2+7+7)=3.
			\end{align*}
			Using the Magma software programming \cite{4}, we see that the weight distribution of $\mathcal{C}$ is $1+9x^{2}+27x^{4}+27x^{6}$, giving $\ell=3$. Theorem \ref{t3.1} therefore ensures that all the non-zero codewords of $\mathcal{C}$ with the same weight are in the same $\langle \mu_{q},\rho,\sigma_{\xi} \rangle$-orbit.}
	\end{Example}
	
	The following two corollaries suggest that Theorem \ref{t3.1} could be used to produce new few-weight irreducible cyclic codes.
	
	\begin{Corollary}\label{c3.1}
		Suppose that $q=2^{m}$ with $m>1$, ${\rm gcd}(q-1,3)=1$ and $3|(q+1)$. Let $\mathcal{C}$ be an irreducible cyclic code of length $n=\frac{q^{2}-1}{3N}$ over $\mathbb{F}_{q}$, where $N$ is a divisor of $q-1$. Suppose $i_{t}$ is an integer belonging to the $q$-cyclotomic coset corresponding to the primitive idempotent generating $\mathcal{C}$. If ${\rm gcd}\big(i_{t}, \frac{q+1}{3}\big)=1$, then $\mathcal{C}$ is a one-weight or two-weight cyclic code.
	\end{Corollary}
	
	\begin{proof}
		Let $k_{t}$ be the least positive integer such that $i_{t}q^{k_{t}}\equiv i_{t}~({\rm mod}~n)$, or equivalently, $q^{k_{t}}\equiv 1~\big({\rm mod}~\frac{n}{{\rm gcd}(n,i_{t})}\big)$. Note that $q^{2}\equiv 1~\big({\rm mod}~\frac{n}{{\rm gcd}(n,i_{t})}\big)$, then $k_{t}|2$, and so $k_{t}=1$ or $2$. If $k_{t}=1$, then $n|i_{t}(q-1)$, that is, $\frac{(q+1)(q-1)}{3N}|i_{t}(q-1)$, implying $\frac{q+1}{3}|i_{t}N$. Since ${\rm gcd}\big(i_{t},\frac{q+1}{3})=1$, we have $\frac{q+1}{3}|N$, which is impossible because $N|(q-1)$, $\frac{q+1}{3}|(q+1)$ and ${\rm gcd}(q-1,q+1)=1$. So $k_{t}=2$. It follows from Theorem \ref{t3.1} that the number of orbits of $\big\langle \mu_{q},\rho,\sigma_{\xi} \big\rangle$ on $\mathcal{C}^{*}=\mathcal{C}\backslash \{\bf 0\}$ is equal to
		\begin{align*}
			&\frac{1}{2}\sum_{r\mid 2}\varphi(\frac{2}{r}){\rm gcd}\big(q^{r}-1,\frac{q^{2}-1}{q-1},\frac{i_{t}(q^{2}-1)}{n}\big)\\
			=&\frac{1}{2}\big[\varphi(2){\rm gcd}\big(q-1,q+1,\frac{i_{t}(q^{2}-1)}{n}\big)+\varphi(1){\rm gcd}\big(q^{2}-1,q+1,\frac{i_{t}(q^{2}-1)}{n}\big)\big]\\
			=&\frac{1}{2}\big[\varphi(2)+\varphi(1){\rm gcd}\big(q+1,\frac{i_{t}(q^{2}-1)}{n}\big)\big]\\
			=&\frac{1}{2}[1+{\rm gcd}(q+1,3i_{t}N)]\\
			=&\frac{1}{2}(1+3)=2,
		\end{align*}
		and hence the number of non-zero weights of $\mathcal{C}$ is less than or equal to $2$, i.e., $\mathcal{C}$ is a one-weight or two-weight cyclic code.
	\end{proof}
	
	\begin{Example}{\rm
			Take $q=8$ and $n=21$. All the distinct $8$-cyclotomic cosets modulo $21$ are given by
			$$\Gamma_{0}=\{0\},~\Gamma_{1}=\{1,8\}, ~\Gamma_{2}=\{2,16\},~ \Gamma_{3}=\{3\},~\Gamma_{4}=\{4,11\},$$
			$$\Gamma_{5}=\{5,19\},~\Gamma_{6}=\{6\},~\Gamma_{7}=\{7,14\}, ~\Gamma_{8}=\{9\},~ \Gamma_{9}=\{10,17\},$$
			$$\Gamma_{10}=\{12\},~
			\Gamma_{11}=\{13,20\},~\Gamma_{12}=\{15\},~
			\Gamma_{13}=\{18\}.$$
			Consider the irreducible cyclic code $\mathcal{C}=\mathcal{R}_{n}\varepsilon_{7}$, where the primitive idempotent $\varepsilon_{7}$ corresponds to $\Gamma_{7}$. According to Corollary \ref{c3.1}, $\mathcal{C}$ is a one-weight or two-weight cyclic code. By use of Magma \cite{4}, the weight distribution of $\mathcal{C}$ is $1+21x^{14}+42x^{21}$, i.e., $\mathcal{C}$ is a two-weight cyclic code. Theorem \ref{t3.1} therefore guarantees that any two codewords of $\mathcal{C}$ with the same weight are in the same $\langle\mu_{q},\rho,\sigma_{\xi} \rangle$-orbit.}
	\end{Example}

	\begin{Corollary}\label{c3.2}
		Suppose that $(q,k)\neq (2,3)$ and that $k$ and $2k+1$ are odd primes satisfying ${\rm gcd}(q-1,k)=1$, ${\rm gcd}(q-1,2k+1)=1$ and $q^{k}\equiv 1~({\rm mod}~2k+1)$. Let $\mathcal{C}$ be an irreducible cyclic code of length $n=\frac{q^{k}-1}{(2k+1)N}$ over $\mathbb{F}_{q}$, where $N$ is a divisor of $q-1$. Suppose $i_{t}$ is an integer belonging to the $q$-cyclotomic coset corresponding to the primitive idempotent generating $\mathcal{C}$. If ${\rm gcd}\big(i_{t}, \frac{q^{k}-1}{(2k+1)(q-1)}\big)=1$, then $\mathcal{C}$ is a one-, two- or three-weight cyclic code.
	\end{Corollary}

	\begin{proof}
		Since $(q-1)|(q^{k}-1)$, $(2k+1)|(q^{k}-1)$ and ${\rm gcd}(q-1,2k+1)=1$, we have
		$$q^{k}-1=(2k+1)(q-1)l$$
		for some positive integer $l$. We claim that $l>1$. Otherwise,
		$$2k+1=\frac{q^{k}-1}{q-1}=q(q+1)(q^{k-3}+q^{k-5}+\cdots+q^{2}+1)+1,$$
		and hence $k=\frac{q(q+1)}{2}(q^{k-3}+q^{k-5}+\cdots+q^{2}+1)$, which is impossible as $k$ is a prime. Let $k_{t}$ be the least positive integer such that $i_{t}q^{k_{t}}\equiv i_{t}~({\rm mod}~n)$, or equivalently, $q^{k_{t}}\equiv 1~\big({\rm mod}~\frac{n}{{\rm gcd}(n,i_{t})}\big)$. Note that $q^{k}\equiv 1~\big({\rm mod}~\frac{n}{{\rm gcd}(n,i_{t})}\big)$, then $k_{t}|k$, and hence $k_{t}=1$ or $k$. If $k_{t}=1$, then $n|i_{t}(q-1)$, that is, $\frac{l(q-1)}{N}|i_{t}(q-1)$, implying $l|i_{t}N$. As ${\rm gcd}(i_{t},l)=1$, we have $l|N$, which is impossible because $N|(q-1)$, $l|\frac{q^{k}-1}{q-1}$ and ${\rm gcd}(q-1,\frac{q^{k}-1}{q-1})={\rm gcd}(q-1,k)=1$. So $k_{t}=k$. It follows from Theorem \ref{t3.1} that the number of orbits of $\langle \mu_{q},\rho,\sigma_{\xi} \rangle$ on $\mathcal{C}^{*}=\mathcal{C}\backslash \{\bf 0\}$ is equal to
		\begin{align*}
			&\frac{1}{k}\sum_{r\mid k}\varphi(\frac{k}{r}){\rm gcd}\big(q^{r}-1,\frac{q^{k}-1}{q-1},\frac{i_{t}(q^{k}-1)}{n}\big)\\
			=&\frac{1}{k}\big[\varphi(k){\rm gcd}\big(q-1,\frac{q^{k}-1}{q-1},\frac{i_{t}(q^{k}-1)}{n}\big)+\varphi(1){\rm gcd}\big(q^{k}-1,\frac{q^{k}-1}{q-1},\frac{i_{t}(q^{k}-1)}{n}\big)\big]\\
			=&\frac{1}{k}\big[\varphi(k)+\varphi(1){\rm gcd}\big(\frac{q^{k}-1}{q-1},\frac{i_{t}(q^{k}-1)}{n}\big)\big]\\
			=&\frac{1}{k}\big[k-1+{\rm gcd}\big((2k+1)l,i_{t}(2k+1)N\big)\big]\\
			=&\frac{1}{k}(k-1+2k+1)\\
			=&3,
		\end{align*}
		and so the number of non-zero weights of $\mathcal{C}$ is less than or equal to $3$, that is, $\mathcal{C}$ is a one-, two- or three-weight cyclic code.
	\end{proof}
	
	\begin{Example}
		Take $q=3$ and $n=22$. All the distinct $3$-cyclotomic cosets modulo $22$ are given by
		$$\Gamma_{0}=\{0\},~\Gamma_{1}=\{1,3,9,5,15\},~\Gamma_{2}=\{2,6,18,10,8\},$$
		$$\Gamma_{3}=\{4,12,14,20,16\},~\Gamma_{4}=\{7,21,19,13,17\},~\Gamma_{5}=\{11\}.$$
		Consider the irreducible cyclic code $\mathcal{C}=\mathcal{R}_{n}\varepsilon_{2}$, where the primitive idempotent $\varepsilon_{2}$ corresponds to $\Gamma_{2}$. According to Corollary 3.2, $\mathcal{C}$ is a one-, two- or three-weight cyclic code. By use of Magma \cite{4}, the weight distribution of $\mathcal{C}$ is $1+132x^{12}+110x^{18}$, i.e., $\mathcal{C}$ is a two-weight cyclic code.
	\end{Example}

	\subsection{An improved upper bound on the number of non-zero weights of a general cyclic code}
	We now turn to consider the action of $\langle\mu_{q},\rho,\sigma_{\xi} \rangle$ on a general cyclic code $\mathcal{C}$. Let $j_{1},j_{2},\cdots,j_{u}$ be positive integers and let $t_{j_{1}},t_{j_{2}},\cdots,t_{j_{u}}$ be integers with $0\leq t_{j_{1}}<t_{j_{2}}<\cdots<t_{j_{u}}\leq s$. Suppose that the irreducible cyclic code $\mathcal{R}_{n}\varepsilon_{t_{j_{\ell}}}$ corresponds to the $q$-cyclotomic coset $\{i_{t_{j_{\ell}}},i_{t_{j_{\ell}}}q,\cdots,i_{t_{j_{\ell}}}q^{k_{t_{j_{\ell}}}-1}\}$ for $1\leq \ell\leq u$. Define
	$$\mathcal{C}_{j_{1},j_{2},\cdots, j_{u}}^{\sharp}=\mathcal{R}_{n}\varepsilon_{t_{j_{1}}}\backslash\{{\bf 0}\}\bigoplus \mathcal{R}_{n}\varepsilon_{t_{j_{2}}}\backslash\{{\bf 0}\}\bigoplus \cdots \bigoplus \mathcal{R}_{n}\varepsilon_{t_{j_{u}}}\backslash\{{\bf 0}\}.$$
	
	The following lemma gives an upper bound on the number of orbits of the group $\langle\mu_{q},\rho,\sigma_{\xi} \rangle$ acting on $\mathcal{C}_{j_{1},j_{2},\cdots, j_{u}}^{\sharp}$.

	\begin{Lemma}\label{l3.1}
		With the notation given above, then the number of orbits of $\langle\mu_{q},\rho,\sigma_{\xi} \rangle$ on $\mathcal{C}_{j_{1},j_{2},\cdots, j_{u}}^{\sharp}$ is less than or equal to
		\begin{align*}
			&\frac{1}{mn(q-1)}\sum_{r=0}^{m-1}{\rm gcd}\big(n,\frac{i_{t_{j_{1}}}II_{t_{j_{1}}}}{{\rm gcd}(I,I_{t_{j_{1}}})},\cdots,\frac{i_{t_{j_{u}}}II_{t_{j_{u}}}}{{\rm gcd}(I,I_{t_{j_{u}}})},\frac{(i_{t_{j_{2}}}-i_{t_{j_{1}}})I_{t_{j_{1}}}I_{t_{j_{2}}}}{{\rm gcd}(I_{t_{j_{1}}},I_{t_{j_{2}}})},\cdots,\\ &\quad \frac{(i_{t_{j_{u}}}-i_{t_{j_{1}}})I_{t_{j_{1}}}I_{t_{j_{u}}}}{{\rm gcd}(I_{t_{j_{1}}},I_{t_{j_{u}}})},\cdots,\frac{(i_{t_{j_{u}}}-i_{t_{j_{u-1}}})I_{t_{j_{u-1}}}I_{t_{j_{u}}}}{{\rm gcd}(I_{t_{j_{u-1}}},I_{t_{j_{u}}})}\big){\rm gcd}(I,I_{t_{j_{1}}},\cdots,I_{t_{j_{u}}})\prod_{\ell=1}^{u}(q^{{\rm gcd}(k_{t_{j_{u}}},r)}-1),
		\end{align*}
		where $I=q-1$ and $I_{t_{j_{\ell}}}=\frac{q^{k_{t_{j_{\ell}}}}-1}{q^{{\rm gcd}(k_{t_{j_{\ell}}},r)}-1}$ for $\ell=1,2,\cdots,u$.
	\end{Lemma}
	
	\begin{proof}
		According to Equation (\ref{e2.1}) and Lemma \ref{l2.2}, we see that
		$$\big|\langle\mu_{q},\rho,\sigma_{\xi} \rangle\backslash \mathcal{C}_{j_{1},j_{2},\cdots, j_{u}}^{\sharp}\big|=\frac{1}{mn(q-1)}\sum_{r_{1}=0}^{m-1}\sum_{r_{2}=0}^{n-1}\sum_{r_{3}=0}^{q-2}\Big|\big\{{\bf c}\in \mathcal{C}_{j_{1},j_{2},\cdots, j_{u}}^{\sharp}~\big|~\mu_{q}^{r_{1}}\rho^{r_{2}}\sigma_{\xi}^{r_{3}}({\bf c})={\bf c}\big\}\Big|.$$
		
		Let ${\bf c}={\bf c}_{t_{j_{1}}}+{\bf c}_{t_{j_{2}}}+\cdots+{\bf c}_{t_{j_{u}}}\in \mathcal{C}_{j_{1},j_{2},\cdots, j_{u}}^{\sharp}$, where ${\bf c}_{t_{j_{\ell}}}\in \mathcal{R}_{n}\varepsilon_{t_{j_{\ell}}}\backslash\{{\bf 0}\}$ for $1\leq \ell \leq u$. Suppose that
		$${\bf c}_{t_{j_{\ell}}}=\sum_{j=0}^{k_{t_{j_{\ell}}}-1}\big(c_{0,t_{j_{\ell}}}+c_{1,t_{j_{\ell}}}\zeta^{i_{t_{j_{\ell}}}q^{j}}+\cdots+c_{k_{t_{j_{\ell}}}-1,t_{j_{\ell}}}\zeta^{(k_{t_{j_{\ell}}}-1)i_{t_{j_{\ell}}}q^{j}}\big)e_{i_{t_{j_{\ell}}}q^{j}}~~{\rm for}~\ell=1,2,\cdots,u.$$
		Then we have
		\begin{align*}
			\mu_{q}^{r_{1}}\rho^{r_{2}}\sigma_{\xi}^{r_{3}}({\bf c})&=\mu_{q}^{r_{1}}\rho^{r_{2}}\sigma_{\xi}^{r_{3}}({\bf c}_{t_{j_{1}}})+\cdots+\mu_{q}^{r_{1}}\rho^{r_{2}}\sigma_{\xi}^{r_{3}}({\bf c}_{t_{j_{u}}})\\
			&=\sum_{j=0}^{k_{t_{j_{1}}}-1}\xi^{r_{3}}\zeta^{i_{t_{j_{1}}}q^{r_{1}+j}r_{2}}\big(\sum_{l=0}^{k_{t_{j_{1}}}-1}c_{l,t_{j_{1}}}\zeta^{li_{t_{j_{1}}}q^{j}}\big)^{q^{r_{1}}}e_{i_{t_{j_{1}}}q^{j}}+\cdots\\ 
			&\quad+\sum_{j=0}^{k_{t_{j_{u}}}-1}\xi^{r_{3}}\zeta^{i_{t_{j_{u}}}q^{r_{1}+j}r_{2}}\big(\sum_{l=0}^{k_{t_{j_{u}}}-1}c_{l,t_{j_{u}}}\zeta^{li_{t_{j_{u}}}q^{j}}\big)^{q^{r_{1}}}e_{i_{t_{j_{u}}}q^{j}}.
		\end{align*}
		It follows that
		\begin{align*}
			\mu_{q}^{r_{1}}\rho^{r_{2}}\sigma_{\xi}^{r_{3}}({\bf c})={\bf c}~& \Leftrightarrow~\mu_{q}^{r_{1}}\rho^{r_{2}}\sigma_{\xi}^{r_{3}}({\bf c}_{t_{j_{1}}})+\cdots+\mu_{q}^{r_{1}}\rho^{r_{2}}\sigma_{\xi}^{r_{3}}({\bf c}_{t_{j_{u}}})={\bf c}_{t_{j_{1}}}+\cdots+{\bf c}_{t_{j_{u}}}\\
			&\Leftrightarrow~ \mu_{q}^{r_{1}}\rho^{r_{2}}\sigma_{\xi}^{r_{3}}({\bf c}_{t_{j_{\ell}}})={\bf c}_{t_{j_{\ell}}} ~{\rm for}~\ell=1,2,\cdots,u\\
			&\Leftrightarrow~\xi^{r_{3}}\big(\sum_{l=0}^{k_{t_{j_{\ell}}}-1}c_{l,t_{j_{\ell}}}\zeta^{li_{t_{j_{\ell}}}q^{j}}\big)^{q^{r_{1}}-1}=\zeta^{-i_{t_{j_{\ell}}}q^{r_{1}}r_{2}}~{\rm for}~\ell=1,2,\cdots,u.
		\end{align*}
		For $1\leq \ell\leq u$, the minimal polynomial of $\zeta^{i_{t_{j_{\ell}}}}$ over $\mathbb{F}_{q}$ is of degree $k_{t_{j_{\ell}}}$, and so the set
		$$\Big\{c_{0,t_{j_{\ell}}}+c_{1,t_{j_{\ell}}}\zeta^{i_{t_{j_{\ell}}}}+\cdots+c_{k_{t_{j_{\ell}}}-1,t_{j_{\ell}}}\zeta^{(k_{t_{j_{\ell}}}-1)i_{t_{j_{\ell}}}}~\Big|~c_{l,t_{j_{\ell}}}\in \mathbb{F}_{q},~ 0\leq l\leq k_{t_{j_{\ell}}}-1\Big\}$$
		forms a subfield $\mathbb{F}_{q^{k_{t_{j_{\ell}}}}}$ of $\mathbb{F}_{q^{m}}$. Then the number of ${\bf c}\in \mathcal{C}_{j_{1},j_{2},\cdots, j_{u}}^{\sharp}$ satisfying $\mu_{q}^{r_{1}}\rho^{r_{2}}\sigma_{\xi}^{r_{3}}({\bf c})={\bf c}$ is equal to the number of $u$-tuples $(\alpha_{t_{j_{1}}},\alpha_{t_{j_{2}}},\cdots,\alpha_{t_{j_{u}}})$ with $\alpha_{t_{j_{\ell}}}\in \mathbb{F}_{q^{k_{t_{j_{\ell}}}}}^{*}$ such that $\xi^{r_{3}}\alpha_{t_{j_{\ell}}}^{q^{r_{1}}-1}=\zeta^{-i_{t_{j_{\ell}}}q^{r_{1}}r_{2}}$ for all $1\leq \ell\leq u$, which is easily checked to be $0$ or $\prod\limits_{\ell=1}^{u}(q^{{\rm gcd}(k_{t_{j_{\ell}}},r_{1})}-1)$. Next we aim to find the number of pairs $(r_{2},r_{3})$ with $0\leq r_{2}\leq n-1$ and $0\leq r_{3}\leq q-2$ such that there exists some $u$-tuple $(\alpha_{t_{j_{1}}},\alpha_{t_{j_{2}}},\cdots,\alpha_{t_{j_{u}}})$ with $\alpha_{t_{j_{\ell}}}\in \mathbb{F}_{q^{k_{t_{j_{\ell}}}}}^{*}$ such that $\xi^{r_{3}}\alpha_{t_{j_{\ell}}}^{q^{r_{1}}-1}=\zeta^{-i_{t_{j_{\ell}}}q^{r_{1}}r_{2}}$ for all $1\leq \ell\leq u$.
		
		For $1\leq \ell\leq u$, write $\mathbb{F}_{q^{k_{t_{j_{\ell}}}}}^{*}=\langle \theta_{t_{j_{\ell}}} \rangle$, that is, $\mathbb{F}_{q^{k_{t_{j_{\ell}}}}}^{*}$ is generated by $\theta_{t_{j_{\ell}}}$. It is easy to check that $\big\{\alpha_{t_{j_{\ell}}}^{q^{r_{1}}-1}~\big|~\alpha_{t_{j_{\ell}}}\in \mathbb{F}_{q^{k_{t_{j_{\ell}}}}}^{*}\big\}=\langle \theta_{t_{j_{\ell}}}^{q^{r_{1}}-1}\rangle$, which is a cyclic subgroup of $\mathbb{F}_{q^{m}}^{*}$ of order $I_{t_{j_{\ell}}}=\frac{q^{k_{t_{j_{\ell}}}}-1}{q^{{\rm gcd}(k_{t_{j_{\ell}}},r_{1})}-1}$. Since $\langle \xi \rangle$ is a cyclic subgroup of $\mathbb{F}_{q^{m}}^{*}$ of order $I=q-1$, $\langle \xi \rangle \langle \theta_{t_{j_{\ell}}}^{q^{r_{1}}-1} \rangle$ is a cyclic subgroup of $\mathbb{F}_{q^{m}}^{*}$ of order $\frac{II_{t_{j_{\ell}}}}{{\rm gcd}(I,I_{t_{j_{\ell}}})}$ for $1\leq \ell\leq u$. For $1\leq \ell< \ell'\leq u$, $\langle \theta_{t_{j_{\ell}}}^{q^{r_{1}}-1}  \rangle \langle \theta_{t_{j_{\ell'}}}^{q^{r_{1}}-1} \rangle$ is a cyclic subgroup of $\mathbb{F}_{q^{m}}^{*}$ of order $\frac{I_{t_{j_{\ell}}}I_{t_{j_{\ell'}}}}{{\rm gcd}(I_{t_{j_{\ell}}},I_{t_{j_{\ell'}}})}$. 
		
		For $0\leq r_{1}\leq m-1$, denote
		$$S(r_{1})=\Big\{0\leq r_{2}\leq n-1~\Big|~\exists~ 0\leq r_{3}\leq q-2 ~s.t.~ \zeta^{-i_{t_{j_{\ell}}}q^{r_{1}}r_{2}}\in \xi^{r_{3}}\langle \theta_{t_{j_{\ell}}}^{q^{r_{1}}-1}  \rangle ~{\rm for~all~}1\leq \ell\leq u\Big\}.$$
		Then we have
		\begin{align*}
			r_{2}\in S(r_{1}) ~&\Rightarrow~\zeta^{-i_{t_{j_{\ell}}}q^{r_{1}}r_{2}}\in \langle \xi\rangle\langle \theta_{t_{j_{\ell}}}^{q^{r_{1}}-1}  \rangle ~{\rm for~}1\leq \ell\leq u,\\
			& \qquad {\rm and}~\zeta^{-(i_{t_{j_{\ell'}}}-i_{t_{j_{\ell}}})q^{r_{1}}r_{2}}\in \langle \theta_{t_{j_{\ell}}}^{q^{r_{1}}-1}  \rangle\langle \theta_{t_{j_{\ell'}}}^{q^{r_{1}}-1}  \rangle ~{\rm for~}1\leq \ell< \ell'\leq u\\
			&\Leftrightarrow~\frac{n}{{\rm gcd}(n,i_{t_{j_{\ell}}}r_{2})}\Big|\frac{II_{t_{j_{\ell}}}}{{\rm gcd}(I,I_{t_{j_{\ell}}})}~{\rm for~}1\leq \ell\leq u,\\
			& \qquad {\rm and}~\frac{n}{{\rm gcd}\big(n,(i_{t_{j_{\ell'}}}-i_{t_{j_{\ell}}})r_{2}\big)}\Big|\frac{I_{t_{j_{\ell}}}I_{t_{j_{\ell'}}}}{{\rm gcd}(I_{t_{j_{\ell}}},I_{t_{j_{\ell'}}})} ~{\rm for~}1\leq \ell< \ell'\leq u\\
			&\Leftrightarrow~\frac{n}{{\rm gcd}\big(n,\frac{i_{t_{j_{\ell}}}II_{t_{j_{\ell}}}}{{\rm gcd}(I,I_{t_{j_{\ell}}})}\big)}\Biggl|r_{2}~{\rm for~}1\leq \ell\leq u,\\
			& \qquad {\rm and}~\frac{n}{{\rm gcd}\big(n,\frac{(i_{t_{j_{\ell'}}}-i_{t_{j_{\ell}}})I_{t_{j_{\ell}}}I_{t_{j_{\ell'}}}}{{\rm gcd}(I_{t_{j_{\ell}}},I_{t_{j_{\ell'}}})}\big)}\Biggl|r_{2} ~{\rm for~}1\leq \ell< \ell'\leq u\\
			&\Leftrightarrow~\frac{n}{{\rm gcd}\big(n,\frac{i_{t_{j_{1}}}II_{t_{j_{1}}}}{{\rm gcd}(I,I_{t_{j_{1}}})},\cdots,\frac{i_{t_{j_{u}}}II_{t_{j_{u}}}}{{\rm gcd}(I,I_{t_{j_{u}}})}\big)}\Biggl|r_{2}~{\rm and}\\
			& \qquad \frac{n}{{\rm gcd}\big(n,\frac{(i_{t_{j_{2}}}-i_{t_{j_{1}}})I_{t_{j_{1}}}I_{t_{j_{2}}}}{{\rm gcd}(I_{t_{j_{1}}},I_{t_{j_{2}}})},\cdots,\frac{(i_{t_{j_{u}}}-i_{t_{j_{1}}})I_{t_{j_{1}}}I_{t_{j_{u}}}}{{\rm gcd}(I_{t_{j_{1}}},I_{t_{j_{u}}})},\cdots,\frac{(i_{t_{j_{u}}}-i_{t_{j_{u-1}}})I_{t_{j_{u-1}}}I_{t_{j_{u}}}}{{\rm gcd}(I_{t_{j_{u-1}}},I_{t_{j_{u}}})}\big)}\Biggl|r_{2}\\
			&\Leftrightarrow~\frac{n}{{\rm gcd}(d_{1},d_{2})}\Big|r_{2},~{\rm where}~d_{1}={\rm gcd}\big(n,\frac{i_{t_{j_{1}}}II_{t_{j_{1}}}}{{\rm gcd}(I,I_{t_{j_{1}}})},\cdots,\frac{i_{t_{j_{u}}}II_{t_{j_{u}}}}{{\rm gcd}(I,I_{t_{j_{u}}})}\big)\\
			& \qquad {\rm and} ~d_{2}={\rm gcd}\big(n,\frac{(i_{t_{j_{2}}}-i_{t_{j_{1}}})I_{t_{j_{1}}}I_{t_{j_{2}}}}{{\rm gcd}(I_{t_{j_{1}}},I_{t_{j_{2}}})},\cdots,\frac{(i_{t_{j_{u}}}-i_{t_{j_{1}}})I_{t_{j_{1}}}I_{t_{j_{u}}}}{{\rm gcd}(I_{t_{j_{1}}},I_{t_{j_{u}}})},\cdots,\\
			&\qquad \frac{(i_{t_{j_{u}}}-i_{t_{j_{u-1}}})I_{t_{j_{u-1}}}I_{t_{j_{u}}}}{{\rm gcd}(I_{t_{j_{u-1}}},I_{t_{j_{u}}})}\big).
		\end{align*}
		Hence
		\begin{align*}
			|S(r_{1})|&\leq {\rm gcd}(d_{1},d_{2})\\
			&= {\rm gcd}\big(n,\frac{i_{t_{j_{1}}}II_{t_{j_{1}}}}{{\rm gcd}(I,I_{t_{j_{1}}})},\cdots,\frac{i_{t_{j_{u}}}II_{t_{j_{u}}}}{{\rm gcd}(I,I_{t_{j_{u}}})},\frac{(i_{t_{j_{2}}}-i_{t_{j_{1}}})I_{t_{j_{1}}}I_{t_{j_{2}}}}{{\rm gcd}(I_{t_{j_{1}}},I_{t_{j_{2}}})},\cdots,\frac{(i_{t_{j_{u}}}-i_{t_{j_{1}}})I_{t_{j_{1}}}I_{t_{j_{u}}}}{{\rm gcd}(I_{t_{j_{1}}},I_{t_{j_{u}}})},\\
			&\quad \cdots,\frac{(i_{t_{j_{u}}}-i_{t_{j_{u-1}}})I_{t_{j_{u-1}}}I_{t_{j_{u}}}}{{\rm gcd}(I_{t_{j_{u-1}}},I_{t_{j_{u}}})}\big).
		\end{align*}
		
		Suppose $r_{2}\in S(r_{1})$, and denote
		$$R(r_{1},r_{2})=\Big\{0\leq r_{3}\leq q-2~\Big|~\zeta^{-i_{t_{j_{\ell}}}q^{r_{1}}r_{2}}\in \xi^{r_{3}}\langle \theta_{t_{j_{\ell}}}^{q^{r_{1}}-1}  \rangle ~{\rm for~all~}1\leq \ell\leq u\Big\}.$$
		Similar to the proof of Theorem \ref{t3.1}, we can prove that
		$$|R(r_{1},r_{2})|={\rm gcd}(I,I_{t_{j_{1}}},\cdots,I_{t_{j_{u}}}).$$
		
		We then conclude that
		\begin{align*}
			&\big|\langle\mu_{q},\rho,\sigma_{\xi} \rangle\backslash \mathcal{C}_{j_{1},j_{2},\cdots, j_{u}}^{\sharp}\big|\\
			=&\frac{1}{mn(q-1)}\sum_{r_{1}=0}^{m-1}\sum_{r_{2}\in S(r_{1})}\sum_{r_{3}\in R(r_{1},r_{2})}\prod_{\ell=1}^{u}(q^{{\rm gcd}(k_{t_{j_{\ell}}},r_{1})}-1)\\
			=&\frac{1}{mn(q-1)}\sum_{r_{1}=0}^{m-1}|S(r_{1})||R(r_{1},r_{2})|\prod_{\ell=1}^{u}(q^{{\rm gcd}(k_{t_{j_{\ell}}},r_{1})}-1)\\
			\leq & \frac{1}{mn(q-1)}\sum_{r_{1}=0}^{m-1}{\rm gcd}\big(n,\frac{i_{t_{j_{1}}}II_{t_{j_{1}}}}{{\rm gcd}(I,I_{t_{j_{1}}})},\cdots,\frac{i_{t_{j_{u}}}II_{t_{j_{u}}}}{{\rm gcd}(I,I_{t_{j_{\ell}}})},\frac{(i_{t_{j_{2}}}-i_{t_{j_{1}}})I_{t_{j_{1}}}I_{t_{j_{2}}}}{{\rm gcd}(I_{t_{j_{1}}},I_{t_{j_{2}}})},\cdots,\\
			&
			\frac{(i_{t_{j_{u}}}-i_{t_{j_{1}}})I_{t_{j_{1}}}I_{t_{j_{u}}}}{{\rm gcd}(I_{t_{j_{1}}},I_{t_{j_{u}}})},\cdots,\frac{(i_{t_{j_{u}}}-i_{t_{j_{u-1}}})I_{t_{j_{u-1}}}I_{t_{j_{u}}}}{{\rm gcd}(I_{t_{j_{u-1}}},I_{t_{j_{u}}})}\big){\rm gcd}(I,I_{t_{j_{1}}},\cdots,I_{t_{j_{u}}})\prod_{\ell=1}^{u}(q^{{\rm gcd}(k_{t_{j_{\ell}}},r_{1})}-1).
		\end{align*}
		We are done.
	\end{proof}
	
	Based on Lemma \ref{l3.1}, an upper bound on the number of non-zero weights of $\mathcal{C}$ is derived as follows.
	\begin{Theorem}\label{t3.2}
		Let $\mathcal{C}$ be a cyclic code of length $n$ over $\mathbb{F}_{q}$. Suppose that
		$$\mathcal{C}=\mathcal{R}_{n}\varepsilon_{t_{1}}\bigoplus \mathcal{R}_{n}\varepsilon_{t_{2}}\bigoplus \cdots \bigoplus \mathcal{R}_{n}\varepsilon_{t_{v}},$$
		where $0\leq t_{1}< t_{2}< \cdots< t_{v}\leq s$, and the primitive idempotent $\varepsilon_{t_{j}}$ corresponds to the $q$-cyclotomic coset $\{i_{t_{j}}, i_{t_{j}}q,\cdots,i_{t_{j}}q^{k_{t_{j}}-1}\}$ for each $1\leq j\leq v$. Then the number of orbits of $\langle\mu_{q},\rho,\sigma_{\xi} \rangle$ on $\mathcal{C}^{*}=\mathcal{C}\backslash \{\bf 0\}$ is less than or equal to
		$$\sum_{\substack{\{j_{1},j_{2},\cdots,j_{u}\}\subseteq \{1,2,\cdots,v\}\\ 1\leq j_{1}< j_{2}<\cdots<j_{u}\leq v}}N_{j_{1},j_{2},\cdots, j_{u}},$$
		where
		\begin{align*}
			N_{j_{1},j_{2},\cdots, j_{u}}&=\frac{1}{mn(q-1)}\sum_{r=0}^{m-1}{\rm gcd}\big(n,\frac{i_{t_{j_{1}}}II_{t_{j_{1}}}}{{\rm gcd}(I,I_{t_{j_{1}}})},\cdots,\frac{i_{t_{j_{u}}}II_{t_{j_{u}}}}{{\rm gcd}(I,I_{t_{j_{u}}})},\frac{(i_{t_{j_{2}}}-i_{t_{j_{1}}})I_{t_{j_{1}}}I_{t_{j_{2}}}}{{\rm gcd}(I_{t_{j_{1}}},I_{t_{j_{2}}})},\\
			&\quad \cdots,\frac{(i_{t_{j_{u}}}-i_{t_{j_{1}}})I_{t_{j_{1}}}I_{t_{j_{u}}}}{{\rm gcd}(I_{t_{j_{1}}},I_{t_{j_{u}}})},\cdots,\frac{(i_{t_{j_{u}}}-i_{t_{j_{u-1}}})I_{t_{j_{u-1}}}I_{t_{j_{u}}}}{{\rm gcd}(I_{t_{j_{u-1}}},I_{t_{j_{u}}})}\big){\rm gcd}(I,I_{t_{j_{1}}},\cdots,I_{t_{j_{u}}})\\
			&\quad \cdot\prod_{\ell=1}^{u}(q^{{\rm gcd}(k_{t_{j_{\ell}}},r)}-1)
		\end{align*}
		with $I=q-1$ and $I_{t_{j_{\ell}}}=\frac{q^{k_{t_{j_{\ell}}}}-1}{q^{{\rm gcd}(k_{t_{j_{\ell}}},r)}-1}$ for $\ell=1,2,\cdots,u$.
		
		In particular, the number of non-zero weights of $\mathcal{C}$ is less than or equal to the number of orbits of $\langle\mu_{q},\rho,\sigma_{\xi} \rangle$ on $\mathcal{C}^{*}$.
	\end{Theorem}
	
	\begin{proof}
		Note that
		$$\mathcal{C}\backslash \{{\bf 0}\}=\bigcup_{\substack{\{j_{1},j_{2},\cdots,j_{u}\}\subseteq \{1,2,\cdots,v\}\\ 1\leq j_{1}< j_{2}<\cdots<j_{u}\leq v}}\big(\mathcal{R}_{n}\varepsilon_{t_{j_{1}}}\backslash\{{\bf 0}\}\bigoplus \mathcal{R}_{n}\varepsilon_{t_{j_{2}}}\backslash\{{\bf 0}\}\bigoplus \cdots \bigoplus \mathcal{R}_{n}\varepsilon_{t_{j_{u}}}\backslash\{{\bf 0}\}\big)$$
		is a disjoint union. Let $s_{j_{1},j_{2},\cdots, j_{u}}$ be the number of orbits of $\langle\mu_{q},\rho,\sigma_{\xi} \rangle$ acting on $$\mathcal{R}_{n}\varepsilon_{t_{j_{1}}}\backslash\{{\bf 0}\}\bigoplus \mathcal{R}_{n}\varepsilon_{t_{j_{2}}}\backslash\{{\bf 0}\}\bigoplus \cdots \bigoplus \mathcal{R}_{n}\varepsilon_{t_{j_{u}}}\backslash\{{\bf 0}\}.$$
		Then
		$$\big|\langle\mu_{q},\rho,\sigma_{\xi} \rangle\backslash \mathcal{C}^{*}\big|=\bigcup_{\substack{\{j_{1},j_{2},\cdots,j_{u}\}\subseteq \{1,2,\cdots,v\}\\ 1\leq j_{1}< j_{2}<\cdots<j_{u}\leq v}}s_{j_{1},j_{2},\cdots, j_{u}}.$$
		We know from Lemma \ref{l3.1} that 
		\begin{align*}
			s_{j_{1},j_{2},\cdots, j_{u}}&\leq \frac{1}{mn(q-1)}\sum_{r=0}^{m-1}{\rm gcd}\big(n,\frac{i_{t_{j_{1}}}II_{t_{j_{1}}}}{{\rm gcd}(I,I_{t_{j_{1}}})},\cdots,\frac{i_{t_{j_{u}}}II_{t_{j_{u}}}}{{\rm gcd}(I,I_{t_{j_{u}}})},\frac{(i_{t_{j_{2}}}-i_{t_{j_{1}}})I_{t_{j_{1}}}I_{t_{j_{2}}}}{{\rm gcd}(I_{t_{j_{1}}},I_{t_{j_{2}}})},\\
			&\quad \cdots,\frac{(i_{t_{j_{u}}}-i_{t_{j_{1}}})I_{t_{j_{1}}}I_{t_{j_{u}}}}{{\rm gcd}(I_{t_{j_{1}}},I_{t_{j_{u}}})},\cdots,\frac{(i_{t_{j_{u}}}-i_{t_{j_{u-1}}})I_{t_{j_{u-1}}}I_{t_{j_{u}}}}{{\rm gcd}(I_{t_{j_{u-1}}},I_{t_{j_{u}}})}\big){\rm gcd}(I,I_{t_{j_{1}}},\cdots,I_{t_{j_{u}}})\\
			&\quad \cdot \prod_{\ell=1}^{u}(q^{{\rm gcd}(k_{t_{j_{\ell}}},r)}-1),
		\end{align*}
		where $I=q-1$ and $I_{t_{j_{\ell}}}=\frac{q^{k_{t_{j_{\ell}}}}-1}{q^{{\rm gcd}(k_{t_{j_{\ell}}},r)}-1}$ for $\ell=1,2,\cdots,u$. 
		Then we obtain the desired result.
	\end{proof}
	
	\begin{Remark}{\rm
			Let $\mathcal{C}$ be the cyclic code in Theorem \ref{t3.2}. According to \cite[Theorem III.17]{7},
			$$\big|\langle\rho,\sigma_{\xi} \rangle\backslash \mathcal{C}^{*}\big|=\sum_{\substack{\{j_{1},j_{2},\cdots,j_{u}\}\subseteq \{1,2,\cdots,v\}\\ 1\leq j_{1}< j_{2}<\cdots<j_{u}\leq v}}\big|\langle\rho,\sigma_{\xi} \rangle\backslash \mathcal{C}_{j_{1},j_{2},\cdots, j_{u}}^{\sharp}\big|,$$
			where $\mathcal{C}_{j_{1},j_{2},\cdots, j_{u}}^{\sharp}$ is as defined above and \cite[Lemma III.16]{7} proved that
			\begin{align*}
				\big|\langle\rho,\sigma_{\xi} \rangle\backslash \mathcal{C}_{j_{1},j_{2},\cdots, j_{u}}^{\sharp}\big|&=\frac{1}{n(q-1)}{\rm gcd}(n,i_{t_{j_{1}}},\cdots,i_{t_{j_{u}}}){\rm gcd}\big(q-1,\frac{n}{{\rm gcd}(n,i_{t_{j_{1}}})},\cdots,\\
				&\quad \frac{n}{{\rm gcd}(n,i_{t_{j_{u}}})}\big)\prod\limits_{\ell=1}^{u}(q^{k_{t_{j_{\ell}}}}-1).
			\end{align*}
			In fact, what \cite[Lemma III.16]{7} gave was not the value of $\big|\langle\rho,\sigma_{\xi} \rangle\backslash \mathcal{C}_{j_{1},j_{2},\cdots, j_{u}}^{\sharp}\big|$ but an upper bound of $\big|\langle\rho,\sigma_{\xi} \rangle\backslash \mathcal{C}_{j_{1},j_{2},\cdots, j_{u}}^{\sharp}\big|$ because of a mistake in calculating the value of $|{\rm Fix}(\sigma_{\xi}^{r})|$. In the proof of \cite[Lemma III.16]{7}, the following two conditions were considered to be equivalent:\\
			\quad 1) $\rho^{z}({\bf c}_{t_{j_{l}}})=\xi^{r}{\bf c}_{t_{j_{l}}}$ for all $l=1,2,\cdots,u$,\\
			\quad 2) $\frac{q-1}{r}\big| \frac{n}{{\rm gcd}(n,i_{t_{j_{l}}})}$ for all $l=1,2,\cdots,u$.\\
			However, condition 2) does not imply condition 1) because condition 1) is also dependent on the choice of $z$. In the proof of Lemma \ref{l3.1}, let $m=1$ and then we can obtain
			\begin{align*}
				\big|\langle\rho,\sigma_{\xi} \rangle\backslash \mathcal{C}_{j_{1},j_{2},\cdots, j_{u}}^{\sharp}\big|&=\frac{1}{n(q-1)}{\rm gcd}\big(n,i_{t_{j_{1}}}(q-1),\cdots,i_{t_{j_{u}}}(q-1),i_{t_{j_{2}}}-i_{t_{j_{1}}},\cdots,i_{t_{j_{u}}}-i_{t_{j_{1}}},\\
				&\quad \cdots,i_{t_{j_{u}}}-i_{t_{j_{u-1}}}\big)\prod\limits_{\ell=1}^{u}(q^{k_{t_{j_{\ell}}}}-1),
			\end{align*}
			which leads to the value of $\big|\langle\rho,\sigma_{\xi} \rangle\backslash \mathcal{C}^{*}\big|$.

			We claim that the upper bound of $\big|\langle\mu_{q},\rho,\sigma_{\xi} \rangle\backslash \mathcal{C}^{*}\big|$ given in Theorem \ref{t3.2} is less than or equal to $\big|\langle\rho,\sigma_{\xi} \rangle\backslash \mathcal{C}^{*}\big|$, and hence the upper bound on the number of non-zero weights of $\mathcal{C}$ given by Theorem \ref{t3.2} is less than or equal to that given by calculating the number of orbits of $\langle\rho,\sigma_{\xi} \rangle$ on $\mathcal{C}^{*}$. To prove this, it is sufficient to prove that $N_{j_{1},j_{2},\cdots, j_{u}}\leq \big|\langle\rho,\sigma_{\xi} \rangle\backslash \mathcal{C}_{j_{1},j_{2},\cdots, j_{u}}^{\sharp}\big|$. Indeed,
			\begin{align*}
				N_{j_{1},j_{2},\cdots, j_{u}}&=\frac{1}{mn(q-1)}\sum_{r=0}^{m-1}{\rm gcd}\big(n,\frac{i_{t_{j_{1}}}II_{t_{j_{1}}}}{{\rm gcd}(I,I_{t_{j_{1}}})},\cdots,\frac{i_{t_{j_{u}}}II_{t_{j_{u}}}}{{\rm gcd}(I,I_{t_{j_{u}}})},\frac{(i_{t_{j_{2}}}-i_{t_{j_{1}}})I_{t_{j_{1}}}I_{t_{j_{2}}}}{{\rm gcd}(I_{t_{j_{1}}},I_{t_{j_{2}}})},\cdots,\\
				&\quad \frac{(i_{t_{j_{u}}}-i_{t_{j_{1}}})I_{t_{j_{1}}}I_{t_{j_{u}}}}{{\rm gcd}(I_{t_{j_{1}}},I_{t_{j_{u}}})},\cdots,\frac{(i_{t_{j_{u}}}-i_{t_{j_{u-1}}})I_{t_{j_{u-1}}}I_{t_{j_{u}}}}{{\rm gcd}(I_{t_{j_{u-1}}},I_{t_{j_{u}}})}\big){\rm gcd}(I,I_{t_{j_{1}}},\cdots,I_{t_{j_{u}}})\\
				&\quad \cdot \prod_{\ell=1}^{u}(q^{{\rm gcd}(k_{t_{j_{\ell}}},r)}-1)\\
				&\leq \frac{1}{mn(q-1)}\sum_{r=0}^{m-1}{\rm gcd}\big(n,i_{t_{j_{1}}}II_{t_{j_{1}}},\cdots,i_{t_{j_{u}}}II_{t_{j_{u}}},(i_{t_{j_{2}}}-i_{t_{j_{1}}})I_{t_{j_{1}}}I_{t_{j_{2}}},\cdots,\\
				&\quad (i_{t_{j_{u}}}-i_{t_{j_{1}}})I_{t_{j_{1}}}I_{t_{j_{u}}},\cdots,(i_{t_{j_{u}}}-i_{t_{j_{u-1}}})I_{t_{j_{u-1}}}I_{t_{j_{u}}}\big)\prod_{\ell=1}^{u}(q^{{\rm gcd}(k_{t_{j_{\ell}}},r)}-1)\\
				&\leq \frac{1}{mn(q-1)}\sum_{r=0}^{m-1}{\rm gcd}\big(n,i_{t_{j_{1}}}I,\cdots,i_{t_{j_{u}}}I,i_{t_{j_{2}}}-i_{t_{j_{1}}},\cdots,i_{t_{j_{u}}}-i_{t_{j_{1}}},\cdots,\\
				&\quad i_{t_{j_{u}}}-i_{t_{j_{u-1}}}\big)I_{t_{j_{1}}}I_{t_{j_{2}}}\cdots I_{t_{j_{u}}}\prod_{\ell=1}^{u}(q^{{\rm gcd}(k_{t_{j_{\ell}}},r)}-1)\\
				&=\frac{1}{mn(q-1)}\sum_{r=0}^{m-1}{\rm gcd}\big(n,i_{t_{j_{1}}}(q-1),\cdots,i_{t_{j_{u}}}(q-1),i_{t_{j_{2}}}-i_{t_{j_{1}}},\cdots,i_{t_{j_{u}}}-i_{t_{j_{1}}},\\
				&\quad \cdots,i_{t_{j_{u}}}-i_{t_{j_{u-1}}}\big)\prod\limits_{\ell=1}^{u}(q^{k_{t_{j_{\ell}}}}-1)\\
				&=\frac{1}{n(q-1)}{\rm gcd}\big(n,i_{t_{j_{1}}}(q-1),\cdots,i_{t_{j_{u}}}(q-1),i_{t_{j_{2}}}-i_{t_{j_{1}}},\cdots,i_{t_{j_{u}}}-i_{t_{j_{1}}},\cdots,\\
				&\quad i_{t_{j_{u}}}-i_{t_{j_{u-1}}}\big)\prod\limits_{\ell=1}^{u}(q^{k_{t_{j_{\ell}}}}-1)\\
				&=\big|\langle\rho,\sigma_{\xi} \rangle\backslash \mathcal{C}_{j_{1},j_{2},\cdots, j_{u}}^{\sharp}\big|,
			\end{align*}
			which gives the desired result.}
	\end{Remark}
	
	The upper bound presented in Theorem \ref{t3.2} seems a little cumbersome. We next study a simple case, where we can provide an explicit formula for the number of orbits of $\langle\mu_{q},\rho,\sigma_{\xi} \rangle$ on $\mathcal{C}^{*}=\mathcal{C}\backslash \{\bf 0\}$.
	
	\begin{Theorem}\label{t3.3}
		Let $\mathcal{C}$ be a cyclic code of length $n$ over $\mathbb{F}_{q}$. Suppose that
		$$\mathcal{C}=\mathcal{R}_{n}\varepsilon_{t_{1}}\bigoplus\mathcal{R}_{n}\varepsilon_{t_{2}},$$
		where $0\leq t_{1},t_{2}\leq s$, and the primitive idempotent $\varepsilon_{t_{\ell}}$ corresponds to the $q$-cyclotomic coset $\{i_{t_{\ell}},i_{t_{\ell}}q,\cdots,i_{t_{\ell}}q^{k_{t_{\ell}}-1}\}$ for $\ell=1,2$. Suppose $k_{t_{1}}|k_{t_{2}}$, then the number of orbits of $\langle\mu_{q},\rho,\sigma_{\xi} \rangle$ on $\mathcal{C}^{*}=\mathcal{C}\backslash \{\bf 0\}$ is equal to
		$$s_{t_{1}}+s_{t_{2}}+s_{t_{1},t_{2}},$$
		where
		$$s_{t_{\ell}}=\frac{1}{k_{t_{\ell}}}\sum_{r\mid k_{t_{\ell}}}\varphi(\frac{k_{t_{\ell}}}{r}){\rm gcd}\big(q^{r}-1,\frac{q^{k_{t_{\ell}}}-1}{q-1},\frac{i_{t_{\ell}}(q^{k_{t_{\ell}}}-1)}{n}\big)~~{\rm for}~\ell=1,2,$$
		and
		\begin{align*}
			s_{t_{1},t_{2}}&=\frac{1}{m}\sum_{r=0}^{m-1}{\rm gcd}\Big((q^{{\rm gcd}(k_{t_{1}},r)}-1){\rm gcd}\big(q^{{\rm gcd}(k_{t_{2}},r)}-1,\frac{(q^{k_{t_{1}}}-1)(q^{{\rm gcd}(k_{t_{2}},r)}-1)}{(q-1)(q^{{\rm gcd}(k_{t_{1}},r)}-1)},\\
			&\quad~\frac{i_{t_{1}}(q^{k_{t_{1}}}-1)(q^{{\rm gcd}(k_{t_{2}},r)}-1)}{n(q^{{\rm gcd}(k_{t_{1}},r)}-1)},\frac{i_{t_{2}}(q^{k_{t_{2}}}-1)}{n}\big), \frac{(i_{t_{2}}-i_{t_{1}})(q^{k_{t_{1}}}-1)(q^{k_{t_{2}}}-1)}{n(q-1)}\Big).
		\end{align*}
		In particular, the number of non-zero weights of $\mathcal{C}$ is less than or equal to the number of orbits of $\langle\mu_{q},\rho,\sigma_{\xi} \rangle$ on $\mathcal{C}^{*}$, with equality if and only if for any two codewords ${\bf c}_{1},{\bf c}_{2}\in \mathcal{C}^{*}$ with the same weight, there exist integers $j_{1}$, $j_{2}$ and $ j_{3}$ such that $\mu_{q}^{j_{1}}\rho^{j_{2}}(\xi^{j_{3}}{\bf c}_{1})={\bf c}_{2}$.
	\end{Theorem}
	
	\begin{proof}
		Let $s_{t_{1}}$, $s_{t_{2}}$ and $s_{t_{1},t_{2}}$ be the number of orbits of $\langle\mu_{q},\rho,\sigma_{\xi} \rangle$ on $\mathcal{R}_{n}\varepsilon_{t_{1}}\backslash \{\bf 0\}$, $\mathcal{R}_{n}\varepsilon_{t_{2}}\backslash \{\bf 0\}$ and $\mathcal{R}_{n}\varepsilon_{t_{1}}\backslash \{{\bf 0}\}\bigoplus\mathcal{R}_{n}\varepsilon_{t_{2}}\backslash \{{\bf 0}\}$, respectively. Then
		$$\big|\langle\mu_{q},\rho,\sigma_{\xi} \rangle\backslash \mathcal{C}^{*}\big|=s_{t_{1}}+s_{t_{2}}+s_{t_{1},t_{2}}.$$
		It follows from Theorem \ref{t3.1} that
		$$s_{t_{\ell}}=\frac{1}{k_{t_{\ell}}}\sum_{r\mid k_{t_{\ell}}}\varphi(\frac{k_{t_{\ell}}}{r}){\rm gcd}\big(q^{r}-1,\frac{q^{k_{t_{\ell}}}-1}{q-1},\frac{i_{t}(q^{k_{t_{\ell}}}-1)}{n}\big)~~{\rm for}~\ell=1,2.$$
		According to the proof of Theorem \ref{t3.2}, we have
		$$s_{t_{1},t_{2}}=\frac{1}{mn(q-1)}\sum_{r_{1}=0}^{m-1}\sum_{r_{2}\in S(r_{1})}\sum_{r_{3}\in R(r_{1},r_{2})}\prod_{\ell=1}^{2}(q^{{\rm gcd}(k_{t_{\ell}},r_{1})}-1),$$
		where for $0\leq r_{1}\leq m-1$,
		$$S(r_{1})=\Big\{0\leq r_{2}\leq n-1~\Big|~\exists~ 0\leq r_{3}\leq q-2 ~s.t.~ \zeta^{-i_{t_{\ell}}q^{r_{1}}r_{2}}\in \xi^{r_{3}}\langle \theta_{t_{\ell}}^{q^{r_{1}}-1}  \rangle ~{\rm for~all~}1\leq \ell \leq2\Big\},$$
		and for $r_{2}\in S(r_{1})$,
		$$R(r_{1},r_{2})=\Big\{0\leq r_{3}\leq q-2~\Big|~\zeta^{-i_{t_{\ell}}q^{r_{1}}r_{2}}\in \xi^{r_{3}}\langle \theta_{t_{\ell}}^{q^{r_{1}}-1}  \rangle ~{\rm for~all~}1\leq \ell \leq2\Big\}.$$
		
		Since $k_{t_{1}}|k_{t_{2}}$, we have $(q^{k_{t_{1}}}-1)|(q^{k_{t_{2}}}-1)$, then $\mathbb{F}_{q^{k_{t_{1}}}}^{*}\leq \mathbb{F}_{q^{k_{t_{2}}}}^{*}$, that is, $\mathbb{F}_{q^{k_{t_{1}}}}^{*}$ is a subgroup of $\mathbb{F}_{q^{k_{t_{2}}}}^{*}$. Recall that $\mathbb{F}_{q^{k_{t_{1}}}}^{*}=\langle  \theta_{t_{1}} \rangle$ and $\mathbb{F}_{q^{k_{t_{2}}}}^{*}=\langle  \theta_{t_{2}} \rangle$, then $\theta_{t_{1}}=\theta_{t_{2}}^{i}$ for some nonnegative integer $i$. Thus $\langle \theta_{t_{1}}^{q^{r_{1}}-1}  \rangle=\langle \theta_{t_{2}}^{i(q^{r_{1}}-1)}  \rangle\leq \langle \theta_{t_{2}}^{q^{r_{1}}-1}  \rangle$, and so $|\langle \theta_{t_{1}}^{q^{r_{1}}-1}  \rangle|\big||\langle \theta_{t_{2}}^{q^{r_{1}}-1}  \rangle|$. Then we obtain
		$$|R(r_{1},r_{2})|={\rm gcd}\big(q-1,|\langle \theta_{t_{1}}^{q^{r_{1}}-1}\rangle|,|\langle \theta_{t_{2}}^{q^{r_{1}}-1}\rangle|\big)={\rm gcd}\big(q-1,|\langle \theta_{t_{1}}^{q^{r_{1}}-1}\rangle|\big).$$
		Next we calculate the cardinality of $S(r_{1})$.
		
		We claim that
		$$r_{2}\in S(r_{1})~\Leftrightarrow~\zeta^{-i_{t_{1}}q^{r_{1}}r_{2}}\in \langle\xi\rangle\langle \theta_{t_{1}}^{q^{r_{1}}-1}  \rangle ~{\rm and}~\zeta^{-(i_{t_{2}}-i_{t_{1}})q^{r_{1}}r_{2}}\in \langle \theta_{t_{2}}^{q^{r_{1}}-1}  \rangle.$$
		Suppose $r_{2}\in S(r_{1})$, that is, there exists some $0\leq r_{3}\leq q-2$ such that $\zeta^{-i_{t_{1}}q^{r_{1}}r_{2}}\in \xi^{r_{3}}\langle \theta_{t_{1}}^{q^{r_{1}}-1}  \rangle$ and $\zeta^{-i_{t_{2}}q^{r_{1}}r_{2}}\in \xi^{r_{3}}\langle \theta_{t_{2}}^{q^{r_{1}}-1}  \rangle$, then $\zeta^{-i_{t_{1}}q^{r_{1}}r_{2}}\in \langle\xi\rangle\langle \theta_{t_{1}}^{q^{r_{1}}-1}  \rangle$ and
		$$\zeta^{-(i_{t_{2}}-i_{t_{1}})q^{r_{1}}r_{2}}\in \langle \theta_{t_{1}}^{q^{r_{1}}-1}  \rangle\langle \theta_{t_{2}}^{q^{r_{1}}-1}  \rangle= \langle \theta_{t_{2}}^{q^{r_{1}}-1}  \rangle.$$
		For the converse, if $\zeta^{-i_{t_{1}}q^{r_{1}}r_{2}}\in \langle\xi\rangle\langle \theta_{t_{1}}^{q^{r_{1}}-1}  \rangle ~{\rm and}~\zeta^{-(i_{t_{2}}-i_{t_{1}})q^{r_{1}}r_{2}}\in \langle \theta_{t_{2}}^{q^{r_{1}}-1}  \rangle$, then $\zeta^{-i_{t_{1}}q^{r_{1}}r_{2}}\in \xi^{r_{3}}\langle \theta_{t_{1}}^{q^{r_{1}}-1}  \rangle$ for some $0\leq r_{3}\leq q-2$, and
		$$\zeta^{-i_{t_{2}}q^{r_{1}}r_{2}}\in \zeta^{-i_{t_{1}}q^{r_{1}}r_{2}}\langle \theta_{t_{2}}^{q^{r_{1}}-1}\rangle\subseteq \xi^{r_{3}}\langle \theta_{t_{1}}^{q^{r_{1}}-1}  \rangle\langle \theta_{t_{2}}^{q^{r_{1}}-1}\rangle=\xi^{r_{3}}\langle \theta_{t_{2}}^{q^{r_{1}}-1}\rangle,$$
		and so $r_{2}\in S(r_{1})$. Consequently, we have
		\begin{align*}
			r_{2}\in S(r_{1})~&\Leftrightarrow~\frac{n}{{\rm gcd}(n,i_{t_{1}}r_{2})}\Big|\frac{(q-1)|\langle \theta_{t_{1}}^{q^{r_{1}}-1}\rangle|}{{\rm gcd}\big(q-1,|\langle \theta_{t_{1}}^{q^{r_{1}}-1}\rangle|\big)}~{\rm and}~\frac{n}{{\rm gcd}\big(n,(i_{t_{2}}-i_{t_{1}})r_{2}\big)}\Big||\langle \theta_{t_{2}}^{q^{r_{1}}-1}\rangle|\\
			&\Leftrightarrow~\frac{n}{{\rm gcd}\big(n,\frac{i_{t_{1}}(q-1)|\langle \theta_{t_{1}}^{q^{r_{1}}-1}\rangle|}{{\rm gcd}(q-1,|\langle \theta_{t_{1}}^{q^{r_{1}}-1}\rangle|)}\big)}\Biggl|r_{2}~{\rm and}~\frac{n}{{\rm gcd}\big(n,(i_{t_{2}}-i_{t_{1}})|\langle \theta_{t_{2}}^{q^{r_{1}}-1}\rangle|\big)}\Biggl|r_{2}\\
			&\Leftrightarrow~\frac{n}{{\rm gcd}\big(n,\frac{i_{t_{1}}(q-1)|\langle \theta_{t_{1}}^{q^{r_{1}}-1}\rangle|}{{\rm gcd}(q-1,|\langle \theta_{t_{1}}^{q^{r_{1}}-1}\rangle|)},(i_{t_{2}}-i_{t_{1}})|\langle \theta_{t_{2}}^{q^{r_{1}}-1}\rangle|\big)}\Biggl|r_{2}.
		\end{align*}
		Hence
		$$|S(r_{1})|={\rm gcd}\big(n,\frac{i_{t_{1}}(q-1)|\langle \theta_{t_{1}}^{q^{r_{1}}-1}\rangle|}{{\rm gcd}(q-1,|\langle \theta_{t_{1}}^{q^{r_{1}}-1}\rangle|)},(i_{t_{2}}-i_{t_{1}})|\langle \theta_{t_{2}}^{q^{r_{1}}-1}\rangle|\big).$$
		
		We then conclude that
		\begin{align*}
			&s_{t_{1},t_{2}}\\
			=&\frac{1}{mn(q-1)}\sum_{r_{1}=0}^{m-1}|S(r_{1})||R(r_{1},r_{2})|\prod_{\ell=1}^{2}(q^{{\rm gcd}(k_{t_{\ell}},r_{1})}-1)\\
			=&\frac{1}{mn(q-1)}\sum_{r_{1}=0}^{m-1}{\rm gcd}\big(n,\frac{i_{t_{1}}(q-1)|\langle \theta_{t_{1}}^{q^{r_{1}}-1}\rangle|}{{\rm gcd}(q-1,|\langle \theta_{t_{1}}^{q^{r_{1}}-1}\rangle|)},(i_{t_{2}}-i_{t_{1}})|\langle \theta_{t_{2}}^{q^{r_{1}}-1}\rangle|\big){\rm gcd}\big(q-1,|\langle \theta_{t_{1}}^{q^{r_{1}}-1}\rangle|\big)\\
			&\cdot \prod_{\ell=1}^{2}(q^{{\rm gcd}(k_{t_{\ell}},r_{1})}-1)\\
			=&\frac{1}{mn(q-1)}\sum_{r_{1}=0}^{m-1}{\rm gcd}\big(n(q-1),n|\langle \theta_{t_{1}}^{q^{r_{1}}-1}\rangle|,i_{t_{1}}(q-1)|\langle \theta_{t_{1}}^{q^{r_{1}}-1}\rangle|,(i_{t_{2}}-i_{t_{1}})(q-1)|\langle \theta_{t_{2}}^{q^{r_{1}}-1}\rangle|,\\
			&(i_{t_{2}}-i_{t_{1}})|\langle \theta_{t_{1}}^{q^{r_{1}}-1}\rangle||\langle \theta_{t_{2}}^{q^{r_{1}}-1}\rangle|\big)\prod_{\ell=1}^{2}(q^{{\rm gcd}(k_{t_{\ell}},r_{1})}-1)\\
			=&\frac{1}{mn(q-1)}\sum_{r_{1}=0}^{m-1}{\rm gcd}\big(n(q-1),n|\langle \theta_{t_{1}}^{q^{r_{1}}-1}\rangle|,i_{t_{1}}(q-1)|\langle \theta_{t_{1}}^{q^{r_{1}}-1}\rangle|,i_{t_{2}}(q-1)|\langle \theta_{t_{2}}^{q^{r_{1}}-1}\rangle|,\\
			&(i_{t_{2}}-i_{t_{1}})|\langle \theta_{t_{1}}^{q^{r_{1}}-1}\rangle||\langle \theta_{t_{2}}^{q^{r_{1}}-1}\rangle|\big)\prod_{\ell=1}^{2}(q^{{\rm gcd}(k_{t_{\ell}},r_{1})}-1)\\
			=&\frac{1}{m}\sum_{r_{1}=0}^{m-1}{\rm gcd}\big(\prod_{\ell=1}^{2}(q^{{\rm gcd}(k_{t_{\ell}},r_{1})}-1),\frac{(q^{k_{t_{1}}}-1)(q^{{\rm gcd}(k_{t_{2}},r_{1})}-1)}{q-1},\frac{i_{t_{1}}(q^{k_{t_{1}}}-1)(q^{{\rm gcd}(k_{t_{2}},r_{1})}-1)}{n},\\
			&\frac{i_{t_{2}}(q^{{\rm gcd}(k_{t_{1}},r_{1})}-1)(q^{k_{t_{2}}}-1)}{n}, \frac{(i_{t_{2}}-i_{t_{1}})(q^{k_{t_{1}}}-1)(q^{k_{t_{2}}}-1)}{n(q-1)}\big)\\
			=&\frac{1}{m}\sum_{r_{1}=0}^{m-1}{\rm gcd}\Big((q^{{\rm gcd}(k_{t_{1}},r_{1})}-1){\rm gcd}\big(q^{{\rm gcd}(k_{t_{2}},r_{1})}-1,\frac{(q^{k_{t_{1}}}-1)(q^{{\rm gcd}(k_{t_{2}},r_{1})}-1)}{(q-1)(q^{{\rm gcd}(k_{t_{1}},r_{1})}-1)},\\
			&\frac{i_{t_{1}}(q^{k_{t_{1}}}-1)(q^{{\rm gcd}(k_{t_{2}},r_{1})}-1)}{n(q^{{\rm gcd}(k_{t_{1}},r_{1})}-1)},\frac{i_{t_{2}}(q^{k_{t_{2}}}-1)}{n}\big), \frac{(i_{t_{2}}-i_{t_{1}})(q^{k_{t_{1}}}-1)(q^{k_{t_{2}}}-1)}{n(q-1)}\Big).
		\end{align*}
		We are done.
	\end{proof}

	As a direct application of Theorem \ref{t3.3}, we immediately obtain the following two corollaries.
	
	\begin{Corollary}\label{c3.3}
		Let $\mathcal{C}$ be a cyclic code of length $n$ over $\mathbb{F}_{q}$. Suppose that
		$$\mathcal{C}=\mathcal{R}_{n}\varepsilon_{t_{1}}\bigoplus\mathcal{R}_{n}\varepsilon_{t_{2}},$$
		where $0\leq t_{1},t_{2}\leq s$, and the primitive idempotent $\varepsilon_{t_{\ell}}$ corresponds to the $q$-cyclotomic coset $\{i_{t_{\ell}},i_{t_{\ell}}q,\cdots,i_{t_{\ell}}q^{k_{t_{\ell}}-1}\}$ for $\ell=1,2$. Suppose $k_{t_{1}}=1$ and $k_{t_{2}}=k$, then the number of orbits of $\langle\mu_{q},\rho,\sigma_{\xi} \rangle$ on $\mathcal{C}^{*}=\mathcal{C}\backslash \{\bf 0\}$ is equal to
		$$1+\frac{1}{k}\sum_{r\mid k}\varphi(\frac{k}{r})\Big( {\rm gcd}\big(q^{r}-1,\frac{q^{k}-1}{q-1},\frac{i_{t_{2}}(q^{k}-1)}{n}\big)+{\rm gcd}\big(q^{r}-1,\frac{(i_{t_{2}}-i_{t_{1}})(q^{k}-1)}{n}\big)\Big) .$$
	\end{Corollary}
	
	\begin{proof}
		According to Theorem \ref{t3.3} and its proof, we see that
		$$\big|\langle\mu_{q},\rho,\sigma_{\xi} \rangle\backslash \mathcal{C}^{*}\big|=s_{t_{1}}+s_{t_{2}}+s_{t_{1},t_{2}},$$
		where
		\begin{align*}
			&s_{t_{1}}={\rm gcd}\big(q-1,1,\frac{i_{t_{1}}(q-1)}{n}\big)=1,\\
			&s_{t_{2}}=\frac{1}{k}\sum_{r\mid k}\varphi(\frac{k}{r}){\rm gcd}\big(q^{r}-1,\frac{q^{k}-1}{q-1},\frac{i_{t_{2}}(q^{k}-1)}{n}\big),
		\end{align*}
		and
		\begin{align*}
			&s_{t_{1},t_{2}}\\
			=&\frac{1}{m}\sum_{r_{1}=0}^{m-1}{\rm gcd}\Big((q-1){\rm gcd}\big(q^{{\rm gcd}(k,r_{1})}-1,\frac{q^{{\rm gcd}(k,r_{1})}-1}{q-1},\frac{i_{t_{1}}(q^{{\rm gcd}(k,r_{1})}-1)}{n},\frac{i_{t_{2}}(q^{k}-1)}{n}\big),\\
			&\frac{(i_{t_{2}}-i_{t_{1}})(q^{k}-1)}{n}\Big)\\
			=&\frac{1}{m}\sum_{r_{1}=0}^{m-1}{\rm gcd}\big(q^{{\rm gcd}(k,r_{1})}-1,\frac{i_{t_{1}}(q-1)(q^{{\rm gcd}(k,r_{1})}-1)}{n},\frac{i_{t_{2}}(q-1)(q^{k}-1)}{n},\frac{(i_{t_{2}}-i_{t_{1}})(q^{k}-1)}{n}\big)\\
			=&\frac{1}{m}\sum_{r_{1}=0}^{m-1}{\rm gcd}\big(q^{{\rm gcd}(k,r_{1})}-1,\frac{i_{t_{1}}(q-1)(q^{{\rm gcd}(k,r_{1})}-1)}{n},\frac{(i_{t_{2}}-i_{t_{1}})(q-1)(q^{k}-1)}{n},\\
			&\frac{(i_{t_{2}}-i_{t_{1}})(q^{k}-1)}{n}\big)\\
			=&\frac{1}{m}\sum_{r_{1}=0}^{m-1}{\rm gcd}\big(q^{{\rm gcd}(k,r_{1})}-1,\frac{(i_{t_{2}}-i_{t_{1}})(q^{k}-1)}{n}\big)\\
			=&\frac{1}{k}\sum_{r|k}\varphi(\frac{k}{r}){\rm gcd}\big(q^{r}-1,\frac{(i_{t_{2}}-i_{t_{1}})(q^{k}-1)}{n}\big),
		\end{align*}
		which gives the desired result.
	\end{proof}
	
	We present an example to illustrate that the upper bound in Corollary \ref{c3.3} improves the upper bound in \cite[Corollary III.18]{7}.
	
	\begin{Example}\label{e3.4}{\rm 
			Take $q=2$ and $n=15$. All the distinct $2$-cyclotomic cosets modulo $15$ are given by
			$$\Gamma_{0}=\{0\},~\Gamma_{1}=\{1,2,4,8\},~\Gamma_{2}=\{3,6,12,9\},$$
			$$\Gamma_{3}=\{5,10\},~\Gamma_{4}=\{7,14,13,11\}.$$
			Let $\ell$ be the number of non-zero weights of the cyclic code $\mathcal{C}=\mathcal{R}_{n}\varepsilon_{0}\bigoplus\mathcal{R}_{n}\varepsilon_{2}$. By \cite[Corollary III.18]{7}, we have
			\begin{align*}
				\ell&\leq \frac{{\rm gcd}(15,0,3)(2-1)(2^{4}-1)}{15(2-1)}{\rm gcd}\big(2-1,\frac{15}{{\rm gcd}(15,0)},\frac{15}{{\rm gcd}(15,3)}\big)\\
				&\quad +\frac{{\rm gcd}(15,0)(2-1)}{15(2-1)}{\rm gcd}\big(2-1,\frac{15}{{\rm gcd}(15,0)}\big)+\frac{{\rm gcd}(15,3)(2^{4}-1)}{15(2-1)}{\rm gcd}\big(2-1,\frac{15}{{\rm gcd}(15,3)}\big)\\
				&=3+1+3=7.
			\end{align*}
			Using Corollary \ref{c3.3}, we have
			\begin{align*}
				\ell&\leq 1+\frac{1}{4}\sum_{r\mid 4}\varphi(\frac{4}{r})\Big({\rm gcd}\big(2^{r}-1,\frac{2^{4}-1}{2-1},\frac{3(2^{4}-1)}{15}\big)+{\rm gcd}\big(2^{r}-1,\frac{(3-0)(2^{4}-1)}{15}\big)\Big) \\
				&=1+\frac{1}{4}[2\varphi(4)+6\varphi(2)+6\varphi(1)]\\
				&=1+\frac{1}{4}(4+6+6)=5.
			\end{align*}
			After using Magma \cite{4}, we know that the weight distribution of $\mathcal{C}$ is $1+5x^{3}+10x^{6}+10x^{9}+5x^{12}+x^{15}$, which implies that the exact value of $\ell$ is 5. Therefore, any two codewords of $\mathcal{C}$ with the same weight are in the same $\langle\mu_{q},\rho,\sigma_{\xi} \rangle$-orbit.}
	\end{Example}
	
	\begin{Corollary}\label{c3.4}
		Let $\mathcal{C}$ be a cyclic code of length $n$ over $\mathbb{F}_{q}$. Suppose that
		$$\mathcal{C}=\mathcal{R}_{n}\varepsilon_{t_{1}}\bigoplus\mathcal{R}_{n}\varepsilon_{t_{2}},$$
		where $0\leq t_{1},t_{2}\leq s$, and the primitive idempotent $\varepsilon_{t_{\ell}}$ corresponds to the $q$-cyclotomic coset $\{i_{t_{\ell}},i_{t_{\ell}}q,\cdots,i_{t_{\ell}}q^{k_{t_{\ell}}-1}\}$ for $\ell=1,2$. Suppose $k_{t_{1}}=k_{t_{2}}=k$, then the number of orbits of $\langle\mu_{q},\rho,\sigma_{\xi} \rangle$ on $\mathcal{C}^{*}=\mathcal{C}\backslash \{\bf 0\}$ is equal to
		\begin{align*}
			&\frac{1}{k}\sum_{r\mid k}\varphi(\frac{k}{r})\left( {\rm gcd}\big(q^{r}-1,\frac{q^{k}-1}{q-1},\frac{i_{t_{1}}(q^{k}-1)}{n}\big)+ {\rm gcd}\big(q^{r}-1,\frac{q^{k}-1}{q-1},\frac{i_{t_{2}}(q^{k}-1)}{n}\big)\right. \\
			&\quad \left. +{\rm gcd}\Big((q^{r}-1){\rm gcd}\big(q^{r}-1,\frac{q^{k}-1}{q-1},\frac{i_{t_{1}}(q^{k}-1)}{n},\frac{i_{t_{2}}(q^{k}-1)}{n}\big),\frac{(i_{t_{2}}-i_{t_{1}})(q^{k}-1)^{2}}{n(q-1)}\Big) \right).
		\end{align*}
	\end{Corollary}
	
	\begin{proof}
		According to Theorem \ref{t3.3} and its proof, we see that
		$$\big|\langle\mu_{q},\rho,\sigma_{\xi} \rangle\backslash \mathcal{C}^{*}\big|=s_{t_{1}}+s_{t_{2}}+s_{t_{1},t_{2}},$$
		where
		$$s_{t_{\ell}}=\frac{1}{k}\sum_{r\mid k}\varphi(\frac{k}{r}){\rm gcd}\big(q^{r}-1,\frac{q^{k}-1}{q-1},\frac{i_{t_{\ell}}(q^{k}-1)}{n}\big)~~{\rm for}~\ell=1,2,$$
		and
		\begin{align*}
			s_{t_{1},t_{2}}&=\frac{1}{m}\sum_{r_{1}=0}^{m-1}{\rm gcd}\Big((q^{{\rm gcd}(k,r_{1})}-1){\rm gcd}\big(q^{{\rm gcd}(k,r_{1})}-1,\frac{q^{k}-1}{q-1},\frac{i_{t_{1}}(q^{k}-1)}{n},\frac{i_{t_{2}}(q^{k}-1)}{n}\big), \\
			&\quad~\frac{(i_{t_{2}}-i_{t_{1}})(q^{k}-1)^{2}}{n(q-1)}\Big)\\
			&=\frac{1}{k}\sum_{r|k}\varphi(\frac{k}{r}){\rm gcd}\Big((q^{r}-1){\rm gcd}\big(q^{r}-1,\frac{q^{k}-1}{q-1},\frac{i_{t_{1}}(q^{k}-1)}{n},\frac{i_{t_{2}}(q^{k}-1)}{n}\big),\\
			&\quad~\frac{(i_{t_{2}}-i_{t_{1}})(q^{k}-1)^{2}}{n(q-1)}\Big),
		\end{align*}
		which gives the desired result.
	\end{proof}
	
	\begin{Example}{\rm
			Take $q=2$ and $n=15$. All the distinct $2$-cyclotomic cosets modulo $15$ are as shown in Example \ref{e3.4}.
			Let $\ell$ be the number of non-zero weights of the cyclic code $\mathcal{C}=\mathcal{R}_{n}\varepsilon_{1}\bigoplus\mathcal{R}_{n}\varepsilon_{2}$. By \cite[Corollary III.18]{7}, we have
			\begin{align*}
				\ell&\leq \frac{{\rm gcd}(15,1,3)(2^{4}-1)(2^{4}-1)}{15(2-1)}{\rm gcd}\big(2-1,\frac{15}{{\rm gcd}(15,1)},\frac{15}{{\rm gcd}(15,3)}\big)\\
				&\quad +\frac{{\rm gcd}(15,1)(2^{4}-1)}{15(2-1)}{\rm gcd}\big(2-1,\frac{15}{{\rm gcd}(15,1)}\big)+\frac{{\rm gcd}(15,3)(2^{4}-1)}{15(2-1)}{\rm gcd}\big(2-1,\frac{15}{{\rm gcd}(15,3)}\big)\\
				&=15+1+3=19.
			\end{align*}
			Using Corollary \ref{c3.4}, we have
			\begin{align*}
				\ell&\leq \frac{1}{4}\sum_{r\mid 4}\varphi(\frac{4}{r})\left( {\rm gcd}(2^{r}-1,\frac{2^{4}-1}{2-1},\frac{2^{4}-1}{15})+ {\rm gcd}(2^{r}-1,\frac{2^{4}-1}{2-1},\frac{3(2^{4}-1)}{15})\right.  \\
				&\quad +\left. {\rm gcd}\Big((2^{r}-1){\rm gcd}\big(2^{r}-1,\frac{2^{4}-1}{2-1},\frac{2^{4}-1}{15},\frac{3(2^{4}-1)}{15}\big),\frac{(3-1)(2^{4}-1)^{2}}{15(2-1)}\Big) \right) \\
				&=\frac{1}{4}[3\varphi(4)+7\varphi(2)+19\varphi(1)]\\
				&=\frac{1}{4}(6+7+19)=8.
			\end{align*}
			After using Magma \cite{4}, we know that the weight distribution of $\mathcal{C}$ is $1+15x^{4}+100x^{6}+75x^{8}+60x^{10}+5x^{12}$, which implies that $\ell=5$.}
	\end{Example}
	
	\subsection{New upper bound on the number of non-zero weights of the cyclic code $\mathcal{C}=\mathcal{R}_{n}\varepsilon_{t}\bigoplus \mu_{-1}(\mathcal{R}_{n}\varepsilon_{t})$}
	
	For $0<t\leq s$, suppose that the irreducible cyclic code $\mathcal{R}_{n}\varepsilon_{t}$ corresponds to the $q$-cyclotomic coset $\{i_{t},i_{t}q,\cdots,i_{t}q^{k-1}\}$, then
	$$\mathcal{R}_{n}\varepsilon_{t}=\Big\{\sum_{j=0}^{k-1}\big(c_{0}+c_{1}\zeta^{i_{t}q^{j}}+\cdots+c_{k-1}\zeta^{(k-1)i_{t}q^{j}}\big)e_{i_{t}q^{j}}~\Big|~c_{\ell}\in \mathbb{F}_{q}, 0\leq \ell \leq k-1\Big\}.$$
	Since $-1\in \mathbb{Z}_{n}^{*}$, $\mu_{-1}$ is an $\mathbb{F}_{q}$-vector space automorphism of $\mathcal{R}_{n}$. One can check that $\mu_{-1}(\mathcal{R}_{n}\varepsilon_{t})$ is also an irreducible cyclic code,
	and the primitive idempotent generating $\mu_{-1}(\mathcal{R}_{n}\varepsilon_{t})$ corresponds to the $q$-cyclotomic coset $\{-i_{t},-i_{t}q,\cdots,-i_{t}q^{k-1}\}$ (see \cite[Corollary 4.4.5]{11}). Therefore,
	$$\mu_{-1}(\mathcal{R}_{n}\varepsilon_{t})=\Big\{\sum_{j=0}^{k-1}\big(c_{0}'+c_{1}'\zeta^{-i_{t}q^{j}}+\cdots+c_{k-1}'\zeta^{-(k-1)i_{t}q^{j}}\big)e_{-i_{t}q^{j}}~\Big|~c_{\ell}'\in \mathbb{F}_{q}, 0\leq \ell \leq k-1\Big\}.$$
	
	Suppose $-i_{t}\notin \{i_{t}, i_{t}q,\cdots,i_{t}q^{k-1}\}$, then $\mu_{-1}(\mathcal{R}_{n}\varepsilon_{t})\cap \mathcal{R}_{n}\varepsilon_{t}=\{{\bf 0}\}$ and $\mu_{-1}^{2}(\mathcal{R}_{n}\varepsilon_{t})=\mathcal{R}_{n}\varepsilon_{t}$. Let
	$$\mathcal{C}=\mathcal{R}_{n}\varepsilon_{t}\bigoplus \mu_{-1}(\mathcal{R}_{n}\varepsilon_{t}).$$
	It is easy to see that $\mu_{-1}\in {\rm Aut}(\mathcal{C})$. In addition, since $-q\in \mathbb{Z}_{n}^{*}$, $\mu_{-q}$ is an $\mathbb{F}_{q}$-vector space automorphism of $\mathcal{R}_{n}$ and,
	$$\mu_{-q}(\mathcal{C})=\mu_{-q}(\mathcal{R}_{n}\varepsilon_{t})\bigoplus \mu_{q}(\mathcal{R}_{n}\varepsilon_{t})=\mu_{-q}(\mathcal{R}_{n}\varepsilon_{t})\bigoplus \mathcal{R}_{n}\varepsilon_{t}.$$
	Note that the primitive idempotent generating $\mu_{-q}(\mathcal{R}_{n}\varepsilon_{t})$ corresponds to the $q$-cyclotomic coset $$(-q)^{-1}\{i_{t},i_{t}q,\cdots,i_{t}q^{k-1}\}=-q^{m-1}\{i_{t},i_{t}q,\cdots,i_{t}q^{k-1}\}=\{-i_{t},-i_{t}q,\cdots,-i_{t}q^{k-1}\},$$
	and hence $\mu_{-q}(\mathcal{R}_{n}\varepsilon_{t})=\mu_{-1}(\mathcal{R}_{n}\varepsilon_{t})$. So  $\mu_{-q}\in {\rm Aut}(\mathcal{C})$. 
	
	The subgroup $\langle \mu_{-1}\rangle$ of ${\rm Aut}(\mathcal{C})$ generated by $\mu_{-1}$ is of order $2$. Let $m'$ denote the order of $-q$ in $\mathbb{Z}_{n}^{*}$, then the subgroup $\langle \mu_{-q}\rangle$ of ${\rm Aut}(\mathcal{C})$ generated by $\mu_{-q}$ is of order $m'$.  It is easy to check that either $\langle \mu_{-1}\rangle \subseteq \langle \mu_{-q}\rangle$ or $\langle \mu_{-1}\rangle \cap \langle \mu_{-q}\rangle=id$, and $\langle \mu_{-1}\rangle \subseteq \langle \mu_{-q}\rangle$ if and only if $-1\in \langle-q \rangle_{_{\mathbb{Z}_{n}^{*}}}$, where $\langle -q \rangle_{_{\mathbb{Z}_{n}^{*}}}$ is the subgroup of $\mathbb{Z}_{n}^{*}$ generated by $-q$.
	
	In this subsection, we first study the action of $\langle\mu_{-q},\rho,\sigma_{\xi} \rangle$ on $\mathcal{C}^{*}=\mathcal{C}\backslash \{\bf 0\}$ when $-1\in \langle-q \rangle_{_{\mathbb{Z}_{n}^{*}}}$, and then study the action of $\langle\mu_{-1},\mu_{q},\rho,\sigma_{\xi} \rangle$ on $\mathcal{C}^{*}=\mathcal{C}\backslash \{\bf 0\}$ when $-1\not\in \langle-q \rangle_{_{\mathbb{Z}_{n}^{*}}}$.  Since $\mu_{q}=\mu_{-1}\mu_{-q}$, $\langle\mu_{q},\rho,\sigma_{\xi} \rangle$ is a subgroup of both $\langle\mu_{-q},\rho,\sigma_{\xi} \rangle$ and $\langle\mu_{-1}, \mu_{q},\rho,\sigma_{\xi} \rangle$. Hence the number of orbits of $\langle\mu_{-q},\rho,\sigma_{\xi} \rangle$ or $\langle\mu_{-1},\mu_{q},\rho,\sigma_{\xi} \rangle$ on $\mathcal{C}^{*}$ is less than or equal to the number of orbits of $\langle\mu_{q},\rho,\sigma_{\xi} \rangle$ on $\mathcal{C}^{*}$. In the following we will show that the former is strictly less than the latter.

	\subsubsection{The action of $\langle\mu_{-q},\rho,\sigma_{\xi} \rangle$ on $\mathcal{C}^{*}$}
	Assume that $-1\in \langle-q \rangle_{_{\mathbb{Z}_{n}^{*}}}$ in this subsection, and we now consider the action of $\langle\mu_{-q},\rho,\sigma_{\xi} \rangle$ on $\mathcal{C}^{*}=\mathcal{C}\backslash \{\bf 0\}$, where $\mathcal{C}=\mathcal{R}_{n}\varepsilon_{t}\bigoplus \mu_{-1}(\mathcal{R}_{n}\varepsilon_{t})$. We first have the following lemma.
	
	\begin{Lemma}\label{L1}
		Let $m$ be the order of $q$ in $\mathbb{Z}_{n}^{*}$ and let $m'$ be the order of $-q$ in $\mathbb{Z}_{n}^{*}$. Suppose that $-i_{t}\notin \{i_{t}, i_{t}q,\cdots,i_{t}q^{k-1}\}$ and $-1\in \langle-q \rangle_{_{\mathbb{Z}_{n}^{*}}}$, then $m$ is odd and $m'=2m$.
	\end{Lemma}
	\begin{proof}
		If $-i_{t}\notin \{i_{t}, i_{t}q,\cdots,i_{t}q^{k-1}\}$, then $-1\not\in \langle q \rangle_{_{\mathbb{Z}_{n}^{*}}}$, where $\langle q \rangle_{_{\mathbb{Z}_{n}^{*}}}$ is the subgroup of $\mathbb{Z}_{n}^{*}$ generated by $q$, and so $m'$ is even. Since $-1\in \langle -q \rangle_{_{\mathbb{Z}_{n}^{*}}}$, $(-q)^{\ell} \equiv -1~({\rm mod}~n)$ for some positive integer $\ell$. We claim that $\ell$ is odd as $-1\not\in \langle q \rangle_{_{\mathbb{Z}_{n}^{*}}}$, then $q^{\ell} \equiv 1~({\rm mod}~n)$, implying $m|\ell$ and thus $m$ is odd. As $(-q)^{2m} \equiv 1~({\rm mod}~n)$, we have $m'|2m$, and so $m'=2l$ for some divisor $l$ of $m$, then $q^{2l}\equiv (-q)^{2l}\equiv 1~({\rm mod}~n)$, and hence $m|2l$. But $m$ is odd, then $l=m$ yielding $m'=2m$.
	\end{proof}
	
	Notice that the group $\langle\mu_{-q},\rho,\sigma_{\xi} \rangle$ can act on the sets $\mathcal{C}'=\big(\mathcal{R}_{n}\varepsilon_{t}\backslash \{{\bf 0}\}\big)\cup \big(\mu_{-1}(\mathcal{R}_{n}\varepsilon_{t})\backslash \{{\bf 0}\}\big)$ and $\mathcal{C}^{\sharp}=\mathcal{R}_{n}\varepsilon_{t}\backslash \{{\bf 0}\}\bigoplus \mu_{-1}(\mathcal{R}_{n}\varepsilon_{t})\backslash \{{\bf 0}\}$, respectively. The numbers of orbits of $\langle\mu_{-q},\rho,\sigma_{\xi} \rangle$ on $\mathcal{C}'$ and $\langle\mu_{-q},\rho,\sigma_{\xi} \rangle$ on $\mathcal{C}^{\sharp}$ are given as below.
	\begin{Lemma}\label{l3.2}
		With the notation given above, then the number of orbits of $\langle\mu_{-q},\rho,\sigma_{\xi} \rangle$ on $$\mathcal{C}'=\big(\mathcal{R}_{n}\varepsilon_{t}\backslash \{{\bf 0}\}\big)\cup \big(\mu_{-1}(\mathcal{R}_{n}\varepsilon_{t})\backslash \{{\bf 0}\}\big)$$ is equal to
		\begin{align*}
			\frac{1}{k}\sum_{r|k}\varphi(\frac{k}{r}){\rm gcd}\big(q^{r}-1,\frac{q^{k}-1}{q-1},\frac{i_{t}(q^{k}-1)}{n}\big).
		\end{align*}
	\end{Lemma}
	\begin{proof}
		Note that
		$$\mathcal{C}'=\big(\mathcal{R}_{n}\varepsilon_{t}\backslash \{{\bf 0}\}\big)\cup \big(\mu_{-1}(\mathcal{R}_{n}\varepsilon_{t})\backslash \{{\bf 0}\}\big)$$
		is a disjoint union. Then it follows from Equation (\ref{e2.1}), Lemmas \ref{l2.2} and \ref{L1} that
		\begin{align*}
			\big|\langle\mu_{-q},\rho,\sigma_{\xi} \rangle\backslash \mathcal{C}'\big|&=\frac{1}{2mn(q-1)}\Big( \sum_{r_{1}=0}^{2m-1}\sum_{r_{2}=0}^{n-1}\sum_{r_{3}=0}^{q-2}\Big|\big\{{\bf c}\in \mathcal{R}_{n}\varepsilon_{t}\backslash \{{\bf 0}\}~\big|~\mu_{-q}^{r_{1}}\rho^{r_{2}}\sigma_{\xi}^{r_{3}}({\bf c})={\bf c}\big\}\Big|\\
			&\quad + \sum_{r_{1}=0}^{2m-1}\sum_{r_{2}=0}^{n-1}\sum_{r_{3}=0}^{q-2}\Big|\big\{{\bf c}\in \mu_{-1}(\mathcal{R}_{n}\varepsilon_{t})\backslash \{{\bf 0}\}~\big|~\mu_{-q}^{r_{1}}\rho^{r_{2}}\sigma_{\xi}^{r_{3}}({\bf c})={\bf c}\big\}\Big|\Big) .
		\end{align*}
		
		Take ${\bf c}=\sum\limits_{j=0}^{k-1}\big(\sum\limits_{\ell=0}^{k-1}c_{\ell}\zeta^{\ell i_{t}q^{j}}\big)e_{i_{t}q^{j}}\in \mathcal{R}_{n}\varepsilon_{t}\backslash \{\bf 0\}$. Note that $e_{i_{t}q^{j}}=\frac{1}{n}\sum\limits_{l=0}^{n-1}\zeta^{-i_{t}q^{j}l}x^{l}$ and $\rho^{r_{2}}\sigma_{\xi}^{r_{3}}(e_{i_{t}q^{j}})=\xi^{r_{3}}\zeta^{i_{t}q^{j}r_{2}}e_{i_{t}q^{j}}$, and thus
		\begin{align*}
			\mu_{-q}^{r_{1}}\rho^{r_{2}}\sigma_{\xi}^{r_{3}}(e_{i_{t}q^{j}})&=\xi^{r_{3}}\zeta^{i_{t}q^{j}r_{2}}\mu_{-q}^{r_{1}}(e_{i_{t}q^{j}})\\
			&=\xi^{r_{3}}\zeta^{i_{t}q^{j}r_{2}}\cdot \frac{1}{n}\sum_{l=0}^{n-1}\zeta^{-i_{t}q^{j}l}x^{(-q)^{r_{1}}l}\\
			&=\xi^{r_{3}}\zeta^{i_{t}q^{j}r_{2}}\cdot \frac{1}{n}\sum_{l=0}^{n-1}\zeta^{-i_{t}q^{j}(-q)^{-r_{1}} (-q)^{r_{1}}l}x^{(-q)^{r_{1}}l}\\
			&=\xi^{r_{3}}\zeta^{i_{t}q^{j}r_{2}}\cdot \frac{1}{n}\sum_{l=0}^{n-1}\zeta^{-i_{t}q^{j}(-q)^{-r_{1}}l}x^{l}\\
			&=\xi^{r_{3}}\zeta^{i_{t}q^{j}r_{2}}\cdot \frac{1}{n}\sum_{l=0}^{n-1}\zeta^{-(-1)^{r_{1}}i_{t}q^{-r_{1}+j}l}x^{l}\\
			&=\xi^{r_{3}}\zeta^{i_{t}q^{j}r_{2}}e_{(-1)^{r_{1}}i_{t}q^{-r_{1}+j}}.
		\end{align*}
		We then have
		\begin{align*}
			\mu_{-q}^{r_{1}}\rho^{r_{2}}\sigma_{\xi}^{r_{3}}({\bf c})&=\sum_{j=0}^{k-1}\big(\sum\limits_{\ell=0}^{k-1}c_{\ell}\zeta^{\ell i_{t}q^{j}}\big)\mu_{-q}^{r_{1}}\rho^{r_{2}}\sigma_{\xi}^{r_{3}}(e_{i_{t}q^{j}})\\
			&=\sum_{j=0}^{k-1}\xi^{r_{3}}\zeta^{i_{t}q^{j}r_{2}}\big(\sum\limits_{\ell=0}^{k-1}c_{\ell}\zeta^{\ell i_{t}q^{j}}\big)e_{(-1)^{r_{1}}i_{t}q^{-r_{1}+j}}\\
			&=\sum_{j=0}^{k-1}\xi^{r_{3}}\zeta^{i_{t}q^{-r_{1}+j}q^{r_{1}}r_{2}}\big(\sum\limits_{\ell=0}^{k-1}c_{\ell}\zeta^{\ell i_{t}q^{-r_{1}+j}}\big)^{q^{r_{1}}}e_{(-1)^{r_{1}}i_{t}q^{-r_{1}+j}}\\
			&=\sum_{j=0}^{k-1}\xi^{r_{3}}\zeta^{i_{t}q^{r_{1}+j}r_{2}}\big(\sum\limits_{\ell=0}^{k-1}c_{\ell}\zeta^{\ell i_{t}q^{j}}\big)^{q^{r_{1}}}e_{(-1)^{r_{1}}i_{t}q^{j}}.
		\end{align*}
		If $r_{1}$ is odd, then
		$$\mu_{-q}^{r_{1}}\rho^{r_{2}}\sigma_{\xi}^{r_{3}}({\bf c})=\sum_{j=0}^{k-1}\xi^{r_{3}}\zeta^{i_{t}q^{r_{1}+j}r_{2}}\big(\sum\limits_{\ell=0}^{k-1}c_{\ell}\zeta^{\ell i_{t}q^{j}}\big)^{q^{r_{1}}}e_{-i_{t}q^{j}}\in \mu_{-1}(\mathcal{R}_{n}\varepsilon_{t})\backslash \{{\bf 0}\},$$
		and so $\mu_{-q}^{r_{1}}\rho^{r_{2}}\sigma_{\xi}^{r_{3}}({\bf c})\neq {\bf c}$. Suppose $r_{1}$ is even, then 
		$$\mu_{-q}^{r_{1}}\rho^{r_{2}}\sigma_{\xi}^{r_{3}}({\bf c})=\sum_{j=0}^{k-1}\xi^{r_{3}}\zeta^{i_{t}q^{r_{1}+j}r_{2}}\big(\sum\limits_{\ell=0}^{k-1}c_{\ell}\zeta^{\ell i_{t}q^{j}}\big)^{q^{r_{1}}}e_{i_{t}q^{j}},$$
		and it follows that
		$$\mu_{-q}^{r_{1}}\rho^{r_{2}}\sigma_{\xi}^{r_{3}}({\bf c})={\bf c}~\Leftrightarrow~\xi^{r_{3}}\big(\sum\limits_{\ell=0}^{k-1}c_{\ell}\zeta^{\ell i_{t}}\big)^{q^{r_{1}}-1}=\zeta^{-i_{t}q^{r_{1}}r_{2}}.$$
		From above analysis and the proof of Theorem \ref{t3.1} we see that
		\begin{align*}
			&\sum_{r_{1}=0}^{2m-1}\sum_{r_{2}=0}^{n-1}\sum_{r_{3}=0}^{q-2}\Big|\big\{{\bf c}\in \mathcal{R}_{n}\varepsilon_{t}\backslash \{{\bf 0}\}~\big|~\mu_{-q}^{r_{1}}\rho^{r_{2}}\sigma_{\xi}^{r_{3}}({\bf c})={\bf c}\big\}\Big|\\
			=&\sum_{\substack{0\leq r_{1}\leq 2m-1\\ r_{1}~{\rm is~even}}}\sum_{r_{2}=0}^{n-1}\sum_{r_{3}=0}^{q-2}\Big|\big\{{\bf c}\in \mathcal{R}_{n}\varepsilon_{t}\backslash \{{\bf 0}\}~\big|~\mu_{-q}^{r_{1}}\rho^{r_{2}}\sigma_{\xi}^{r_{3}}({\bf c})={\bf c}\big\}\Big|\\
			=&\sum_{\substack{0\leq r_{1}\leq 2m-1\\ r_{1}~{\rm is~even}}}{\rm gcd}\big(n(q-1)(q^{{\rm gcd}(k,r_{1})}-1),n(q^{k}-1),i_{t}(q-1)(q^{k}-1)\big).
		\end{align*}
		
		Take ${\bf c}=\sum\limits_{j=0}^{k-1}\big(\sum\limits_{\ell=0}^{k-1}c_{\ell}'\zeta^{-\ell i_{t}q^{j}}\big)e_{-i_{t}q^{j}}\in \mu_{-1}(\mathcal{R}_{n}\varepsilon_{t})\backslash \{\bf 0\}$.
		Similarly as above, we have
		$$\mu_{-q}^{r_{1}}\rho^{r_{2}}\sigma_{\xi}^{r_{3}}({\bf c})=\sum_{j=0}^{k-1}\xi^{r_{3}}\zeta^{-i_{t}q^{r_{1}+j}r_{2}}\big(\sum\limits_{\ell=0}^{k-1}c_{\ell}'\zeta^{-\ell i_{t}q^{j}}\big)^{q^{r_{1}}}e_{(-1)^{r_{1}+1}i_{t}q^{j}}.$$
		If $r_{1}$ is odd, then
		$$\mu_{-q}^{r_{1}}\rho^{r_{2}}\sigma_{\xi}^{r_{3}}({\bf c})=\sum_{j=0}^{k-1}\xi^{r_{3}}\zeta^{-i_{t}q^{r_{1}+j}r_{2}}\big(\sum\limits_{\ell=0}^{k-1}c_{\ell}'\zeta^{-\ell i_{t}q^{j}}\big)^{q^{r_{1}}}e_{i_{t}q^{j}}\in \mathcal{R}_{n}\varepsilon_{t}\backslash \{\bf 0\},$$
		and so $\mu_{-q}^{r_{1}}\rho^{r_{2}}\sigma_{\xi}^{r_{3}}({\bf c})\neq {\bf c}$. Suppose $r_{1}$ is even, then
		$$\mu_{-q}^{r_{1}}\rho^{r_{2}}\sigma_{\xi}^{r_{3}}({\bf c})=\sum_{j=0}^{k-1}\xi^{r_{3}}\zeta^{-i_{t}q^{r_{1}+j}r_{2}}\big(\sum\limits_{\ell=0}^{k-1}c_{\ell}'\zeta^{-\ell i_{t}q^{j}}\big)^{q^{r_{1}}}e_{-i_{t}q^{j}},$$
		and it follows that
		$$\mu_{-q}^{r_{1}}\rho^{r_{2}}\sigma_{\xi}^{r_{3}}({\bf c})={\bf c}~\Leftrightarrow~\xi^{r_{3}}\big(\sum\limits_{\ell=0}^{k-1}c_{\ell}'\zeta^{-\ell i_{t}}\big)^{q^{r_{1}}-1}=\zeta^{i_{t}q^{r_{1}}r_{2}}.$$
		From above analysis and the proof of Theorem \ref{t3.1} we see that
		\begin{align*}
			&\sum_{r_{1}=0}^{2m-1}\sum_{r_{2}=0}^{n-1}\sum_{r_{3}=0}^{q-2}\Big|\big\{{\bf c}\in \mu_{-1}(\mathcal{R}_{n}\varepsilon_{t})\backslash \{{\bf 0}\}~\big|~\mu_{-q}^{r_{1}}\rho^{r_{2}}\sigma_{\xi}^{r_{3}}({\bf c})={\bf c}\big\}\Big|\\
			=&\sum_{\substack{0\leq r_{1}\leq 2m-1\\ r_{1}~{\rm is~even}}}\sum_{r_{2}=0}^{n-1}\sum_{r_{3}=0}^{q-2}\Big|\big\{{\bf c}\in \mu_{-1}(\mathcal{R}_{n}\varepsilon_{t})\backslash \{{\bf 0}\}~\big|~\mu_{-q}^{r_{1}}\rho^{r_{2}}\sigma_{\xi}^{r_{3}}({\bf c})={\bf c}\big\}\Big|\\
			=&\sum_{\substack{0\leq r_{1}\leq 2m-1\\ r_{1}~{\rm is~even}}}{\rm gcd}\big(n(q-1)(q^{{\rm gcd}(k,r_{1})}-1),n(q^{k}-1),i_{t}(q-1)(q^{k}-1)\big).
		\end{align*}
		
		We then conclude that the number of orbits of $\langle\mu_{-q},\rho,\sigma_{\xi} \rangle$ on $$\mathcal{C}'=\big(\mathcal{R}_{n}\varepsilon_{t}\backslash \{{\bf 0}\}\big)\cup \big(\mu_{-1}(\mathcal{R}_{n}\varepsilon_{t})\backslash \{{\bf 0}\}\big)$$ is equal to
		\begin{align*}
			&\frac{2}{2mn(q-1)}\sum_{\substack{0\leq r_{1}\leq 2m-1\\ r_{1}~{\rm is~even}}}{\rm gcd}\big(n(q-1)(q^{{\rm gcd}(k,r_{1})}-1),n(q^{k}-1),i_{t}(q-1)(q^{k}-1)\big)\\
			=&\frac{1}{m}\sum_{\substack{0\leq r_{1}\leq 2m-1\\ r_{1}~{\rm is~even}}}{\rm gcd}\big(q^{{\rm gcd}(k,r_{1})}-1,\frac{q^{k}-1}{q-1},\frac{i_{t}(q^{k}-1)}{n}\big)\\
			=&\frac{1}{m}\sum_{r=0}^{m-1}{\rm gcd}\big(q^{{\rm gcd}(k,2r)}-1,\frac{q^{k}-1}{q-1},\frac{i_{t}(q^{k}-1)}{q-1}\big)\\
			=&\frac{1}{m}\sum_{r=0}^{m-1}{\rm gcd}\big(q^{{\rm gcd}(k,r)}-1,\frac{q^{k}-1}{q-1},\frac{i_{t}(q^{k}-1)}{q-1}\big)\\
			=&\frac{1}{k}\sum_{r|k}\varphi(\frac{k}{r}){\rm gcd}\big(q^{r}-1,\frac{q^{k}-1}{q-1},\frac{i_{t}(q^{k}-1)}{q-1}\big).
		\end{align*}
		The proof is then completed.
	\end{proof}
	
	\begin{Lemma}\label{l3.3}
		With the notation given above, then the number of orbits of $\langle\mu_{-q},\rho,\sigma_{\xi} \rangle$ on $$\mathcal{C}^{\sharp}=\mathcal{R}_{n}\varepsilon_{t}\backslash \{{\bf 0}\}\bigoplus \mu_{-1}(\mathcal{R}_{n}\varepsilon_{t})\backslash \{{\bf 0}\}$$ is equal to
		\begin{align*}
			&\frac{1}{2k}\sum_{r|k}\varphi(\frac{k}{r})\left({\rm gcd}\big(q^{r}-1,\frac{2(q^{k}-1)}{q-1}\big)\right. \\
			&\quad \left. +{\rm gcd}\Big((q^{r}-1){\rm gcd}\big(q^{r}-1,\frac{q^{k}-1}{q-1},\frac{i_{t}(q^{k}-1)}{n}\big),\frac{2i_{t}(q^{k}-1)^{2}}{n(q-1)}\Big)\right).
		\end{align*}
	\end{Lemma}
	\begin{proof}
		According to Equation (\ref{e2.1}), Lemmas \ref{l2.2} and \ref{L1}, we have
		$$\big|\langle\mu_{-q},\rho,\sigma_{\xi} \rangle\backslash \mathcal{C}^{\sharp}\big|=\frac{1}{2mn(q-1)}\sum_{r_{1}=0}^{2m-1}\sum_{r_{2}=0}^{n-1}\sum_{r_{3}=0}^{q-2}\Big|\big\{{\bf c}\in \mathcal{C}^{\sharp}~\big|~\mu_{-q}^{r_{1}}\rho^{r_{2}}\sigma_{\xi}^{r_{3}}({\bf c})={\bf c}\big\}\Big|.$$
		
		Take ${\bf c}={\bf c}_{t}+{\bf c}_{t'}\in \mathcal{C}^{\sharp}$, where ${\bf c}_{t}\in \mathcal{R}_{n}\varepsilon_{t}\backslash \{{\bf 0}\}$ and ${\bf c}_{t}'\in \mu_{-1}(\mathcal{R}_{n}\varepsilon_{t})\backslash \{{\bf 0}\}$. Suppose that
		$${\bf c}_{t}=\sum_{j=0}^{k-1}\big(\sum_{\ell=0}^{k-1}c_{\ell}\zeta^{\ell i_{t}q^{j}}\big)e_{i_{t}q^{j}}~~{\rm and}~~{\bf c}_{t'}=\sum_{j=0}^{k-1}\big(\sum_{\ell=0}^{k-1}c_{\ell}'\zeta^{-\ell i_{t}q^{j}}\big)e_{-i_{t}q^{j}}.$$
		Then
		\begin{align*}
			\mu_{-q}^{r_{1}}\rho^{r_{2}}\sigma_{\xi}^{r_{3}}({\bf c})&=\mu_{-q}^{r_{1}}\rho^{r_{2}}\sigma_{\xi}^{r_{3}}({\bf c}_{t})+\mu_{-q}^{r_{1}}\rho^{r_{2}}\sigma_{\xi}^{r_{3}}({\bf c}_{t'})\\
			&=\sum_{j=0}^{k-1}\xi^{r_{3}}\zeta^{i_{t}q^{r_{1}+j}r_{2}}\big(\sum_{\ell=0}^{k-1}c_{\ell}\zeta^{\ell i_{t}q^{j}}\big)^{q^{r_{1}}}e_{(-1)^{r_{1}}i_{t}q^{j}}\\
			&\quad +\sum_{j=0}^{k-1}\xi^{r_{3}}\zeta^{-i_{t}q^{r_{1}+j}r_{2}}\big(\sum_{\ell=0}^{k-1}c_{\ell}'\zeta^{-\ell i_{t}q^{j}}\big)^{q^{r_{1}}}e_{(-1)^{r_{1}+1}i_{t}q^{j}}.
		\end{align*}
		
		Suppose $r_{1}$ is odd,  then
		\begin{align*}
			\mu_{-q}^{r_{1}}\rho^{r_{2}}\sigma_{\xi}^{r_{3}}({\bf c})&=\sum_{j=0}^{k-1}\xi^{r_{3}}\zeta^{i_{t}q^{r_{1}+j}r_{2}}\big(\sum_{\ell=0}^{k-1}c_{\ell}\zeta^{\ell i_{t}q^{j}}\big)^{q^{r_{1}}}e_{-i_{t}q^{j}}\\
			&\quad +\sum_{j=0}^{k-1}\xi^{r_{3}}\zeta^{-i_{t}q^{r_{1}+j}r_{2}}\big(\sum_{\ell=0}^{k-1}c_{\ell}'\zeta^{-\ell i_{t}q^{j}}\big)^{q^{r_{1}}}e_{i_{t}q^{j}},
		\end{align*}
		and it follows that
		\begin{align*}
			\mu_{-q}^{r_{1}}\rho^{r_{2}}\sigma_{\xi}^{r_{3}}({\bf c})={\bf c}~& \Leftrightarrow~\mu_{-q}^{r_{1}}\rho^{r_{2}}\sigma_{\xi}^{r_{3}}({\bf c}_{t})+\mu_{-q}^{r_{1}}\rho^{r_{2}}\sigma_{\xi}^{r_{3}}({\bf c}_{t'})={\bf c}_{t}+{\bf c}_{t'}\\
			&\Leftrightarrow~ \mu_{-q}^{r_{1}}\rho^{r_{2}}\sigma_{\xi}^{r_{3}}({\bf c}_{t})={\bf c}_{t'}~{\rm and}~\mu_{-q}^{r_{1}}\rho^{r_{2}}\sigma_{\xi}^{r_{3}}({\bf c}_{t'})={\bf c}_{t}\\
			&\Leftrightarrow~\xi^{r_{3}}\zeta^{i_{t}q^{r_{1}}r_{2}}\big(\sum_{\ell=0}^{k-1}c_{\ell}\zeta^{\ell i_{t}}\big)^{q^{r_{1}}} =\sum_{\ell=0}^{k-1}c_{\ell}'\zeta^{-\ell i_{t}}\\
			&\qquad {\rm and}~ \xi^{r_{3}}\zeta^{-i_{t}q^{r_{1}}r_{2}}\big(\sum_{\ell=0}^{k-1}c_{\ell}'\zeta^{-\ell i_{t}}\big)^{q^{r_{1}}}=\sum_{\ell=0}^{k-1}c_{\ell}\zeta^{\ell i_{t}}.
		\end{align*}
		Hence the number of ${\bf c}\in \mathcal{C}^{\sharp}$ satisfying $\mu_{-q}^{r_{1}}\rho^{r_{2}}\sigma_{\xi}^{r_{3}}({\bf c})={\bf c}$ is equal to the number of pairs $(\alpha,\beta)$ with $\alpha,\beta\in \mathbb{F}_{q^{k}}^{*}$ such that $\xi^{r_{3}}\zeta^{i_{t}q^{r_{1}}r_{2}}\alpha^{q^{r_{1}}}=\beta$ and $\xi^{r_{3}}\zeta^{-i_{t}q^{r_{1}}r_{2}}\beta^{q^{r_{1}}}=\alpha$. One can check that
		\begin{align*}
			\begin{cases}\xi^{r_{3}}\zeta^{i_{t}q^{r_{1}}r_{2}}\alpha^{q^{r_{1}}}=\beta, \\ \xi^{r_{3}}\zeta^{-i_{t}q^{r_{1}}r_{2}}\beta^{q^{r_{1}}}=\alpha.
			\end{cases}\Leftrightarrow~
			\begin{cases}\xi^{2r_{3}}\alpha^{q^{2r_{1}}-1}=\zeta^{-i_{t}(q^{r_{1}}-1)q^{r_{1}}r_{2}},\\
				\beta=\xi^{r_{3}}\zeta^{i_{t}q^{r_{1}}r_{2}}\alpha^{q^{r_{1}}}.
			\end{cases}
		\end{align*}
		Therefore, the number of ${\bf c}\in \mathcal{C}^{\sharp}$ satisfying $\mu_{-q}^{r_{1}}\rho^{r_{2}}\sigma_{\xi}^{r_{3}}({\bf c})={\bf c}$ is equal to the number of $\alpha\in \mathbb{F}_{q^{k}}^{*}$ such that $\xi^{2r_{3}}\alpha^{q^{2r_{1}}-1}=\zeta^{-i_{t}(q^{r_{1}}-1)q^{r_{1}}r_{2}}$. Similar to the proof of Theorem \ref{t3.1}, we have
		\begin{align*}
			&\sum_{\substack{0\leq r_{1}\leq 2m-1\\ r_{1}~{\rm is~odd}}}\sum_{r_{2}=0}^{n-1}\sum_{r_{3}=0}^{q-2}\Big|\big\{{\bf c}\in \mathcal{C}^{\sharp}~\big|~\mu_{-q}^{r_{1}}\rho^{r_{2}}\sigma_{\xi}^{r_{3}}({\bf c})={\bf c}\big\}\Big|\\
			=&\sum_{\substack{0\leq r_{1}\leq 2m-1\\ r_{1}~{\rm is~odd}}}{\rm gcd}\big(n(q-1)(q^{{\rm gcd}(k,2r_{1})}-1),2n(q^{k}-1),i_{t}(q^{r_{1}}-1)(q-1)(q^{k}-1)\big)\\
			=&\sum_{r=0}^{m-1}{\rm gcd}\big(n(q-1)(q^{{\rm gcd}(k,4r+2)}-1),2n(q^{k}-1),i_{t}(q^{2r+1}-1)(q-1)(q^{k}-1)\big)\\
			=&\sum_{r=0}^{m-1}{\rm gcd}\big(n(q-1)(q^{{\rm gcd}(k,2r+1)}-1),2n(q^{k}-1)\big)\\
			=&\sum_{r=0}^{m-1}{\rm gcd}\big(n(q-1)(q^{{\rm gcd}(k,r)}-1),2n(q^{k}-1)\big).
		\end{align*}
		The last equation holds because ${\rm gcd}(2,m)=1$ and $k|m$.
		
		Suppose $r_{1}$ is even, then 
		\begin{align*}
			\mu_{-q}^{r_{1}}\rho^{r_{2}}\sigma_{\xi}^{r_{3}}({\bf c})&=\sum_{j=0}^{k-1}\xi^{r_{3}}\zeta^{i_{t}q^{r_{1}+j}r_{2}}\big(\sum_{\ell=0}^{k-1}c_{\ell}\zeta^{\ell i_{t}q^{j}}\big)^{q^{r_{1}}}e_{i_{t}q^{j}}\\
			&\quad +\sum_{j=0}^{k-1}\xi^{r_{3}}\zeta^{-i_{t}q^{r_{1}+j}r_{2}}\big(\sum_{\ell=0}^{k-1}c_{\ell}'\zeta^{-\ell i_{t}q^{j}}\big)^{q^{r_{1}}}e_{-i_{t}q^{j}},
		\end{align*}
		and it follows that
		\begin{align*}
			\mu_{-q}^{r_{1}}\rho^{r_{2}}\sigma_{\xi}^{r_{3}}({\bf c})={\bf c}~& \Leftrightarrow~\mu_{-q}^{r_{1}}\rho^{r_{2}}\sigma_{\xi}^{r_{3}}({\bf c}_{t})+\mu_{-q}^{r_{1}}\rho^{r_{2}}\sigma_{\xi}^{r_{3}}({\bf c}_{t'})={\bf c}_{t}+{\bf c}_{t'}\\
			&\Leftrightarrow~ \mu_{-q}^{r_{1}}\rho^{r_{2}}\sigma_{\xi}^{r_{3}}({\bf c}_{t})={\bf c}_{t}~{\rm and}~\mu_{-q}^{r_{1}}\rho^{r_{2}}\sigma_{\xi}^{r_{3}}({\bf c}_{t'})={\bf c}_{t'}\\
			&\Leftrightarrow~\xi^{r_{3}}\big(\sum_{\ell=0}^{k-1}c_{\ell}\zeta^{\ell i_{t}}\big)^{q^{r_{1}}-1} =\zeta^{-i_{t}q^{r_{1}}r_{2}}~ {\rm and}~ \xi^{r_{3}}\big(\sum_{\ell=0}^{k-1}c_{\ell}'\zeta^{-\ell i_{t}}\big)^{q^{r_{1}}-1}=\zeta^{i_{t}q^{r_{1}}r_{2}}.
		\end{align*}
		Hence the number of ${\bf c}\in \mathcal{C}^{\sharp}$ satisfying $\mu_{-q}^{r_{1}}\rho^{r_{2}}\sigma_{\xi}^{r_{3}}({\bf c})={\bf c}$ is equal to the number of pairs $(\alpha,\beta)$ with $\alpha,\beta\in \mathbb{F}_{q^{k}}^{*}$ such that $\xi^{r_{3}}\alpha^{q^{r_{1}}-1}=\zeta^{-i_{t}q^{r_{1}}r_{2}}$ and $\xi^{r_{3}}\beta^{q^{r_{1}}-1}=\zeta^{i_{t}q^{r_{1}}r_{2}}$. We deduce from the proof of Corollary \ref{c3.4} that
		\begin{align*}
			&\sum_{\substack{0\leq r_{1}\leq 2m-1\\ r_{1}~{\rm is~even}}}\sum_{r_{2}=0}^{n-1}\sum_{r_{3}=0}^{q-2}\Big|\big\{{\bf c}\in \mathcal{C}^{\sharp}~\big|~\mu_{-q}^{r_{1}}\rho^{r_{2}}\sigma_{\xi}^{r_{3}}({\bf c})={\bf c}\big\}\Big|\\
			=&\sum_{\substack{0\leq r_{1}\leq 2m-1\\ r_{1}~{\rm is~even}}}{\rm gcd}\Big(n(q-1)(q^{{\rm gcd}(k,r_{1})}-1){\rm gcd}\big(q^{{\rm gcd}(k,r_{1})}-1,\frac{q^{k}-1}{q-1},\frac{i_{t}(q^{k}-1)}{n}\big),2i_{t}(q^{k}-1)^{2}\Big)\\
			=&\sum_{r=0}^{m-1}{\rm gcd}\Big(n(q-1)(q^{{\rm gcd}(k,2r)}-1){\rm gcd}\big(q^{{\rm gcd}(k,2r)}-1,\frac{q^{k}-1}{q-1},\frac{i_{t}(q^{k}-1)}{n}\big),2i_{t}(q^{k}-1)^{2}\Big)\\
			=&\sum_{r=0}^{m-1}{\rm gcd}\Big(n(q-1)(q^{{\rm gcd}(k,r)}-1){\rm gcd}\big(q^{{\rm gcd}(k,r)}-1,\frac{q^{k}-1}{q-1},\frac{i_{t}(q^{k}-1)}{n}\big),2i_{t}(q^{k}-1)^{2}\Big).
		\end{align*}
		
		We then conclude that the number of orbits of $\langle\mu_{-q},\rho,\sigma_{\xi} \rangle$ on $$\mathcal{C}^{\sharp}=\mathcal{R}_{n}\varepsilon_{t}\backslash \{{\bf 0}\}\bigoplus \mu_{-1}(\mathcal{R}_{n}\varepsilon_{t})\backslash \{{\bf 0}\}$$ is equal to
		\begin{align*}
			&\frac{1}{2mn(q-1)}\left(\sum_{r=0}^{m-1}{\rm gcd}\big(n(q-1)(q^{{\rm gcd}(k,r)}-1),2n(q^{k}-1)\big)\right. \\
			&+\left. \sum_{r=0}^{m-1}{\rm gcd}\Big(n(q-1)(q^{{\rm gcd}(k,r)}-1){\rm gcd}\big(q^{{\rm gcd}(k,r)}-1,\frac{q^{k}-1}{q-1},\frac{i_{t}(q^{k}-1)}{n}\big),2i_{t}(q^{k}-1)^{2}\Big)\right) \\
			=&\frac{1}{2m}\sum_{r=0}^{m-1}\left({\rm gcd}\big(q^{{\rm gcd}(k,r)}-1,\frac{2(q^{k}-1)}{q-1}\big)\right.\\
			& +\left.{\rm gcd}\Big((q^{{\rm gcd}(k,r)}-1){\rm gcd}\big(q^{{\rm gcd}(k,r)}-1,\frac{q^{k}-1}{q-1},\frac{i_{t}(q^{k}-1)}{n}\big),\frac{2i_{t}(q^{k}-1)^{2}}{n(q-1)}\Big)\right)\\
			=&\frac{1}{2k}\sum_{r|k}\varphi(\frac{k}{r})\left({\rm gcd}\big(q^{r}-1,\frac{2(q^{k}-1)}{q-1}\big)\right.\\
			&+\left.{\rm gcd}\Big((q^{r}-1){\rm gcd}\big(q^{r}-1,\frac{q^{k}-1}{q-1},\frac{i_{t}(q^{k}-1)}{n}\big),\frac{2i_{t}(q^{k}-1)^{2}}{n(q-1)}\Big)\right).
		\end{align*}
		The proof is then completed.
	\end{proof}
	
	By virtue of Lemmas \ref{l3.2} and \ref{l3.3}, the number of orbits of $\langle\mu_{-q},\rho,\sigma_{\xi} \rangle$ on $\mathcal{C}^{*}=\mathcal{C}\backslash \{\bf 0\}$ can be immediately obtained.
	
	\begin{Theorem}\label{t3.4}
		Let $\mathcal{C}$ be a cyclic code of length $n$ over $\mathbb{F}_{q}$. Suppose that $$\mathcal{C}=\mathcal{R}_{n}\varepsilon_{t}\bigoplus\mu_{-1}(\mathcal{R}_{n}\varepsilon_{t}),$$ 
		where $0<t\leq s$, the primitive idempotent $\varepsilon_{t}$ corresponds to the $q$-cyclotomic coset $\{i_{t}, i_{t}q,\cdots,i_{t}q^{k-1}\}$ and $-i_{t}\notin \{i_{t}, i_{t}q,\cdots,i_{t}q^{k-1}\}$. Suppose $-1\in \langle-q \rangle_{_{\mathbb{Z}_{n}^{*}}}$, then the number of orbits of $\langle\mu_{-q},\rho,\sigma_{\xi} \rangle$ on $\mathcal{C}^{*}=\mathcal{C}\backslash \{\bf 0\}$ is equal to
		\begin{align*}
			&\frac{1}{2k}\sum_{r|k}\varphi(\frac{k}{r})\left(2{\rm gcd}\big(q^{r}-1,\frac{q^{k}-1}{q-1},\frac{i_{t}(q^{k}-1)}{n}\big)+{\rm gcd}\big(q^{r}-1,\frac{2(q^{k}-1)}{q-1}\big)\right.\\
			&\quad +\left.{\rm gcd}\Big((q^{r}-1){\rm gcd}\big(q^{r}-1,\frac{q^{k}-1}{q-1},\frac{i_{t}(q^{k}-1)}{n}\big),\frac{2i_{t}(q^{k}-1)^{2}}{n(q-1)}\Big)\right).
		\end{align*}
		In particular, the number of non-zero weights of $\mathcal{C}$ is less than or equal to the number of orbits of $\langle\mu_{-q},\rho,\sigma_{\xi} \rangle$ on $\mathcal{C}^{*}$, with equality if and only if for any two codewords ${\bf c}_{1},{\bf c}_{2}\in \mathcal{C}^{*}$ with the same weight, there exist integers $j_{1}$, $j_{2}$ and $j_{3}$ such that $\mu_{-q}^{j_{1}}\rho^{j_{2}}(\xi^{j_{3}}{\bf c}_{1})={\bf c}_{2}$.
	\end{Theorem}
	
	\begin{proof}
		Note that
		$$\mathcal{C}\backslash \{{\bf 0}\}=\mathcal{C}'\cup \mathcal{C}^{\sharp}$$
		is a disjoint union, where $$\mathcal{C}'=\big(\mathcal{R}_{n}\varepsilon_{t}\backslash \{{\bf 0}\}\big)\cup \big(\mu_{-1}(\mathcal{R}_{n}\varepsilon_{t})\backslash \{{\bf 0}\}\big)~{\rm and}~\mathcal{C}^{\sharp}=\mathcal{R}_{n}\varepsilon_{t}\backslash \{{\bf 0}\}\bigoplus\mu_{-1}(\mathcal{R}_{n}\varepsilon_{t})\backslash \{{\bf 0}\}.$$ Then 
		$$\big|\langle\mu_{-q},\rho,\sigma_{\xi} \rangle\backslash \mathcal{C}^{*}\big|=\big|\big\langle \mu_{-q},\rho,M \big\rangle\big\backslash \mathcal{C}'\big|+\big|\big\langle \mu_{-q},\rho,M \big\rangle\big\backslash \mathcal{C}^{\sharp}\big|.$$ The rest of the proof is clear with the help of Lemmas \ref{l3.2} and \ref{l3.3}.
	\end{proof}
	
	\begin{Remark}{\rm
			Let $\mathcal{C}$ be the cyclic code in Theorem \ref{t3.4}. It follows from Corollary \ref{c3.4} and Theorem \ref{t3.4} that
			\begin{align*}
				&\big|\langle\mu_{q},\rho,\sigma_{\xi} \rangle\backslash \mathcal{C}^{*}\big|-\big|\langle\mu_{-q},\rho,\sigma_{\xi} \rangle\backslash \mathcal{C}^{*}\big|\\
				=&\frac{1}{k}\sum_{r|k}\varphi(\frac{k}{r}){\rm gcd}\big(q^{r}-1,\frac{q^{k}-1}{q-1},\frac{i_{t}(q^{k}-1)}{n}\big)+\frac{1}{2k}\sum_{r|k}\varphi(\frac{k}{r})\left({\rm gcd}\Big((q^{r}-1){\rm gcd}\big(q^{r}-1,\frac{q^{k}-1}{q-1},\right.\\
				& \left.\frac{i_{t}(q^{k}-1)}{n}\big),\frac{2i_{t}(q^{k}-1)^{2}}{n(q-1)}\Big)-{\rm gcd}\big(q^{r}-1,\frac{2(q^{k}-1)}{q-1}\big)\right)\\
				\geq &\frac{1}{k}\sum_{r|k}\varphi(\frac{k}{r}){\rm gcd}\big(q^{r}-1,\frac{q^{k}-1}{q-1},\frac{i_{t}(q^{k}-1)}{n}\big).
			\end{align*}
			Hence the upper bound on the number of non-zero weights of $\mathcal{C}$ given by Theorem \ref{t3.4} is less than that given by Corollary \ref{c3.4}.}
	\end{Remark}
	
	\begin{Example}{\rm
			Take $q=2$ and $n=7$. All the distinct $2$-cyclotomic cosets modulo $7$ are given by
			$$\Gamma_{0}=\{0\},~\Gamma_{1}=\{1,2,4\},~\Gamma_{2}=\{3,6,5\}.$$
			Let $\ell$ be the number of non-zero weights of the cyclic code $\mathcal{C}=\mathcal{R}_{n}\varepsilon_{1}\bigoplus \mu_{-1}(\mathcal{R}_{n}\varepsilon_{1})=\mathcal{R}_{n}\varepsilon_{1} \bigoplus \mathcal{R}_{n}\varepsilon_{2}$. By Corollary \ref{c3.4}, we have
			\begin{align*}
				\ell&\leq \frac{1}{3}\sum_{r\mid 3}\varphi(\frac{3}{r})\left( {\rm gcd}\big(2^{r}-1,\frac{2^{3}-1}{2-1},\frac{2^{3}-1}{7}\big)+ {\rm gcd}\big(2^{r}-1,\frac{2^{3}-1}{2-1},\frac{3(2^{3}-1)}{7}\big)\right.  \\
				&\quad +\left. {\rm gcd}\Big((2^{r}-1){\rm gcd}\big(2^{r}-1,\frac{2^{3}-1}{2-1},\frac{2^{3}-1}{7},\frac{3(2^{3}-1)}{7}\big),\frac{(3-1)(2^{3}-1)^{2}}{7(2-1)}\Big) \right)\\
				&=\frac{1}{3}[\varphi(3)(1+1+1)+\varphi(1)(1+1+7)]\\
				&=\dfrac{1}{3}(6+9)=5.
			\end{align*}
			Using Theorem \ref{t3.5}, we have
			\begin{align*}
				\ell&\leq \frac{1}{6}\sum_{r|3}\varphi(\frac{3}{r})\left(2{\rm gcd}\big(2^{r}-1,\frac{2^{3}-1}{2-1},\frac{2^{3}-1}{7}\big)+{\rm gcd}\big(2^{r}-1,\frac{2(2^{3}-1)}{2-1}\big)\right.\\
				&\quad +\left.{\rm gcd}\Big((2^{r}-1){\rm gcd}\big(2^{r}-1,\frac{2^{3}-1}{2-1},\frac{2^{3}-1}{7}\big),\frac{2(2^{3}-1)^{2}}{7(2-1)}\Big)\right) \\
				&=\frac{1}{6}[\varphi(3)(2+1+1)+\varphi(1)(2+7+7)]\\
				&=\frac{1}{6}(8+16)=4.
			\end{align*}
			After using Magma \cite{4}, we know that the weight distribution of $\mathcal{C}$ is $1+21x^{2}+35x^{4}+7x^{6}$, which implies that the exact value of $\ell$ is 3. }
	\end{Example}
	
	\subsubsection{The action of $\langle\mu_{-1},\mu_{-q},\rho,\sigma_{\xi} \rangle$ on $\mathcal{C}^{*}$}
	Assume that $-1\not\in \langle-q \rangle_{_{\mathbb{Z}_{n}^{*}}}$ in this subsection, and we now turn to consider the action of $\langle\mu_{-1},\mu_{-q},\rho,\sigma_{\xi} \rangle$ on $\mathcal{C}^{*}=\mathcal{C}\backslash \{\bf 0\}$, where $\mathcal{C}=\mathcal{R}_{n}\varepsilon_{t}\bigoplus\mu_{-1}(\mathcal{R}_{n}\varepsilon_{t})$. 
	
	\begin{Lemma}\label{L2}
		Let $m$ be the order of $q$ in $\mathbb{Z}_{n}^{*}$ and let $m'$ be the order of $-q$ in $\mathbb{Z}_{n}^{*}$. Suppose that $-i_{t}\notin \{i_{t}, i_{t}q,\cdots,i_{t}q^{k-1}\}$ and $-1\not\in \langle-q \rangle_{_{\mathbb{Z}_{n}^{*}}}$, then $m'=m$.
	\end{Lemma}
	\begin{proof}
		As is pointed out in Lemma \ref{L1}, $m'$ is even. Since $-1\not\in \langle-q \rangle_{_{\mathbb{Z}_{n}^{*}}}$, $m$ is even. It follows that $(-q)^{m}\equiv q^{m} \equiv 1~({\rm mod}~n)$ and $q^{m'}\equiv (-q)^{m'} \equiv 1~({\rm mod}~n)$, then $m'|m$ and $m|m'$, and hence $m'=m$.
	\end{proof}
	
	\begin{Lemma}\label{l3.4}
		Suppose $-1\not\in \langle-q \rangle_{_{\mathbb{Z}_{n}^{*}}}$. The subgroup $\langle\mu_{-1},\mu_{-q},\rho,\sigma_{\xi} \rangle$ of ${\rm Aut}(\mathcal{C})$ is of order $2mn(q-1)$, and each element of $\big\langle \mu_{-1},\mu_{-q},\rho,\sigma_{\xi} \big\rangle$ can be written uniquely as a product $\mu_{-1}^{r_{0}}\mu_{-q}^{r_{1}}\rho^{r_{2}}\sigma_{\xi}^{r_{3}}$ for some $0\leq r_{0}\leq 1$, $0\leq r_{1}\leq m-1$, $0\leq r_{2}\leq n-1$ and $0\leq r_{3}\leq q-2$.
	\end{Lemma}
	
	\begin{proof}
		The proof is similar to that of Lemma \ref{l2.2}.
	\end{proof}
	
	\begin{Theorem}\label{t3.5}
		Let $\mathcal{C}$ be a cyclic code of length $n$ over $\mathbb{F}_{q}$. Suppose that $$\mathcal{C}=\mathcal{R}_{n}\varepsilon_{t}\bigoplus\mu_{-1}(\mathcal{R}_{n}\varepsilon_{t}),$$
		where $0<t\leq s$, the primitive idempotent $\varepsilon_{t}$ corresponds to the $q$-cyclotomic coset $\{i_{t}, i_{t}q,\cdots,i_{t}q^{k-1}\}$ and $-i_{t}\notin \{i_{t}, i_{t}q,\cdots,i_{t}q^{k-1}\}$. Suppose $-1\not\in \langle-q \rangle_{_{\mathbb{Z}_{n}^{*}}}$, then the number of orbits of $\langle\mu_{-1},\mu_{-q},\rho,\sigma_{\xi} \rangle$ on $\mathcal{C}^{*}=\mathcal{C}\backslash \{\bf 0\}$ is equal to
		\begin{align*}
			&\frac{1}{2m}\sum_{r=0}^{m-1}\left(2{\rm gcd}\big(q^{{\rm gcd}(k,r)}-1,\frac{q^{k}-1}{q-1},\frac{i_{t}(q^{k}-1)}{n}\big)\right. \\
			&\quad +\left. {\rm gcd}\big(q^{{\rm gcd}(k,2r)}-1,\frac{2(q^{k}-1)}{q-1},\frac{i_{t}(q^{r}-1)(q^{k}-1)}{n}\big)\right. \\
			&\quad +\left. {\rm gcd}\Big((q^{{\rm gcd}(k,r)}-1){\rm gcd}\big(q^{{\rm gcd}(k,r)}-1,\frac{q^{k}-1}{q-1},\frac{i_{t}(q^{k}-1)}{n}\big),\frac{2i_{t}(q^{k}-1)^{2}}{n(q-1)}\Big)\right).
		\end{align*}
		In particular, the number of non-zero weights of $\mathcal{C}$ is less than or equal to the number of orbits of $\langle\mu_{-1},\mu_{-q},\rho,\sigma_{\xi} \rangle$ on $\mathcal{C}^{*}$, with equality if and only if for any two codewords ${\bf c}_{1},{\bf c}_{2}\in \mathcal{C}^{*}$ with the same weight, there exist integers $j_{0}$, $j_{1}$, $j_{2}$ and $j_{3}$ such that $\mu_{-1}^{j_{0}}\mu_{-q}^{j_{1}}\rho^{j_{2}}(\xi^{j_{3}}{\bf c}_{1})={\bf c}_{2}$.
	\end{Theorem}
	
	\begin{proof}
		Denote $\mathcal{C}'=\big(\mathcal{R}_{n}\varepsilon_{t}\backslash \{{\bf 0}\}\big)\cup \big(\mu_{-1}(\mathcal{R}_{n}\varepsilon_{t})\backslash \{{\bf 0}\}\big)$ and $\mathcal{C}^{\sharp}=\mathcal{R}_{n}\varepsilon_{t}\backslash \{{\bf 0}\}\bigoplus \mu_{-1}(\mathcal{R}_{n}\varepsilon_{t})\backslash \{{\bf 0}\}$. Then 
		$$\big|\langle\mu_{-1},\mu_{-q},\rho,\sigma_{\xi} \rangle\backslash \mathcal{C}^{*}\big|=\big|\langle\mu_{-1},\mu_{-q},\rho,\sigma_{\xi} \rangle\backslash \mathcal{C}'\big|+\big|\langle\mu_{-1},\mu_{-q},\rho,\sigma_{\xi} \rangle\backslash \mathcal{C}^{\sharp}\big|,$$
		and it follows from Equation \ref{e2.1} and Lemma \ref{l3.4} that
		$$\big|\langle\mu_{-1},\mu_{-q},\rho,\sigma_{\xi} \rangle\backslash \mathcal{C}'\big|=\frac{1}{2mn(q-1)}\sum_{r_{0}=0}^{1}\sum_{r_{1}=0}^{m-1}\sum_{r_{2}=0}^{n-1}\sum_{r_{3}=0}^{q-2}\Big|\big\{{\bf c}\in \mathcal{C}'~\big|~\mu_{-1}^{r_{0}}\mu_{-q}^{r_{1}}\rho^{r_{2}}\sigma_{\xi}^{r_{3}}({\bf c})={\bf c}\big\}\Big|,$$
		and
		$$\big|\langle\mu_{-1},\mu_{-q},\rho,\sigma_{\xi} \rangle\backslash \mathcal{C}^{\sharp}\big|=\frac{1}{2mn(q-1)}\sum_{r_{0}=0}^{1}\sum_{r_{1}=0}^{m-1}\sum_{r_{2}=0}^{n-1}\sum_{r_{3}=0}^{q-2}\Big|\big\{{\bf c}\in \mathcal{C}^{\sharp}~\big|~\mu_{-1}^{r_{0}}\mu_{-q}^{r_{1}}\rho^{r_{2}}\sigma_{\xi}^{r_{3}}({\bf c})={\bf c}\big\}\Big|.$$
		
		Take ${\bf c}_{t}=\sum\limits_{j=0}^{k-1}\big(\sum\limits_{\ell=0}^{k-1}c_{\ell}\zeta^{\ell i_{t}q^{j}}\big)e_{i_{t}q^{j}}\in \mathcal{R}_{n}\varepsilon_{t}\backslash \{\bf 0\}$. Note that $e_{i_{t}q^{j}}=\frac{1}{n}\sum\limits_{l=0}^{n-1}\zeta^{-i_{t}q^{j}l}x^{l}$ and $\mu_{-q}^{r_{1}}\rho^{r_{2}}\sigma_{\xi}^{r_{3}}(e_{i_{t}q^{j}})=\xi^{r_{3}}\zeta^{i_{t}q^{j}r_{2}}e_{(-1)^{r_{1}}i_{t}q^{-r_{1}+j}}$, and thus
		\begin{align*}
			\mu_{-1}^{r_{0}}\mu_{-q}^{r_{1}}\rho^{r_{2}}\sigma_{\xi}^{r_{3}}(e_{i_{t}q^{j}})&=\xi^{r_{3}}\zeta^{i_{t}q^{j}r_{2}}\mu_{-1}^{r_{0}}(e_{(-1)^{r_{1}}i_{t}q^{-r_{1}+j}})\\
			&=\xi^{r_{3}}\zeta^{i_{t}q^{j}r_{2}}\cdot \frac{1}{n}\sum_{l=0}^{n-1}\zeta^{-(-1)^{r_{1}}i_{t}q^{-r_{1}+j}l}x^{(-1)^{r_{0}}l}\\
			&=\xi^{r_{3}}\zeta^{i_{t}q^{j}r_{2}}\cdot \frac{1}{n}\sum_{l=0}^{n-1}\zeta^{-(-1)^{r_{0}+r_{1}}i_{t}q^{-r_{1}+j} (-1)^{r_{0}}l}x^{(-1)^{r_{0}}l}\\
			&=\xi^{r_{3}}\zeta^{i_{t}q^{j}r_{2}}\cdot \frac{1}{n}\sum_{l=0}^{n-1}\zeta^{-(-1)^{r_{0}+r_{1}}i_{t}q^{-r_{1}+j}l}x^{l}\\
			&=\xi^{r_{3}}\zeta^{i_{t}q^{j}r_{2}}e_{(-1)^{r_{0}+r_{1}}i_{t}q^{-r_{1}+j}}.
		\end{align*}
		We then have
		\begin{align*}
			\mu_{-1}^{r_{0}}\mu_{-q}^{r_{1}}\rho^{r_{2}}\sigma_{\xi}^{r_{3}}({\bf c})&=\sum_{j=0}^{k-1}\big(\sum\limits_{\ell=0}^{k-1}c_{\ell}\zeta^{\ell i_{t}q^{j}}\big)\mu_{-1}^{r_{0}}\mu_{-q}^{r_{1}}\rho^{r_{2}}\sigma_{\xi}^{r_{3}}(e_{i_{t}q^{j}})\\
			&=\sum_{j=0}^{k-1}\xi^{r_{3}}\zeta^{i_{t}q^{j}r_{2}}\big(\sum\limits_{\ell=0}^{k-1}c_{\ell}\zeta^{\ell i_{t}q^{j}}\big)e_{(-1)^{r_{0}+r_{1}}i_{t}q^{-r_{1}+j}}\\
			&=\sum_{j=0}^{k-1}\xi^{r_{3}}\zeta^{i_{t}q^{-r_{1}+j}q^{r_{1}}r_{2}}\big(\sum\limits_{\ell=0}^{k-1}c_{\ell}\zeta^{\ell i_{t}q^{-r_{1}+j}}\big)^{q^{r_{1}}}e_{(-1)^{r_{0}+r_{1}}i_{t}q^{-r_{1}+j}}\\
			&=\sum_{j=0}^{k-1}\xi^{r_{3}}\zeta^{i_{t}q^{r_{1}+j}r_{2}}\big(\sum\limits_{\ell=0}^{k-1}c_{\ell}\zeta^{\ell i_{t}q^{j}}\big)^{q^{r_{1}}}e_{(-1)^{r_{0}+r_{1}}i_{t}q^{j}}.
		\end{align*}
		Take ${\bf c}_{t'}=\sum\limits_{j=0}^{k-1}\big(\sum\limits_{\ell=0}^{k-1}c_{\ell}'\zeta^{-\ell i_{t}q^{j}}\big)e_{-i_{t}q^{j}}\in \mu_{-1}(\mathcal{R}_{n}\varepsilon_{t})\backslash \{\bf 0\}$. Similarly as above, we have
		$$\mu_{-1}^{r_{0}}\mu_{-q}^{r_{1}}\rho^{r_{2}}\sigma_{\xi}^{r_{3}}({\bf c}_{t'})=\sum_{j=0}^{k-1}\xi^{r_{3}}\zeta^{-i_{t}q^{r_{1}+j}r_{2}}\big(\sum\limits_{\ell=0}^{k-1}c_{\ell}'\zeta^{-\ell i_{t}q^{j}}\big)^{q^{r_{1}}}e_{(-1)^{r_{0}+r_{1}+1}i_{t}q^{j}}.$$
		
		Similar discussion as in the proof of Lemmas \ref{l3.2} and \ref{l3.3} shows that
		\begin{align*}
			&\big|\langle\mu_{-1},\mu_{-q},\rho,\sigma_{\xi} \rangle\backslash \mathcal{C}'\big|\\
			=&\frac{1}{2mn(q-1)}\Big(\sum_{r_{0}=0}^{1}\sum_{r_{1}=0}^{m-1}\sum_{r_{2}=0}^{n-1}\sum_{r_{3}=0}^{q-2}\Big|\big\{{\bf c}\in \mathcal{R}_{n}\varepsilon_{t}\backslash \{{\bf 0}\}~\big|~\mu_{-1}^{r_{0}}\mu_{-q}^{r_{1}}\rho^{r_{2}}\sigma_{\xi}^{r_{3}}({\bf c})={\bf c}\big\}\Big|\\
			&+ \sum_{r_{0}=0}^{1}\sum_{r_{1}=0}^{m-1}\sum_{r_{2}=0}^{n-1}\sum_{r_{3}=0}^{q-2}\Big|\big\{{\bf c}\in \mu_{-1}(\mathcal{R}_{n}\varepsilon_{t})\backslash \{{\bf 0}\}~\big|~\mu_{-1}^{r_{0}}\mu_{-q}^{r_{1}}\rho^{r_{2}}\sigma_{\xi}^{r_{3}}({\bf c})={\bf c}\big\}\Big|\Big) \\
			=&\frac{1}{mn(q-1)}\sum_{r_{1}=0}^{m-1}{\rm gcd}\big(n(q-1)(q^{{\rm gcd}(k,r_{1})}-1),n(q^{k}-1),i_{t}(q-1)(q^{k}-1)\big)\\
			=&\frac{1}{m}\sum_{r_{1}=0}^{m-1}{\rm gcd}\big(q^{{\rm gcd}(k,r_{1})}-1,\frac{q^{k}-1}{q-1},\frac{i_{t}(q^{k}-1)}{n}\big),
		\end{align*}
		and 
		\begin{align*}
			&\big|\langle\mu_{-1},\mu_{-q},\rho,\sigma_{\xi} \rangle\backslash \mathcal{C}^{\sharp}\big|\\
			=&\frac{1}{2mn(q-1)}\sum_{r_{0}=0}^{1}\sum_{r_{1}=0}^{m-1}\sum_{r_{2}=0}^{n-1}\sum_{r_{3}=0}^{q-2}\Big|\big\{{\bf c}\in \mathcal{C}^{\sharp}~\big|~\mu_{-1}^{r_{0}}\mu_{-q}^{r_{1}}\rho^{r_{2}}\sigma_{\xi}^{r_{3}}({\bf c})={\bf c}\big\}\Big|\\
			=&\frac{1}{2mn(q-1)}\left(\sum_{r_{1}=0}^{m-1}{\rm gcd}\big(n(q-1)(q^{{\rm gcd}(k,2r_{1})}-1),2n(q^{k}-1),i_{t}(q^{r_{1}}-1)(q-1)(q^{k}-1)\big)\right. \\
			&+\left. \sum_{r_{1}=0}^{m-1}{\rm gcd}\Big(n(q-1)(q^{{\rm gcd}(k,r_{1})}-1){\rm gcd}\big(q^{{\rm gcd}(k,r_{1})}-1,\frac{q^{k}-1}{q-1},\frac{i_{t}(q^{k}-1)}{n}\big),2i_{t}(q^{k}-1)^{2}\Big)\right)\\
			=&\frac{1}{2m}\sum_{r_{1}=0}^{m-1}\left({\rm gcd}\big(q^{{\rm gcd}(k,2r_{1})}-1,\frac{2(q^{k}-1)}{q-1},\frac{i_{t}(q^{r_{1}}-1)(q^{k}-1)}{n}\big)\right. \\
			&+\left. {\rm gcd}\Big((q^{{\rm gcd}(k,r_{1})}-1){\rm gcd}\big(q^{{\rm gcd}(k,r_{1})}-1,\frac{q^{k}-1}{q-1},\frac{i_{t}(q^{k}-1)}{n}\big),\frac{2i_{t}(q^{k}-1)^{2}}{n(q-1)}\Big)\right).
		\end{align*}
		
		Summarizing all the conclusions above, the number of orbits of $\langle\mu_{-1},\mu_{-q},\rho,\sigma_{\xi} \rangle$ on $\mathcal{C}^{*}$ is equal to
		\begin{align*}
			&\big|\langle\mu_{-1},\mu_{-q},\rho,\sigma_{\xi} \rangle\backslash \mathcal{C}'\big|+\big|\langle\mu_{-1},\mu_{-q},\rho,\sigma_{\xi} \rangle\backslash \mathcal{C}^{\sharp}\big|\\
			=&\frac{1}{2m}\sum_{r_{1}=0}^{m-1}\left(2{\rm gcd}\big(q^{{\rm gcd}(k,r_{1})}-1,\frac{q^{k}-1}{q-1},\frac{i_{t}(q^{k}-1)}{n}\big)\right. \\
			&+\left. {\rm gcd}\big(q^{{\rm gcd}(k,2r_{1})}-1,\frac{2(q^{k}-1)}{q-1},\frac{i_{t}(q^{r_{1}}-1)(q^{k}-1)}{n}\big)\right. \\
			&+\left. {\rm gcd}\Big((q^{{\rm gcd}(k,r_{1})}-1){\rm gcd}\big(q^{{\rm gcd}(k,r_{1})}-1,\frac{q^{k}-1}{q-1},\frac{i_{t}(q^{k}-1)}{n}\big),\frac{2i_{t}(q^{k}-1)^{2}}{n(q-1)}\Big)\right).
		\end{align*}
		The proof is then completed.
	\end{proof}
	
	\begin{Remark}{\rm
			Let $\mathcal{C}$ be the cyclic code in Theorem \ref{t3.5}. Since $\langle\mu_{q},\rho,\sigma_{\xi} \rangle$ is a subgroup of $\langle\mu_{-1},\mu_{-q},\rho,\sigma_{\xi} \rangle$, $\big|\langle\mu_{-1},\mu_{-q},\rho,\sigma_{\xi} \rangle\backslash \mathcal{C}^{\sharp}\big|\leq \big|\langle\mu_{q},\rho,\sigma_{\xi} \rangle\backslash \mathcal{C}^{\sharp}\big|$. Then by Corollary \ref{c3.4} and Theorem \ref{t3.5}, we have
			\begin{align*}
				&\big|\langle\mu_{q},\rho,\sigma_{\xi} \rangle\backslash \mathcal{C}^{*}\big|-\big|\langle\mu_{-1},\mu_{-q},\rho,\sigma_{\xi} \rangle\backslash \mathcal{C}^{*}\big|\\
				=&\big|\langle\mu_{q},\rho,\sigma_{\xi} \rangle\backslash \mathcal{C}^{'}\big|-\big|\langle\mu_{-1},\mu_{-q},\rho,\sigma_{\xi} \rangle\backslash \mathcal{C}^{'}\big|+\big|\langle\mu_{q},\rho,\sigma_{\xi} \rangle\backslash \mathcal{C}^{\sharp}\big|-\big|\langle\mu_{-1},\mu_{-q},\rho,\sigma_{\xi} \rangle\backslash \mathcal{C}^{\sharp}\big|\\
				\geq &\big|\langle \mu_{q},\rho,\sigma_{\xi} \rangle\backslash \mathcal{C}^{'}\big|-\big|\langle\mu_{-1},\mu_{-q},\rho,\sigma_{\xi} \rangle\backslash \mathcal{C}^{'}\big|\\
				=&\frac{1}{k}\sum_{r|k}\varphi(\frac{k}{r}){\rm gcd}\big(q^{r}-1,\frac{q^{k}-1}{q-1},\frac{i_{t}(q^{k}-1)}{n}\big),
			\end{align*}
			which illustrates that the upper bound on the number of non-zero weights of $\mathcal{C}$ given by Theorem \ref{t3.5} is less than that given by Corollary \ref{c3.4}.}
	\end{Remark}
	
	\begin{Example}{\rm
			Take $q=2$ and $n=21$. All the distinct $2$-cyclotomic cosets modulo $21$ are given by
			$$\Gamma_{0}=\{0\},~\Gamma_{1}=\{1,2,4,8,16,11\},~\Gamma_{2}=\{3,6,12\},$$
			$$\Gamma_{3}=\{5,10,20,19,17,13\},~\Gamma_{4}=\{7,14\},~\Gamma_{5}=\{9,18,15\}.$$
			Let $\ell$ be the number of non-zero weights of the cyclic code $\mathcal{C}=\mathcal{R}_{n}\varepsilon_{2}\bigoplus \mu_{-1}(\mathcal{R}_{n}\varepsilon_{2})=\mathcal{R}_{n}\varepsilon_{2} \bigoplus \mathcal{R}_{n}\varepsilon_{5}$. By Corollary \ref{c3.4}, we have
			\begin{align*}
				\ell&\leq \frac{1}{3}\sum_{r\mid 3}\varphi(\frac{3}{r})\left( {\rm gcd}\big(2^{r}-1,\frac{2^{3}-1}{2-1},\frac{3(2^{3}-1)}{21}\big)+ {\rm gcd}\big(2^{r}-1,\frac{2^{3}-1}{2-1},\frac{9(2^{3}-1)}{21}\big)\right.  \\
				&\quad +\left. {\rm gcd}\Big((2^{r}-1){\rm gcd}\big(2^{r}-1,\frac{2^{3}-1}{2-1},\frac{3(2^{3}-1)}{21},\frac{9(2^{3}-1)}{21}\big),\frac{(9-3)(2^{3}-1)^{2}}{21(2-1)}\Big) \right)\\
				&=\frac{1}{3}[\varphi(3)(1+1+1)+\varphi(1)(1+1+7)]\\
				&=\dfrac{1}{3}(6+9)=5.
			\end{align*}
			Using Theorem \ref{t3.5}, we have
			\begin{align*}
				\ell&\leq \frac{1}{12}\sum_{r=0}^{5}\left(2{\rm gcd}\big(2^{{\rm gcd}(3,r)}-1,\frac{2^{3}-1}{2-1},\frac{3(2^{3}-1)}{21}\big)\right. \\
				&\quad +\left. {\rm gcd}\big(2^{{\rm gcd}(3,2r)}-1,\frac{2(2^{3}-1)}{2-1},\frac{3(2^{r}-1)(2^{3}-1)}{21}\big)\right. \\
				&\quad +\left. {\rm gcd}\Big((2^{{\rm gcd}(3,r)}-1){\rm gcd}\big(2^{{\rm gcd}(3,r)}-1,\frac{2^{3}-1}{2-1},\frac{3(2^{3}-1)}{21}\big),\frac{6(2^{3}-1)^{2}}{21(2-1)}\Big)\right) \\
				&=\frac{1}{12}(16+4+4+16+4+4)=4.
			\end{align*}
			After using Magma \cite{4}, we know that the weight distribution of $\mathcal{C}$ is $1+21x^{6}+35x^{12}+7x^{18}$, which implies that the exact value of $\ell$ is 3. }
	\end{Example}
	
	\subsection{New upper bound on the number of non-zero weights of the cyclic code $\mathcal{C}_{l_{0}}=\mathcal{R}_{n}\varepsilon_{t}\bigoplus \mu_{(-1)^{l_{0}}p^{\frac{e}{2}}}(\mathcal{R}_{n}\varepsilon_{t})$}
	
	Recall that $q=p^{e}$, where $p$ is a prime and $e$ is a positive integer. In this subsection, we assume that $e$ is even.
	
	For $0<t\leq s$, suppose that the irreducible cyclic code $\mathcal{R}_{n}\varepsilon_{t}$ corresponds to the $q$-cyclotomic coset $\{i_{t},i_{t}q,\cdots,i_{t}q^{k-1}\}$, then
	$$\mathcal{R}_{n}\varepsilon_{t}=\Big\{\sum_{j=0}^{k-1}\big(c_{0}+c_{1}\zeta^{i_{t}q^{j}}+\cdots+c_{k-1}\zeta^{(k-1)i_{t}q^{j}}\big)e_{i_{t}q^{j}}~\Big|~c_{\ell}\in \mathbb{F}_{q}, 0\leq \ell \leq k-1\Big\}.$$
	Let $l_{0}$ be a integer with $0\leq l_{0}\leq 1$. One can check that $\mu_{(-1)^{l_{0}}p^{\frac{e}{2}}}(\mathcal{R}_{n}\varepsilon_{t})$ is also an irreducible cyclic code,
	and the primitive idempotent generating $\mu_{(-1)^{l_{0}}p^{\frac{e}{2}}}(\mathcal{R}_{n}\varepsilon_{t})$ corresponds to the $q$-cyclotomic coset $(-1)^{l_{0}}p^{-\frac{e}{2}}\{i_{t},i_{t}q,\cdots,i_{t}q^{k-1}\}$ (see \cite[Corollary 4.4.5]{11}). Therefore,
	$$\mu_{(-1)^{l_{0}}p^{\frac{e}{2}}}(\mathcal{R}_{n}\varepsilon_{t})=\Big\{\sum_{j=0}^{k-1}\big(\sum_{\ell=0}^{k-1}c_{\ell}'\zeta^{\ell(-1)^{l_{0}}p^{-\frac{e}{2}}i_{t}q^{j}}\big)e_{(-1)^{l_{0}}p^{-\frac{e}{2}}i_{t}q^{j}}~\Big|~c_{\ell}'\in \mathbb{F}_{q}, 0\leq \ell \leq k-1\Big\}.$$
	
	Suppose $(-1)^{l_{0}}p^{-\frac{e}{2}}i_{t}\notin \{i_{t}, i_{t}q,\cdots,i_{t}q^{k-1}\}$, then $\mu_{(-1)^{l_{0}}p^{\frac{e}{2}}}(\mathcal{R}_{n}\varepsilon_{t})\cap \mathcal{R}_{n}\varepsilon_{t}=\{{\bf 0}\}$ and $\mu_{(-1)^{l_{0}}p^{\frac{e}{2}}}^{2}(\mathcal{R}_{n}\varepsilon_{t})=\mathcal{R}_{n}\varepsilon_{t}$. Let
	$$\mathcal{C}_{l_{0}}=\mathcal{R}_{n}\varepsilon_{t}\bigoplus \mu_{(-1)^{l_{0}}p^{\frac{e}{2}}}(\mathcal{R}_{n}\varepsilon_{t}).$$
	It is easy to see that $\mu_{(-1)^{l_{0}}p^{\frac{e}{2}}}\in {\rm Aut}(\mathcal{C}_{l_{0}})$. Let $m_{_{l_{0}}}$ denote the order of $(-1)^{l_{0}}p^{\frac{e}{2}}$ in $\mathbb{Z}_{n}^{*}$, then the subgroup $\langle \mu_{(-1)^{l_{0}}p^{\frac{e}{2}}}\rangle$ of ${\rm Aut}(\mathcal{C}_{l_{0}})$ generated by $\mu_{(-1)^{l_{0}}p^{\frac{e}{2}}}$ is of order $m_{_{l_{0}}}$.  
	
	In this subsection, we consider the action of $\langle\mu_{(-1)^{l_{0}}p^{\frac{e}{2}}},\rho,\sigma_{\xi} \rangle$ on $\mathcal{C}_{l_{0}}^{*}=\mathcal{C}_{l_{0}}\backslash \{\bf 0\}$. Since $\mu_{q}=\mu_{(-1)^{l_{0}}p^{\frac{e}{2}}}^{2}$, $\langle\mu_{q},\rho,\sigma_{\xi} \rangle$ is a subgroup of $\langle\mu_{(-1)^{l_{0}}p^{\frac{e}{2}}},\rho,\sigma_{\xi} \rangle$. Hence the number of orbits of $\langle\mu_{(-1)^{l_{0}}p^{\frac{e}{2}}},\rho,\sigma_{\xi} \rangle$ on $\mathcal{C}_{l_{0}}^{*}$ is less than or equal to the number of orbits of $\langle\mu_{q},\rho,\sigma_{\xi} \rangle$ on $\mathcal{C}_{l_{0}}^{*}$.
	Further we will show that the former is strictly less than the latter.
	
	We first have the following lemma.
	\begin{Lemma}\label{L3}
		Let	$m$ be the order of $q$ in $\mathbb{Z}_{n}^{*}$ and let $m_{_{l_{0}}}$ be the order of $(-1)^{l_{0}}p^{\frac{e}{2}}$ in $\mathbb{Z}_{n}^{*}$. Suppose that $(-1)^{l_{0}}p^{-\frac{e}{2}}i_{t}\notin \{i_{t}, i_{t}q,\cdots,i_{t}q^{k-1}\}$, then $m_{_{l_{0}}}=2m$.
	\end{Lemma}
	\begin{proof}
		As $(-1)^{l_{0}}p^{-\frac{e}{2}}i_{t}\notin \{i_{t}, i_{t}q,\cdots,i_{t}q^{k-1}\}$, we have $(-1)^{l_{0}}p^{-\frac{e}{2}}q^{l}\not\equiv 1~({\rm mod}~n)$, or equivalently, $\big((-1)^{l_{0}}p^{\frac{e}{2}}\big)^{2l-1}\not\equiv 1~({\rm mod}~n)$, for any nonnegative integer $l$, and hence $m_{_{l_{0}}}$ is even. It is easy to see that $q^{\frac{m_{_{l_{0}}}}{2}}\equiv \big((-1)^{l_{0}}p^{\frac{e}{2}}\big)^{m_{_{l_{0}}}}\equiv 1~({\rm mod}~n)$ and $\big((-1)^{l_{0}}p^{\frac{e}{2}}\big)^{2m}\equiv q^{m}\equiv 1~({\rm mod}~n)$, and so $m|\frac{m_{_{l_{0}}}}{2}$ and $m_{_{l_{0}}}|2m$, which implies that $m_{_{l_{0}}}=2m$. 
	\end{proof}
	
	\begin{Theorem}\label{t3.6}
		Suppose that $$\mathcal{C}_{l_{0}}=\mathcal{R}_{n}\varepsilon_{t}\bigoplus\mu_{(-1)^{l_{0}}p^{\frac{e}{2}}}(\mathcal{R}_{n}\varepsilon_{t}),$$where $0<t\leq s$, the primitive idempotent $\varepsilon_{t}$ corresponds to the $q$-cyclotomic coset $\{i_{t}, i_{t}q,\cdots,i_{t}q^{k-1}\}$ and $(-1)^{l_{0}}p^{-\frac{e}{2}}i_{t}\notin \{i_{t}, i_{t}q,\cdots,i_{t}q^{k-1}\}$. Then the number of orbits of $\langle\mu_{(-1)^{l_{0}}p^{\frac{e}{2}}},\rho,\sigma_{\xi} \rangle$ on $\mathcal{C}_{l_{0}}^{*}=\mathcal{C}_{l_{0}}\backslash \{\bf 0\}$ is equal to
		\begin{align*}
			&\frac{1}{2m}\sum_{r=0}^{m-1}\left(2{\rm gcd}\big(q^{{\rm gcd}(k,r)}-1,\frac{q^{k}-1}{q-1},\frac{i_{t}(q^{k}-1)}{n}\big)\right. \\
			&+\left. {\rm gcd}\big(q^{{\rm gcd}(k,2r+1)}-1,\frac{2(q^{k}-1)}{q-1},\frac{((-1)^{l_{0}}p^{-\frac{e}{2}}+q^{r})i_{t}(q^{k}-1)}{n}\big)\right.\\
			&+\left. {\rm gcd}\Big((q^{{\rm gcd}(k,r)}-1){\rm gcd}\big(q^{{\rm gcd}(k,r)}-1,\frac{q^{k}-1}{q-1},\frac{i_{t}(q^{k}-1)}{n}\big), \frac{((-1)^{l_{0}}p^{-\frac{e}{2}}-1)i_{t}(q^{k}-1)^{2}}{n(q-1)}\Big)\right).
		\end{align*}
		In particular, the number of non-zero weights of $\mathcal{C}_{l_{0}}$ is less than or equal to the number of orbits of $\langle\mu_{(-1)^{l_{0}}p^{\frac{e}{2}}},\rho,\sigma_{\xi} \rangle$ on $\mathcal{C}_{l_{0}}^{*}$, with equality if and only if for any two codewords ${\bf c}_{1},{\bf c}_{2}\in \mathcal{C}_{l_{0}}^{*}$ with the same weight, there exist integers $j_{1}$, $j_{2}$ and $j_{3}$ such that $\mu_{(-1)^{l_{0}}p^{\frac{e}{2}}}^{j_{1}}\rho^{j_{2}}(\xi^{j_{3}}{\bf c}_{1})={\bf c}_{2}$.
	\end{Theorem}
	
	\begin{proof}
		Let $$\mathcal{C}_{l_{0}}'=\big(\mathcal{R}_{n}\varepsilon_{t}\backslash \{{\bf 0}\}\big)\cup \big(\mu_{(-1)^{l_{0}}p^{\frac{e}{2}}}(\mathcal{R}_{n}\varepsilon_{t})\backslash \{{\bf 0}\}\big)$$
		and
		$$\mathcal{C}_{l_{0}}^{\sharp}=\mathcal{R}_{n}\varepsilon_{t}\backslash \{{\bf 0}\}\bigoplus\mu_{(-1)^{l_{0}}p^{\frac{e}{2}}}(\mathcal{R}_{n}\varepsilon_{t})\backslash \{{\bf 0}\}.$$
		Note that
		$$\mathcal{C}_{l_{0}}\backslash \{{\bf 0}\}=\mathcal{C}_{l_{0}}'\cup \mathcal{C}_{l_{0}}^{\sharp}$$
		is a disjoint union, and that $\langle\mu_{(-1)^{l_{0}}p^{\frac{e}{2}}},\rho,\sigma_{\xi} \rangle$ can act on the sets $\mathcal{C}_{l_{0}}'$ and $\mathcal{C}_{l_{0}}^{\sharp}$, respectively. Then
		$$\big|\langle\mu_{(-1)^{l_{0}}p^{\frac{e}{2}}},\rho,\sigma_{\xi} \rangle\backslash \mathcal{C}_{l_{0}}^{*}\big|=\big|\langle\mu_{(-1)^{l_{0}}p^{\frac{e}{2}}},\rho,\sigma_{\xi} \rangle\backslash \mathcal{C}_{l_{0}}'\big|+\big|\langle\mu_{(-1)^{l_{0}}p^{\frac{e}{2}}},\rho,\sigma_{\xi} \rangle\backslash \mathcal{C}_{l_{0}}^{\sharp}\big|,$$
		and we know from Equation (\ref{e2.1}), Lemmas \ref{l2.2} and \ref{L3} that
		$$\big|\langle\mu_{(-1)^{l_{0}}p^{\frac{e}{2}}},\rho,\sigma_{\xi} \rangle\backslash \mathcal{C}_{l_{0}}'\big|=\frac{1}{2mn(q-1)}\sum_{r_{1}=0}^{2m-1}\sum_{r_{2}=0}^{n-1}\sum_{r_{3}=0}^{q-2}\Big|\big\{{\bf c}\in \mathcal{C}_{l_{0}}'~\big|~\mu_{(-1)^{l_{0}}p^{\frac{e}{2}}}^{r_{1}}\rho^{r_{2}}\sigma_{\xi}^{r_{3}}({\bf c})={\bf c}\big\}\Big|,$$
		and
		$$\big|\langle\mu_{(-1)^{l_{0}}p^{\frac{e}{2}}},\rho,\sigma_{\xi} \rangle\backslash \mathcal{C}_{l_{0}}^{\sharp}\big|=\frac{1}{2mn(q-1)}\sum_{r_{1}=0}^{2m-1}\sum_{r_{2}=0}^{n-1}\sum_{r_{3}=0}^{q-2}\Big|\big\{{\bf c}\in \mathcal{C}_{l_{0}}^{\sharp}~\big|~\mu_{(-1)^{l_{0}}p^{\frac{e}{2}}}^{r_{1}}\rho^{r_{2}}\sigma_{\xi}^{r_{3}}({\bf c})={\bf c}\big\}\Big|.$$
		
		Take ${\bf c}_{t}=\sum\limits_{j=0}^{k-1}\big(\sum\limits_{\ell=0}^{k-1}c_{\ell}\zeta^{\ell i_{t}q^{j}}\big)e_{i_{t}q^{j}}\in \mathcal{R}_{n}\varepsilon_{t}\backslash \{\bf 0\}$. Note that $e_{i_{t}q^{j}}=\frac{1}{n}\sum\limits_{l=0}^{n-1}\zeta^{-i_{t}q^{j}l}x^{l}$ and $\rho^{r_{2}}\sigma_{\xi}^{r_{3}}(e_{i_{t}q^{j}})=\xi^{r_{3}}\zeta^{i_{t}q^{j}r_{2}}e_{i_{t}q^{j}}$, and thus
		\begin{align*}
			\mu_{(-1)^{l_{0}}p^{\frac{e}{2}}}^{r_{1}}\rho^{r_{2}}\sigma_{\xi}^{r_{3}}(e_{i_{t}q^{j}})&=\xi^{r_{3}}\zeta^{i_{t}q^{j}r_{2}}\mu_{(-1)^{l_{0}}p^{\frac{e}{2}}}^{r_{1}}(e_{i_{t}q^{j}})\\
			&=\xi^{r_{3}}\zeta^{i_{t}q^{j}r_{2}}\cdot \frac{1}{n}\sum_{l=0}^{n-1}\zeta^{-i_{t}q^{j}l}x^{(-1)^{l_{0}r_{1}}p^{\frac{e}{2}r_{1}}l}\\
			&=\xi^{r_{3}}\zeta^{i_{t}q^{j}r_{2}}\cdot \frac{1}{n}\sum_{l=0}^{n-1}\zeta^{-(-1)^{l_{0}r_{1}}p^{-\frac{e}{2}r_{1}}i_{t}q^{j}\cdot (-1)^{l_{0}r_{1}}p^{\frac{e}{2}r_{1}}l}x^{(-1)^{l_{0}r_{1}}p^{\frac{e}{2}r_{1}}l}\\
			&=\xi^{r_{3}}\zeta^{i_{t}q^{j}r_{2}}\cdot \frac{1}{n}\sum_{l=0}^{n-1}\zeta^{-(-1)^{l_{0}r_{1}}p^{-\frac{e}{2}r_{1}}i_{t}q^{j}l}x^{l}\\
			&=\xi^{r_{3}}\zeta^{i_{t}q^{j}r_{2}}e_{(-1)^{l_{0}r_{1}}p^{-\frac{e}{2}r_{1}}i_{t}q^{j}}.
		\end{align*}
		We then have
		$$\mu_{(-1)^{l_{0}}p^{\frac{e}{2}}}^{r_{1}}\rho^{r_{2}}\sigma_{\xi}^{r_{3}}({\bf c}_{t})=\sum_{j=0}^{k-1}\xi^{r_{3}}\zeta^{i_{t}q^{j}r_{2}}\big(\sum\limits_{\ell=0}^{k-1}c_{\ell}\zeta^{\ell i_{t}q^{j}}\big)e_{(-1)^{l_{0}r_{1}}p^{-\frac{e}{2}r_{1}}i_{t}q^{j}}.$$
		Take ${\bf c}_{t'}=\sum\limits_{j=0}^{k-1}\big(\sum\limits_{\ell=0}^{k-1}c_{\ell}'\zeta^{\ell (-1)^{l_{0}}p^{-\frac{e}{2}}i_{t}q^{j}}\big)e_{(-1)^{l_{0}}p^{-\frac{e}{2}}i_{t}q^{j}}\in \mu_{(-1)^{l_{0}}p^{\frac{e}{2}}}(\mathcal{R}_{n}\varepsilon_{t})\backslash \{\bf 0\}$. Similarly as above, we have
		\begin{align*}
			\mu_{(-1)^{l_{0}}p^{\frac{e}{2}}}^{r_{1}}\rho^{r_{2}}\sigma_{\xi}^{r_{3}}({\bf c}_{t'})&=\sum_{j=0}^{k-1}\xi^{r_{3}}\zeta^{(-1)^{l_{0}}p^{-\frac{e}{2}}i_{t}q^{j}r_{2}}\big(\sum\limits_{\ell=0}^{k-1}c_{\ell}'\zeta^{\ell (-1)^{l_{0}}p^{-\frac{e}{2}} i_{t}q^{j}}\big)e_{(-1)^{l_{0}(r_{1}+1)}p^{-\frac{e}{2}(r_{1}+1)}i_{t}q^{j}}.
		\end{align*}
		
		It is easy to check that
		\begin{align*}
			(-1)^{l_{0}r_{1}}p^{-\frac{e}{2}r_{1}}=\begin{cases}(-1)^{l_{0}}p^{-\frac{e}{2}}q^{-\frac{r_{1}-1}{2}},& {\rm if~r_{1}~is~odd},\\
				q^{-\frac{r_{1}}{2}},& {\rm if~r_{1}~is~even}.
			\end{cases}
		\end{align*}
		Suppose $r_{1}$ is odd, that is, $r_{1}=2r+1$ with $0\leq r\leq m-1$, then
		$$\mu_{(-1)^{l_{0}}p^{\frac{e}{2}}}^{r_{1}}\rho^{r_{2}}\sigma_{\xi}^{r_{3}}({\bf c}_{t})=\sum_{j=0}^{k-1}\xi^{r_{3}}\zeta^{i_{t}q^{r+j}r_{2}}\big(\sum\limits_{\ell=0}^{k-1}c_{\ell}\zeta^{\ell i_{t}q^{j}}\big)^{q^{r}}e_{(-1)^{l_{0}}p^{-\frac{e}{2}}i_{t}q^{j}},$$
		and
		$$\mu_{(-1)^{l_{0}}p^{\frac{e}{2}}}^{r_{1}}\rho^{r_{2}}\sigma_{\xi}^{r_{3}}({\bf c}_{t'})=\sum_{j=0}^{k-1}\xi^{r_{3}}\zeta^{(-1)^{l_{0}}p^{-\frac{e}{2}}i_{t}q^{r+1+j}r_{2}}\big(\sum\limits_{\ell=0}^{k-1}c_{\ell}'\zeta^{\ell (-1)^{l_{0}}p^{-\frac{e}{2}} i_{t}q^{j}}\big)^{q^{r+1}}e_{i_{t}q^{j}}.$$
		Suppose $r_{1}$ is even, that is, $r_{1}=2r$ with $0\leq r\leq m-1$, then
		$$\mu_{(-1)^{l_{0}}p^{\frac{e}{2}}}^{r_{1}}\rho^{r_{2}}\sigma_{\xi}^{r_{3}}({\bf c}_{t})=\sum_{j=0}^{k-1}\xi^{r_{3}}\zeta^{i_{t}q^{r+j}r_{2}}\big(\sum\limits_{\ell=0}^{k-1}c_{\ell}\zeta^{\ell i_{t}q^{j}}\big)^{q^{r}}e_{i_{t}q^{j}},$$
		and
		$$\mu_{(-1)^{l_{0}}p^{\frac{e}{2}}}^{r_{1}}\rho^{r_{2}}\sigma_{\xi}^{r_{3}}({\bf c}_{t'})=\sum_{j=0}^{k-1}\xi^{r_{3}}\zeta^{(-1)^{l_{0}}p^{-\frac{e}{2}}i_{t}q^{r+j}r_{2}}\big(\sum\limits_{\ell=0}^{k-1}c_{\ell}'\zeta^{\ell (-1)^{l_{0}}p^{-\frac{e}{2}} i_{t}q^{j}}\big)^{q^{r}}e_{(-1)^{l_{0}}p^{-\frac{e}{2}}i_{t}q^{j}}.$$
		
		Similar discussion as in Lemmas \ref{l3.2} and \ref{l3.3} shows that 
		\begin{align*}
			&\big|\langle\mu_{(-1)^{l_{0}}p^{\frac{e}{2}}},\rho,\sigma_{\xi} \rangle\backslash \mathcal{C}_{l_{0}}'\big|\\
			=&\frac{1}{2mn(q-1)}\Big( \sum_{r_{1}=0}^{2m-1}\sum_{r_{2}=0}^{n-1}\sum_{r_{3}=0}^{q-2}\Big|\big\{{\bf c}\in \mathcal{R}_{n}\varepsilon_{t}\backslash \{{\bf 0}\}~\big|~\mu_{(-1)^{l_{0}}p^{\frac{e}{2}}}^{r_{1}}\rho^{r_{2}}\sigma_{\xi}^{r_{3}}({\bf c})={\bf c}\big\}\Big|\\
			&+\sum_{r_{1}=0}^{2m-1}\sum_{r_{2}=0}^{n-1}\sum_{r_{3}=0}^{q-2}\Big|\big\{{\bf c}\in \mu_{(-1)^{l_{0}}p^{\frac{e}{2}}}(\mathcal{R}_{n}\varepsilon_{t})\backslash \{{\bf 0}\}~\big|~\mu_{(-1)^{l_{0}}p^{\frac{e}{2}}}^{r_{1}}\rho^{r_{2}}\sigma_{\xi}^{r_{3}}({\bf c})={\bf c}\big\}\Big|\Big) \\
			=&\frac{1}{2mn(q-1)}\Big(\sum_{r=0}^{m-1}\sum_{r_{2}=0}^{n-1}\sum_{r_{3}=0}^{q-2}\Big|\big\{{\bf c}\in \mathcal{R}_{n}\varepsilon_{t}\backslash \{{\bf 0}\}~\big|~\mu_{(-1)^{l_{0}}p^{\frac{e}{2}}}^{2r}\rho^{r_{2}}\sigma_{\xi}^{r_{3}}({\bf c})={\bf c}\big\}\Big|\\
			&+\sum_{r=0}^{m-1}\sum_{r_{2}=0}^{n-1}\sum_{r_{3}=0}^{q-2}\Big|\big\{{\bf c}\in \mu_{(-1)^{l_{0}}p^{\frac{e}{2}}}(\mathcal{R}_{n}\varepsilon_{t})\backslash \{{\bf 0}\}~\big|~\mu_{(-1)^{l_{0}}p^{\frac{e}{2}}}^{2r}\rho^{r_{2}}\sigma_{\xi}^{r_{3}}({\bf c})={\bf c}\big\}\Big|\Big) \\
			=&\frac{1}{mn(q-1)}\sum_{r=0}^{m-1}{\rm gcd}\big(n(q-1)(q^{{\rm gcd}(k,r)}-1),n(q^{k}-1),i_{t}(q-1)(q^{k}-1)\big)\\
			=&\frac{1}{m}\sum_{r=0}^{m-1}{\rm gcd}\big(q^{{\rm gcd}(k,r)}-1,\frac{q^{k}-1}{q-1},\frac{i_{t}(q^{k}-1)}{n}\big),
		\end{align*}
		and 
		\begin{align*}
			&\big|\langle\mu_{(-1)^{l_{0}}p^{\frac{e}{2}}},\rho,\sigma_{\xi} \rangle\backslash \mathcal{C}_{l_{0}}^{\sharp}\big|\\
			=&\frac{1}{2mn(q-1)}\sum_{r_{1}=0}^{2m-1}\sum_{r_{2}=0}^{n-1}\sum_{r_{3}=0}^{q-2}\Big|\big\{{\bf c} \in \mathcal{C}_{l_{0}}^{\sharp}~\big|~\mu_{(-1)^{l_{0}}p^{\frac{e}{2}}}^{r_{1}}\rho^{r_{2}}\sigma_{\xi}^{r_{3}}({\bf c})={\bf c}\big\}\Big|\\
			=&\frac{1}{2mn(q-1)}\Big( \sum_{r=0}^{m-1}\sum_{r_{2}=0}^{n-1}\sum_{r_{3}=0}^{q-2}\Big|\big\{{\bf c}\in \mathcal{C}_{l_{0}}^{\sharp}~\big|~\mu_{(-1)^{l_{0}}p^{\frac{e}{2}}}^{2r+1}\rho^{r_{2}}\sigma_{\xi}^{r_{3}}({\bf c})={\bf c}\big\}\Big|\\
			&+\sum_{r=0}^{m-1}\sum_{r_{2}=0}^{n-1}\sum_{r_{3}=0}^{q-2}\Big|\big\{{\bf c}\in \mathcal{C}_{l_{0}}^{\sharp}~\big|~\mu_{(-1)^{l_{0}}p^{\frac{e}{2}}}^{2r}\rho^{r_{2}}\sigma_{\xi}^{r_{3}}({\bf c})={\bf c}\big\}\Big|\Big) \\
			=&\frac{1}{2mn(q-1)}\left(\sum_{r=0}^{m-1}{\rm gcd}\big(n(q-1)(q^{{\rm gcd}(k,2r+1)}-1),2n(q^{k}-1),\right. \\
			&\left. ((-1)^{l_{0}}p^{-\frac{e}{2}}+q^{r})i_{t}(q-1)(q^{k}-1)\big)+\sum_{r=0}^{m-1}{\rm gcd}\Big(n(q-1)(q^{{\rm gcd}(k,r)}-1){\rm gcd}\big(q^{{\rm gcd}(k,r)}-1,\right. \\
			&\left.  \frac{q^{k}-1}{q-1},\frac{i_{t}(q^{k}-1)}{n}\big), ((-1)^{l_{0}}p^{-\frac{e}{2}}-1)i_{t}(q^{k}-1)^{2}\Big)\right)\\
			=&\frac{1}{2m}\sum_{r=0}^{m-1}\left( {\rm gcd}\big(q^{{\rm gcd}(k,2r+1)}-1,\frac{2(q^{k}-1)}{q-1},\frac{((-1)^{l_{0}}p^{-\frac{e}{2}}+q^{r})i_{t}(q^{k}-1)}{n}\big)\right.\\
			&+\left. {\rm gcd}\Big((q^{{\rm gcd}(k,r)}-1){\rm gcd}\big(q^{{\rm gcd}(k,r)}-1,\frac{q^{k}-1}{q-1},\frac{i_{t}(q^{k}-1)}{n}\big), \frac{((-1)^{l_{0}}p^{-\frac{e}{2}}-1)i_{t}(q^{k}-1)^{2}}{n(q-1)}\Big)\right).
		\end{align*}
		
		Summarizing all the conclusions above,  the number of orbits of $\langle\mu_{(-1)^{l_{0}}p^{\frac{e}{2}}},\rho,\sigma_{\xi} \rangle$ on $\mathcal{C}_{l_{0}}^{*}=\mathcal{C}_{l_{0}}\backslash \{\bf 0\}$ is equal to
		\begin{align*}
			&\big|\langle\mu_{(-1)^{l_{0}}p^{\frac{e}{2}}},\rho,\sigma_{\xi} \rangle\backslash \mathcal{C}_{l_{0}}'\big|+\big|\langle\mu_{(-1)^{l_{0}}p^{\frac{e}{2}}},\rho,\sigma_{\xi} \rangle\backslash \mathcal{C}_{l_{0}}^{\sharp}\big|\\
			=&\frac{1}{2m}\sum_{r=0}^{m-1}\left(2{\rm gcd}\big(q^{{\rm gcd}(k,r)}-1,\frac{q^{k}-1}{q-1},\frac{i_{t}(q^{k}-1)}{n}\big)\right. \\
			&+\left. {\rm gcd}\big(q^{{\rm gcd}(k,2r+1)}-1,\frac{2(q^{k}-1)}{q-1},\frac{((-1)^{l_{0}}p^{-\frac{e}{2}}+q^{r})i_{t}(q^{k}-1)}{n}\big)\right.\\
			&+\left. {\rm gcd}\Big((q^{{\rm gcd}(k,r)}-1){\rm gcd}\big(q^{{\rm gcd}(k,r)}-1,\frac{q^{k}-1}{q-1},\frac{i_{t}(q^{k}-1)}{n}\big), \frac{((-1)^{l_{0}}p^{-\frac{e}{2}}-1)i_{t}(q^{k}-1)^{2}}{n(q-1)}\Big)\right).
		\end{align*}
		The proof is then completed.
	\end{proof}
	
	\begin{Remark}{\rm
			Let $\mathcal{C}_{l_{0}}$ be the cyclic code in Theorem \ref{t3.6}. Since $\langle\mu_{q},\rho,\sigma_{\xi} \rangle$ is a subgroup of $\langle\mu_{(-1)^{l_{0}}p^{\frac{e}{2}}},\rho,\sigma_{\xi} \rangle$, we have 
			$$\big|\langle\mu_{(-1)^{l_{0}}p^{\frac{e}{2}}},\rho,\sigma_{\xi} \rangle\backslash \mathcal{C}_{l_{0}}^{\sharp}\big|\leq \big|\langle\mu_{q},\rho,\sigma_{\xi} \rangle\backslash \mathcal{C}_{l_{0}}^{\sharp}\big|.$$ Then we see from Corollary \ref{c3.4} and Theorem \ref{t3.6} that
			\begin{align*}
				&\big|\langle\mu_{q},\rho,\sigma_{\xi} \rangle\backslash \mathcal{C}_{l_{0}}^{*}\big|-\big|\langle\mu_{(-1)^{l_{0}}p^{\frac{e}{2}}},\rho,\sigma_{\xi} \rangle\backslash \mathcal{C}_{l_{0}}^{*}\big|\\
				=&\big|\langle\mu_{q},\rho,\sigma_{\xi} \rangle\backslash \mathcal{C}_{l_{0}}^{'}\big|-\big|\langle\mu_{(-1)^{l_{0}}p^{\frac{e}{2}}},\rho,\sigma_{\xi} \rangle\backslash \mathcal{C}_{l_{0}}^{'}\big|+\big|\langle\mu_{q},\rho,\sigma_{\xi} \rangle\backslash \mathcal{C}_{l_{0}}^{\sharp}\big|-\big|\langle\mu_{(-1)^{l_{0}}p^{\frac{e}{2}}},\rho,\sigma_{\xi} \rangle\backslash \mathcal{C}_{l_{0}}^{\sharp}\big|\\
				\geq &\big|\langle\mu_{q},\rho,\sigma_{\xi} \rangle\backslash \mathcal{C}_{l_{0}}^{'}\big|-\big|\langle\mu_{(-1)^{l_{0}}p^{\frac{e}{2}}},\rho,\sigma_{\xi} \rangle\backslash \mathcal{C}_{l_{0}}^{'}\big|\\
				=&\frac{1}{k}\sum_{r|k}\varphi(\frac{k}{r}){\rm gcd}\big(q^{r}-1,\frac{q^{k}-1}{q-1},\frac{i_{t}(q^{k}-1)}{n}\big).
			\end{align*}
			Therefore, the upper bound on the number of non-zero weights of $\mathcal{C}_{l_{0}}$ given by Theorem \ref{t3.6} is less than that given by Corollary \ref{c3.4}.}
	\end{Remark}
	
	\begin{Example}{\rm
			Take $q=4$ and $n=15$. All the distinct $4$-cyclotomic cosets modulo $15$ are given by
			$$\Gamma_{0}=\{0\},~\Gamma_{1}=\{1,4\},~\Gamma_{2}=\{2,8\},~\Gamma_{3}=\{3,12\},$$
			$$\Gamma_{4}=\{5\},~\Gamma_{5}=\{6,9\},~\Gamma_{6}=\{7,13\},~\Gamma_{7}=\{10\},~\Gamma_{8}=\{11,14\}.$$
			We first consider the the cyclic code $\mathcal{C}_{0}=\mathcal{R}_{n}\varepsilon_{1}\bigoplus \mu_{2}(\mathcal{R}_{n}\varepsilon_{1})=\mathcal{R}_{n}\varepsilon_{1}\bigoplus \mathcal{R}_{n}\varepsilon_{2}$. Let $\ell_{0}$ be the number of non-zero weights of $\mathcal{C}_{0}$. By Corollary \ref{c3.4}, we have
			\begin{align*}
				\ell_{0}&\leq \frac{1}{2}\sum_{r\mid 2}\varphi(\frac{2}{r})\left( {\rm gcd}\big(4^{r}-1,\frac{4^{2}-1}{4-1},\frac{4^{2}-1}{15}\big)+ {\rm gcd}\big(4^{r}-1,\frac{4^{2}-1}{4-1},\frac{2(4^{2}-1)}{15}\big)\right.  \\
				&\quad +\left. {\rm gcd}\Big((4^{r}-1){\rm gcd}\big(4^{r}-1,\frac{4^{2}-1}{4-1},\frac{4^{2}-1}{15},\frac{2(4^{2}-1)}{15}\big),\frac{(2-1)(4^{2}-1)^{2}}{15(4-1)}\Big) \right) \\
				&=\frac{1}{2}[\varphi(2)(1+1+1)+\varphi(1)(1+1+5)]\\
				&=\frac{1}{2}(3+7)=5.
			\end{align*}
			Using Theorem \ref{t3.6}, we have
			\begin{align*}
				\ell_{0}&\leq \frac{1}{4}\sum_{r=0}^{1}\left(2{\rm gcd}\big(4^{{\rm gcd}(2,r)}-1,\frac{4^{2}-1}{4-1},\frac{4^{2}-1}{15}\big)\right. \\
				&\quad +\left. {\rm gcd}\big(4^{{\rm gcd}(2,2r+1)}-1,\frac{2(4^{2}-1)}{4-1},\frac{(8+4^{r})(4^{2}-1)}{15}\big)\right.\\
				&\quad +\left. {\rm gcd}\Big((4^{{\rm gcd}(2,r)}-1){\rm gcd}\big(4^{{\rm gcd}(2,r)}-1,\frac{4^{2}-1}{4-1},\frac{4^{2}-1}{15}\big), \frac{(8-1)(4^{2}-1)^{2}}{15(4-1)}\Big)\right)\\
				&=\frac{1}{4}(8+4)=3.
			\end{align*}
			After using Magma \cite{4}, we know that the weight distribution of $\mathcal{C}_{0}$ is $1+45x^{8}+210x^{12}$, which implies that $\ell_{0}=2$. 
			
			In addition, for the cyclic code $\mathcal{C}_{1}=\mathcal{R}_{n}\varepsilon_{1}\bigoplus \mu_{-2}(\mathcal{R}_{n}\varepsilon_{1})=\mathcal{R}_{n}\varepsilon_{1}\bigoplus \mathcal{R}_{n}\varepsilon_{6}$, let $\ell_{1}$ be the number of non-zero weights of  $\mathcal{C}_{1}$. One can verify that by Corollary \ref{c3.4} we have $\ell_{1}\leq 11$, and by Theorem \ref{t3.6} we have $\ell_{1}\leq 6$. After using Magma \cite{4}, we see that the weight distribution of $\mathcal{C}_{1}$ is $1+30x^{6}+60x^{9}+105x^{12}+60x^{15}$, which implies that $\ell_{1}=4$. }
	\end{Example}
	
	\section{Concluding remarks and future works}
	In this paper, we improve the upper bounds in \cite{7} on the number of non-zero weights of any simple-root cyclic code by replacing $\langle\rho,\sigma_{\xi} \rangle$ with larger subgroups of the automorphism group of the code. We first replace $\langle\rho,\sigma_{\xi} \rangle$ with $\langle\mu_{q},\rho,\sigma_{\xi} \rangle$ and obtain improved upper bounds which in some cases are strictly less than the upper bounds in \cite{7} (see subsections 3.1 and 3.2). In addition, for two special classes of cyclic codes, we replace $\langle\mu_{q},\rho,\sigma_{\xi} \rangle$ with larger subgroups of the automorphism groups of these codes, and then we obtain smaller upper bounds on the number of non-zero weights of such codes (see subsections 3.3 and 3.4). 
	
	A possible direction for future work is to find new few-weight cyclic codes using the main results of this paper. It would also be interesting to find tight upper bounds on the number of symbol-pair weights of cyclic codes.

	\noindent\textbf{Acknowledgement.}
	
	This work was supported by National Natural Science Foundation of China
	under Grant Nos.12271199 and 12171191 and The Fundamental Research Funds for the Central Universities
	30106220482.

\end{document}